\documentclass[12pt,twoside]{article} 
\usepackage[T1]{fontenc} 

\usepackage{amsthm, amsmath, amssymb, mathrsfs, upgreek, enumerate, mathtools, comment, natbib, graphicx, mathabx, booktabs, threeparttable, tikz, relsize, exscale, epstopdf, scrextend, multirow, xr-hyper, pdfpages, bm, caption, subcaption}
\usepackage[headheight=110pt,margin=1.25in]{geometry}
\usepackage[algoruled]{algorithm2e}
\usepackage[section]{placeins} 

\usepackage{fancyhdr, setspace} 
\usepackage[pagebackref=false]{hyperref} 
\hypersetup{
    colorlinks=true,
    citecolor=blue,
    filecolor=black,
    linkcolor=black,
    urlcolor=blue,
    bookmarksopen=true,
    pdfstartview=FitH
}

\newtheorem{theorem}{Theorem}
\newtheorem{assump}{Assumption}

\newtheorem{cor}{Corollary}

\newtheorem{proposition}{Proposition}

\newtheorem{sectheorem}{Theorem}[section]

\newtheorem{seclemma}{Lemma}[section]

\newtheorem{secprop}{Proposition}[section]

\theoremstyle{definition}
\newtheorem{definition}{Definition}

\newtheorem{example}{Example}

\newtheorem{remark}{Remark}

\newtheorem{secremark}{Remark}[section]

\def\Snospace~{\S{}}

\makeatletter
\def\thm@space@setup{
  \thm@preskip=15pt \thm@postskip=15pt 
}
\makeatother

\def\indep{\perp\!\!\!\perp}

\newcommand{\cov}{\text{Cov}}
\newcommand{\var}{\text{Var}}
\newcommand{\E}{{\bf E}}
\newcommand{\R}{\mathbb{R}}
\newcommand{\RR}{\bm{R}}

\newcommand{\prob}{{\bf P}}
\newcommand{\plimarrow}{\stackrel{p}\longrightarrow}
\newcommand{\dlimarrow}{\stackrel{d}\longrightarrow}

\newcommand{\N}{\mathcal{N}}
\newcommand{\Zeta}{\bm{\zeta}}
\newcommand{\zero}{\bm{0}}
\newcommand{\ind}{\bm{1}}

\newcommand*{\medcap}{\mathbin{\scalebox{1.5}{\ensuremath{\cap}}}}
\newcommand*{\medcup}{\mathbin{\scalebox{1.5}{\ensuremath{\cup}}}}

\let\emptyset\varnothing
\providecommand{\abs}[1]{\lvert#1\rvert} 
\providecommand{\norm}[1]{\lVert#1\rVert}

\renewcommand{\qed}{\hfill \mbox{\raggedright \rule{0.08in}{0.08in}}} 
\renewenvironment{proof}[1][\proofname]{{\noindent\sc#1. }}{\qed\vspace{15pt}} 

\title{\bf\sc Normal Approximation in Large Network Models\thanks{We thank the editor and referees for valuable comments and suggestions that helped improve the exposition of the paper. We gratefully acknowledge financial support from the National Science Foundation under grants SES-1755100 and SES-1625586.}}

\author{Michael P.\ Leung\thanks{Department of Economics, UC Santa Cruz. E-mail: leungm@ucsc.edu.} \and Hyungsik Roger Moon\thanks{Department of Economics, University of Southern California and Yonsei University. E-mail: moonr@usc.edu.}}

\begin{document}
\maketitle
\onehalfspacing
 
\begin{abstract}

  {\sc Abstract.} We prove a central limit theorem for network formation models with strategic interactions and homophilous agents. Since data often consists of observations on a single large network, we consider an asymptotic framework in which the network size diverges. We argue that a modification of ``stabilization'' conditions from the literature on geometric graphs provides a useful high-level formulation of weak dependence which we utilize to establish an abstract central limit theorem. Using results in branching process theory, we derive interpretable primitive conditions for stabilization. The main conditions restrict the strength of strategic interactions and equilibrium selection mechanism. We discuss practical inference procedures justified by our results.

  \bigskip 


 
\end{abstract}

\addcontentsline{toc}{part}{Main Paper}
\newpage

\section{Introduction}\label{sintro}

Network models have attracted considerable attention as tractable representations of non-market interactions, such as peer effects and social learning, and of economic relationships, such as financial and trade networks. The economic perspective on networks emphasizes the importance of strategic interactions or externalities \citep{jackson2016networks}. One strand of the literature studies social interactions, how an individual's behavior interacts with those of her social contacts. A second strand studies network formation, why an individual chooses particular social contacts, and how their choices affect those of others. In both cases, externalities generate a wedge between choices that are optimal from the individual's perspective and those that are efficient for society. This has important consequences for policymaking, for instance motivating policies of ``associational redistribution'' that intervene on the network structure \citep{durlauf1996associational}. However, when preferences are misaligned with the policy objective, the endogenous response to a policy intervention may diverge from the intended outcome \citep{carrell_natural_2013}. It is therefore of interest to develop econometric methodologies for measuring network externalities.

The focus of our paper is models of network formation, which have diverse applications ranging from risk sharing in the rural Philippines \citep{fg2007} to research partnerships in the biotechnology industry \citep{pwko2005}. They can be used to distinguish between different theoretical mechanisms for link formation, including preferential attachment \citep{barabasi1999emergence}, strategic transitivity \citep{mele2022structural,ridder2022two}, and homophily \citep{mpslc2001}. They enable policymakers to forecast the effects of counterfactual interventions on a social network \citep{mele2017segregation} and to account for dependence between the network and unobservables of a social interactions model \citep{badev2021nash}. 

We augment a class of latent space \citep{hrh2002} and geometric graph \citep{py2003} models with strategic interactions. In such models, nodes are positioned in a possibly latent social or economic space such that the rate of link formation decays with distance in space. However, due to strategic interactions, link formation depends not only on the attributes of the node pair involved in the link but also on links between other pairs. This induces a non-standard form of cross-sectional dependence between network units. The challenge for large-sample theory lies in establishing conditions under which the ``amount of independent information'' increases with the network size.

We establish a CLT for a large class of network moments computed from a single large network. The moments are averages of node-level statistics $n^{-1} \sum_{i=1}^n \psi_i$, where $\psi_i$ is some function of the network and node types. Examples include the {\em degree} of node $i$, which is the number of links involving $i$; the average clustering coefficient; and subnetwork counts which have been utilized for inference in econometric models of network formation \citep{sheng2014}. We discuss practical inference procedures justified by the result. 

We emphasize two technical contributions. The first is an abstract CLT that holds under high-level conditions, which is an extension of limit theorems from the literature on geometric graphs \citep[e.g.][]{penrose2007gaussian}. The key condition is a modification of ``stabilization'' assumptions from that literature that provide a useful formulation of weak dependence for our purposes. Stabilization essentially requires $\psi_i$ to only be a function of a random subset of nodes whose size has a distribution with exponential tails. This implies that the size is asymptotically bounded, so node $i$'s statistic only depends on a small fraction of alters. In this sense, $\{\psi_i\}_{i=1}^n$ is weakly dependent.

In our applications, this random subset involves a union of network components, which are disconnected subnetworks that are challenging combinatorial objects to analyze (a formal definition is provided at the end of this section). Our primary contribution is to develop a methodology for deriving the tail bounds required for stabilization. We employ techniques from branching process theory commonly used to bound component sizes in the literature on random graphs. Using the methodology, we derive interpretable primitive sufficient conditions for stabilization in models of strategic network formation. 

One of the main conditions restricts the strength of strategic interactions. These induce cross-sectional dependence since the realization of a link can depend on the existence of ``neighboring'' links, which in turn can depend on other links, and so on. The longer the lengths of these chains of dependent links, the stronger the degree of cross-sectional dependence. We adopt a well-known technique in random graph theory to bound the lengths of these chains by branching processes whose sizes can more tractably be proven to have exponential tails, provided the processes are ``subcritical'' or non-explosive in growth. We argue that conditions for subcriticality in our context are analogous to weak dependence conditions imposed on spatial and temporal autoregressive models that bound the magnitude of the autoregressive parameter below one.

The other main condition is a restriction on the equilibrium selection mechanism. Even if strategic interactions are sufficiently weak, strong cross-sectional dependence may exist if nodes can globally ``coordinate'' on the equilibrium network through a common signal such as the type of a particular node. We instead require equilibrium selection to be sufficiently ``decentralized,'' which holds for instance under myopic best-response dynamics. This is the single-network analog of the usual requirement under many-network asymptotics that equilibrium selection is independent across networks.

A growing literature studies frequentist inference in network formation models when the econometrician observes a single network. \cite{leung2016} and \cite{menzel2016} develop laws of large numbers for models of strategic network formation. The former paper modifies a weaker stabilization condition due to \cite{py2003} and uses branching processes to derive primitive sufficient conditions. Our paper tackles the more difficult problem of obtaining a normal approximation, which naturally requires a stronger stabilization condition, and discusses practical inference procedures newly justified by the result. We prove the CLT by adapting results from the literature on geometric graphs \citep{penrose2005normal,penrose2007gaussian}; a more detailed discussion of our contributions relative to this literature can be found in \autoref{srelit}. \cite{kuersteiner2019limit} takes a different approach, using a novel conditional mixingale type assumption defined in terms of a random metric of distance.

\cite{leung2015a} and \cite{ridder2022two} study strategic network formation under incomplete information. In this setting, links are independent conditional on observables, whereas the models we study can be microfounded as games of complete information. These allow for unobserved heterogeneity, which generates dependence between potential links even conditional on observables. 

\cite{charbonneau2017multiple}, \cite{dzemski2019}, \cite{graham2017}, and \cite{jochmans2018semiparametric} consider dyadic link formation models without strategic interactions but allow for node-level fixed effects. A large literature in statistics studies models without strategic interactions, for example stochastic block models \citep{bickel2011method} and latent-space models \citep{hrh2002}. These are useful for their parsimony and tasks such as community detection. \cite{cj2015} and \cite{bm2015} study statistical models allowing for interdependence between links.

\cite{kojevnikov2021limit} prove a CLT for node-level data conditional on the network. This does not apply to network formation models since the network is the outcome. We prove an unconditional CLT that may be applied to network formation as well as network processes \citep{leung2019inference}.

The next section presents a model of strategic network formation and defines network moments. In \autoref{smethpf}, we state high-level conditions for a CLT and outline its proof. Readers interested in the low-level CLT for strategic network formation may skip \autoref{smethpf} for \autoref{scltstatic}, which presents primitive sufficient conditions. We also outline a general methodology for verifying the high-level condition that can be applied to other network models. In \autoref{sinfer}, we discuss practical inference procedures. We present results from a simulation study in \autoref{smc}, and \autoref{sconclude} concludes. Proofs can be found in the supplemental appendix.

We introduce standard notation and terminology for networks. We represent a network on a set of $n$ nodes by an $n\times n$ adjacency matrix, where the $ij$th entry $A_{ij}$, termed the {\em potential link}, is an indicator for whether nodes $i,j$ are connected. We assume $A_{ii} = 0$ for all nodes $i$, meaning that there are no self-links, and we focus on undirected networks, so $A_{ij} = A_{ji}$. For two networks $\bm{A},\bm{A}'$, we say that $\bm{A}$ is a {\em subnetwork} of $\bm{A}'$ if every link in $\bm{A}$ is a link in $\bm{A}'$. A {\em path} in a network from node $i$ to $j$ is a sequence of distinct nodes starting with $i$ and ending with $j$ such that for each consecutive node pair $k,k'$ in this sequence, $A_{kk'}=1$. Its {\em length} is the number of links it involves. The {\em path distance} between two nodes $i\neq j$ in $\bm{A}$ is the length of the shortest path that connects them if a path exists and $\infty$ otherwise. The path distance between a node and itself is defined as zero. The {\em $K$-neighborhood} of a node $i$ in $\bm{A}$, denoted by $\N_{\bm{A}}(i,K)$, is the set of nodes of path distance at most $K$ from $i$. Finally the {\em component} of a node $i$ in a network $\bm{A}$ is the set of nodes at finite path distance from $i$.

\section{Model}\label{smodel}

Let $\N_n = \{1,\dots,n\}$ be a set of nodes, and endow each $i\in\N_n$ with an i.i.d.\ vector-valued {\em type} $(X_i,Z_i) \in \R^d \times \R^{d_z}$. We distinguish $X_i$ as the {\em position} of node $i$, a continuously distributed vector of {\em homophilous} (defined below) attributes with density $f$ that has bounded support. Endow each node pair $(i,j)$ with an i.i.d.\ $\mathbb{R}$-valued random utility shock $\zeta_{ij} = \zeta_{ji}$ independent of types. 

The network $\bm{A}$ satisfies, for all $i\neq j$,
\begin{equation}
  A_{ij} = \mathbf{1}\left\{V_{ij} > 0\right\} \quad\text{for}\quad V_{ij} \equiv V( r_n^{-1}\norm{X_i-X_j}, S_{ij}, Z_i, Z_j, \zeta_{ij}), \label{model}
\end{equation}

\noindent where the {\em joint-surplus function} $V(\cdot)$ is an $\R$-valued function, $\norm{\cdot}$ a norm on $\R^d$, and $S_{ij}$ a vector of statistics that captures strategic interactions through its dependence on $\bm{A}$. As discussed below, the sparsity of $\bm{A}$ will in part be determined by the positive constant $r_n$. Model \eqref{model} corresponds to the well-known {\em pairwise-stability} solution concept under transferable utility \citep{jackson2010}. 

\begin{example}\label{elinear}
  Consider the linear joint surplus function
  \begin{equation}
    V_{ij} = \theta_1 + \theta_2 S_{ij} - \theta_3 (r_n^{-1}\norm{X_i-X_j})^2 + \zeta_{ij} \quad\text{with}\quad S_{ij} = \max_k A_{ik}A_{jk}. \label{modelS}
  \end{equation}

  \noindent Our theory will require $\theta_3>0$ which captures homophily in position since it disincentivizes link formation between positionally dissimilar nodes. The term $S_{ij}$ is an indicator for whether $i$ and $j$ share a common neighbor \citep[used for example by][]{menzel2016}. If $\theta_2>0$, then node pairs $(i,j)$ sharing a common neighbor are more likely to form a link. Both rationalize the well-known stylized fact that networks are commonly {\em clustered} in that nodes with common neighbors are themselves typically neighbors \citep{jackson2010}.
\end{example}

\begin{example}\label{esheng}
  \cite{sheng2014} studies a specification similar to
  \begin{multline}
    V_{ij} = \beta_0 + (Z_i+Z_j)'\beta_1 + \beta_2 r_n^{-1}\norm{X_i-X_j} \\ + \gamma_1 \sum_{k=1}^n (A_{ik} + A_{jk}) + \gamma_2 \sum_{k=1}^n A_{ik}A_{jk} + \zeta_{ij}, \label{shengU}
  \end{multline}

  \noindent which corresponds to linear joint surplus function with
  \begin{equation}
    S_{ij} = \bigg( \sum_{k=1}^n A_{ik}A_{jk}, \sum_{k=1}^n A_{ik}, \sum_{k=1}^n A_{jk} \bigg). \label{modelS2}
  \end{equation}

  \noindent The first component of \eqref{modelS2} plays a role analogous to $S_{ij}$ in \eqref{modelS}, being a count of the number of common neighbors. The second and third components are respectively the {\em degrees} (number of neighbors) of $i$ and $j$. If $\gamma_1>0$, this captures a popularity effect. 
\end{example}

As explained in the next subsection, the assumptions we impose for a CLT in many cases require $S_{ij}$ to be uniformly bounded over $i,j\in\N_n$ and $n\in\mathbb{N}$. Existing work on large-network asymptotics for network formation shares this limitation \citep{bm2015,menzel2016,ridder2022two}. Statistics such as \eqref{modelS2} may be modified to satisfy uniform boundedness by truncation, for example $\min\{\sum_{k=1}^n A_{ik}, \Delta\}$ for some user-specified constant $\Delta$.

The next subsection states a restriction on $V(\cdot)$ that generates homophily and sparsity. In \autoref{eselection}, we formalize how $S_{ij}$ may depend on $\bm{A}$ and introduce the equilibrium selection mechanism. Finally, in \autoref{smoms}, we define the class of network moments for which we provide a CLT.
  
\subsection{Homophily and Sparsity}\label{rsparse}

A common feature of social networks is {\em homophily}, the tendency for those with similar characteristics to associate. We require homophily in position, specifically that the joint-surplus function $V(\cdot)$ is decreasing in the first component, so nodes dissimilar in position are less likely to form links. Positions may represent latent node locations in an abstract ``social space,'' as in latent-space models \citep{hrh2002}, or attributes such as income and geographic location. 

Another common feature is {\em sparsity}, meaning that the number of connections formed by the typical node is of significantly smaller order than $n$ \citep{barabasi2015}. This is often accomplished by scaling the sequence of models such that the expected degree $n^{-1} \sum_{i=1}^n \sum_{j=1}^n \E[A_{ij}]$ is asymptotically bounded.

Our first assumption restricts the joint-surplus function $V(\cdot)$ and the distribution of $\zeta_{ij}$, so that the network generated by model \eqref{model} exhibits homophily in positions and sparsity. Let $\bar{\Phi}_\zeta(z) = \prob(\zeta_{ij} \geq z)$. Define $\bar{V}(r_n^{-1}\norm{X_i-X_j}, \zeta_{ij}) = \sup_{s,z,z'} V(r_n^{-1}\norm{X_i-X_j}, s, z, z', \zeta_{ij})$, where the supremum is taken over $(s,z,z')$ in the support of $(S_{ij},Z_i,Z_j)$ (which will be identical across $i,j$ under assumptions below). Finally, recall that $d$ is the dimension of $X_1$.

\begin{assump}[Homophily and Sparsity]\label{Vc} 
  For any $\delta\in\R_+$, $\bar{V}(\delta,\cdot)$ is invertible and increasing in its second argument, and its inverse $\bar{V}^{-1}(\delta,\cdot)$ satisfies $\limsup_{\delta\rightarrow\infty} \delta^{-1} \log \bar{\Phi}_\zeta(\bar{V}^{-1}(\delta, 0)) < 0$. Furthermore, there exists $\kappa>0$ such that, for any $n\in\mathbb{N}$,
  \begin{equation}
    r_n \equiv (\kappa/n)^{1/d}. \label{r_n}
  \end{equation}
\end{assump}

\noindent Homophily is a consequence of the first sentence because $\prob(A_{ij}=1 \mid r_n^{-1}\norm{X_i-X_j}=\delta) \leq \bar{\Phi}_\zeta(\bar{V}^{-1}(\delta, 0))$, which is required to decrease exponentially with $\delta$. If $V(\cdot)$ is not too nonlinear in its first component and the distribution of $\zeta_{ij}$ has exponential tails, then this assumption is satisfied, as shown in the next example.

\begin{example}\label{ehphil}
  Consider \autoref{elinear}, and suppose $S_{ij}$ has uniformly bounded support. Then for some universal constant $M$,
  \begin{equation*}
    \prob(V_{ij} > 0 \mid r_n^{-1}\norm{X_i-X_j}=\delta) \leq \prob(\underbrace{M - \theta_3\delta^2 + \zeta_{ij}}_{\bar{V}(\delta,\zeta_{ij})} > 0) = \bar{\Phi}_\zeta(\underbrace{\theta_3\delta^2 - M}_{\bar{V}^{-1}(\delta,0)} ),
  \end{equation*}

  \noindent which decays to zero exponentially with $\delta$ if $\theta_3>0$ and the distribution of $\zeta_{ij}$ has exponential tails.
\end{example}

\begin{example}
  In the previous example, \autoref{Vc} implicitly imposes restrictions on the support of $(S_{ij},Z_i,Z_j,\zeta_{ij})$. Consider a variant of \eqref{modelS} in which we replace $\theta_3 (r_n^{-1}\norm{X_i-X_j})^2$ with the ``random geometric graph'' penalty $\infty \cdot \ind\{r_n^{-1}\norm{X_i-X_j} > c\}$ for some $c>0$, with the convention that $\infty \cdot 0 = 0$. That is, nodes do not link with alters whose scaled positions are sufficiently far from the ego's. Then $\prob(V_{ij} > 0 \mid r_n^{-1}\norm{X_i-X_j}=\delta) \leq \ind\{\delta \leq c\}$ which satisfies \autoref{Vc} without support restrictions.
\end{example}

Lastly we discuss the role of \eqref{r_n} for sparsity. As $n$ increases, there are more opportunities to form links, which increases expected degree, corresponding to a denser network. By sending $r_n$ to zero with $n$, we increase the ``cost'' of link formation due to homophily, thereby decreasing the expected degree. Our choice of $r_n$ balances these two forces to achieve a sparse network in which the expected degree is asymptotically bounded. To see this, notice that expected degree is
\begin{align}
  \frac{1}{n} \sum_{i=1}^n &\sum_{j=1}^n \E[A_{ij}] = (n-1) \prob\big(V(r_n^{-1}\norm{X_i-X_j}, S_{ij}, Z_i, Z_j, \zeta_{ij}) > 0\big) \nonumber\\
			   &\leq (n-1) \prob\big(\zeta_{ij} > \bar{V}^{-1}(r_n^{-1}\norm{X_i-X_j},0)\big) \nonumber\\
			   &= (n-1)r_n^d \int_{\R^d} \int_{\R^d} \prob\big(\zeta_{ij} > \bar{V}^{-1}(\norm{x-x'},0)\big) f(x) f(x+r_n(x'-x)) \,\text{d}x \,\text{d}x' \label{sparsecalc}
\end{align}

\noindent by a change of variables $x' \mapsto x+r_n(x'-x)$. By \eqref{r_n}, if $f$ is continuous, this converges to
\begin{equation*}
  \kappa \int_{\R^d} \int_{\R^d} \bar{\Phi}_\zeta\big(\bar{V}^{-1}(\norm{x-x'},0)\big) f(x)^2 \,\text{d}x' \,\text{d}x,
\end{equation*}

\noindent which is finite because the integrand decays exponentially with $\norm{x-x'}$ by \autoref{Vc}.

\subsection{Strategic Interactions and Equilibrium Selection}\label{eselection}

We next define the statistics $S_{ij}$ that capture strategic interactions. For any $H\subseteq\N_n$, let $r_n^{-1}\bm{T}_H = ((r_n^{-1}X_i,Z_i))_{i \in H}$ be the array of (scaled) types for nodes in $H$ and $\Zeta_H = (\zeta_{ij})_{i,j \in H}$ (with $\zeta_{ii} \equiv 0$ for all $i$) be the corresponding array of random-utility shocks. In the case where $H = \N_n$, we abbreviate $r_n^{-1}\bm{T}_n \equiv r_n^{-1}\bm{T}_{\N_n}$ and $\Zeta_n \equiv \Zeta_{\N_n}$. Let $\bm{A}_H = (A_{ij})_{i,j\in H}$ denote the subnetwork of $\bm{A}$ on $H$. We define
\begin{equation}
  S_{ij} = S_n(i, j, r_n^{-1}\bm{T}_n, \bm{A}) \label{Sijt}
\end{equation}

\noindent where $\{S_n\}_{n\in\mathbb{N}}$ is a sequence of permutation-invariant functions.\footnote{For any bijection $\pi\colon \N_n \rightarrow \N_n$, we abuse notation by defining $\pi(r_n^{-1}\bm{T}_n) = ( (r_n^{-1}X_{\pi(i)}, Z_{\pi(i)}) )_{i\in \N_n}$ and $\pi(\bm{A}) = (A_{\pi(i)\pi(j)})_{i,j}$. We say $S_n(\cdot)$ is {\em permutation-invariant} if $S_n(i, j, r_n^{-1}\bm{T}_n, \bm{A}) = S_n(\pi(i), \pi(j), \pi(r_n^{-1}\bm{T}_n), \pi(\bm{A}))$. This means that the function does not directly depend on node labels which have no economic content in our model. \label{pinvar}} 

We next impose the common restriction that $S_{ij}$ only depends on its arguments through the 1-neighborhoods of $i$ and $j$, recalling the definition from \autoref{sintro}.

\begin{assump}[Local Externalities]\label{S}
  For any $r>0$, $n\in\mathbb{N}$, $i,j\in\N_n$, and $H = \N_{\bm{A}}(i,1) \cup \N_{\bm{A}}(j,1)$, $S_n(i, j, r^{-1}\bm{T}_n, \bm{A}) = S_{\abs{H}}(i, j, r^{-1}\bm{T}_H, \bm{A}_H)$.
\end{assump}

\noindent That is, $S_{ij}$ is only a function of nodes linked to either $i$ or $j$. Most statistics used in the literature satisfy this assumption, including those in Examples \ref{elinear} and \ref{esheng}, which appear to be the most common choices \citep{christakis2020empirical,leung2016,mele2017,sheng2014}. Dependence of $S_n(\cdot)$ on types also allows for covariate-weighted versions of these examples, such as the weighted degree $\sum_{k=1}^n A_{ik} Z_k$.

The model thus far is incomplete since multiple networks $\bm{A}$ may satisfy \eqref{model} due to strategic interactions. Let $\mathcal{E}(r_n^{-1}\bm{T}_n, \Zeta_n)$ denote the set of such networks, which corresponds to the set of pairwise stable equilibria. To complete the model, we introduce a {\em selection mechanism}, which is the reduced-form mapping from primitives to the observed network. It is a representation of the latent social process by which nodes coordinate on an equilibrium.

\begin{assump}[Equilibrium Selection]\label{aesel} 
  For any $r>0$ and $n\in\mathbb{N}$, an equilibrium exists in that $\mathcal{E}(r^{-1}\bm{T}_n, \Zeta_n)$ is non-empty, and there exists a permutation-equivariant function $\lambda_n(\cdot)$ (the {\em equilibrium selection mechanism}) such that $\bm{A} = \lambda_n(r^{-1}\bm{T}_n, \Zeta_n) \in \mathcal{E}(r^{-1}\bm{T}_n, \Zeta_n)$.\footnote{Following the notation in \autoref{pinvar}, we say $\lambda_n(\cdot)$ is permutation-equivariant if $\pi(\lambda_n(r^{-1}\bm{T}_n, \Zeta_n)) = \lambda_n(\pi(r^{-1}\bm{T}_n), \pi(\Zeta_n))$, where $\pi(\Zeta_n)$ is defined similarly to $\pi(\bm{A})$.}
\end{assump}

\begin{remark}\label{rselectmech}
  The empirical games literature typically represents the equilibrium selection mechanism as a conditional distribution $\sigma(\cdot \mid r^{-1}\bm{T}_n, \Zeta_n)$ over $\mathcal{E}(r^{-1}\bm{T}_n, \Zeta_n)$. We represent it as a deterministic function, which is more convenient for our purposes, especially for formulating \autoref{nocoord} below. We next show that for any conditional distribution, we can construct a deterministic function that induces the same distribution over equilibria. The main idea is that types can always include payoff-irrelevant components that can be used to randomize over equilibria. Specifically, let $\{\nu_i\}_{i=1}^n \stackrel{iid} \sim \mathcal{U}([0,1])$ be independent of all structural primitives. These will serve to generate a distribution over equilibria conditional on structural primitives. Redefine $r^{-1}\bm{T}_n = ((r^{-1}X_i,Z_i,\nu_i))_{i \in \N_n}$ (scaled types with payoff-irrelevant components), and let $r^{-1}\tilde{\bm{T}}_n = ((r^{-1}X_i,Z_i))_{i \in \N_n}$ (scaled types as originally defined).

  The set of equilibria is given by $\mathcal{E}(r^{-1}\tilde{\bm{T}}_n, \Zeta_n) = \{\bm{A}_k\}_{k=1}^m$. Consider any conditional distribution $\sigma$ over this set, where we abbreviate $\sigma_k \equiv \sigma(\bm{A}_k \mid r^{-1}\tilde{\bm{T}}_n, \Zeta_n)$. Define $\nu^* = F(1 - \max \{\nu_i\}_{i=1}^n)$, where $F$ is the CDF of $1 - \max \{\nu_i\}_{i=1}^n$. Partition $[0,1]$ into $m$ intervals of lengths $\sigma_1, \ldots, \sigma_m$, and let $\lambda_n(r^{-1}\bm{T}_n, \Zeta_n)$ be the function that selects equilibrium $\bm{A}_k$ if $\nu^*$ falls within the interval associated with $\sigma_k$. Since $\nu^* \sim \mathcal{U}([0,1])$ for any $n$, we have $\prob(\lambda_n(r^{-1}\bm{T}_n, \Zeta_n) = \bm{A}_k \mid r^{-1}\tilde{\bm{T}}_n, \Zeta_n) = \sigma_k$. In other words, our deterministic construction of $\lambda_n(\cdot)$ generates the desired conditional distribution over equilibria.
\end{remark}

\subsection{Network Moments}\label{smoms}

Our objective is to prove a CLT for network moments that are averages of {\em node statistics} 
\begin{equation*}
  \frac{1}{n} \sum_{i=1}^n \psi_n(i,r_n^{-1}\bm{T}_n, \Zeta_n, \bm{A}),
\end{equation*}

\noindent where $\{\psi_n\}_{n\in\mathbb{N}}$ is a sequence of $\R^{d_\psi}$-valued, permutation-invariant functions (see \autoref{pinvar}). We often abbreviate $\psi_i(\N_n) \equiv \psi_n(i,r_n^{-1}\bm{T}_n, \Zeta_n, \bm{A})$.

\begin{example}[Subnetwork Counts]\label{esubcts}
  A simple network moment is the {\em average degree} for which the node statistic is the degree $\psi_i(\N_n) = \sum_{j=1}^n A_{ij}$. This is a permutation-invariant function of the network. Average degree is proportional to the dyad count (number of links). More generally, we can count any other connected subnetwork, such as the number of triangles, $k$-stars, or complete networks on $k$-tuples. For instance, the triangle count is proportional to $\sum_{i,j,k} A_{ij}A_{jk}A_{ik}$, with corresponding node statistic $\psi_i(\N_n) = \sum_{j,k} A_{ij}A_{jk}A_{ik}$. See \autoref{sbdmoms} for a formal definition of subnetwork counts.
\end{example}

More generally, we consider the following class of node statistics that includes the previous examples. Recall that $\N_{\bm{A}}(i,K)$ is $i$'s $K$-neighborhood, defined in \autoref{sintro}. 

\begin{assump}[Node Statistics]\label{dyklocal}
  For some $K\in\mathbb{Z}_+$ and any $r>0$, $n\in\mathbb{N}$, $i\in\N_n$, and $H = \N_{\bm{A}}(i,K)$, $\psi_n(i, r^{-1}\bm{T}_n, \Zeta_n, \bm{A}) = \psi_{\abs{H}}(i, r^{-1}\bm{T}_H, \Zeta_H, \bm{A}_H)$.
\end{assump}

\noindent This states that the node statistic only depends on the types, random-utility shocks, and subnetwork of nodes on $i$'s $K$-neighborhood. The average degree and triangle count satisfy this for $K=1$. A more complex example is the number of nodes at most path distance $D$ from node $i$, which satisfies this assumption for $K=D$. 

\begin{example}\label{eshenginf}
  Subnetwork counts are useful for structural inference. Suppose $V(\cdot)$ is known up to some vector of parameters $\theta_0$. \cite{sheng2014} defines an identified set for $\theta_0$ in terms of moment inequalities of the form
  \begin{equation}
    \frac{1}{n} \sum_{i_1=1}^n \dots \sum_{i_m=1}^n  \E[(G_{i_1,\ldots,i_m} - H_{i_1,\ldots,i_m}(\theta_0)) q_{i_1,\ldots,i_m}(\bm{T}_n)] \leq 0.\footnote{Technically this is proportional to the expectation of (5.3) in \cite{sheng2014}. Our scaling is different because we consider sparse networks.} \label{sheng}
  \end{equation}

  \noindent Here $\sum_{i_1=1}^n \ldots \sum_{i_m=1}^n  G_{i_1,\ldots,i_m}$ is proportional to a count of a particular connected subnetwork of size $m$. For example, for counting triangles, which are fully connected subnetworks with $m=3$, $G_{i_1,\ldots,i_m} = A_{i_1i_2}A_{i_2i_3}A_{i_3i_1}$. The upper bound $H_{i_1,\ldots,i_m}(\theta_0)$ and instrument function $q_{i_1,\ldots,i_m}(\bm{T}_n)$ are known, deterministic functions of the observed component of $\bm{T}_n$. 

  Furthermore, $n^{-1} \sum_{i_1=1}^n \ldots \sum_{i_m=1}^n H_{i_1,\ldots,i_m}(\theta_0) q_{i_1,\ldots,i_m}(\bm{T}_n)$ is proportional to a U-statistic of order $m$ with a kernel that is a deterministic function of $\{(X_{i_k},Z_{i_k})\}_{k=1}^m$. By the Hoeffding decomposition, it equals $n^{-1} \sum_{i=1}^n J_i(\theta_0) + o_p(n^{-1/2})$ for some $J_i(\theta_0)$ with the same mean that is a deterministic function of $i$'s type $(X_i,Z_i)$. Hence, $J_i(\theta_0)$ satisfies \autoref{dyklocal} for $K=0$. Additionally, $\tilde\psi_{i_1}(\N_n) \equiv \sum_{i_2=1}^n \dots \sum_{i_m=1}^n G_{i_1,\ldots,i_m} q_{i_1,\ldots,i_m}(\bm{T}_n)$ is a node statistic satisfying \autoref{dyklocal} for some $K \leq m-1$ since the subnetwork is connected. We have therefore shown that
  \begin{multline*}
    \frac{1}{\sqrt{n}} \sum_{i_1=1}^n \dots \sum_{i_m=1}^n \big[ (G_{i_1,\ldots,i_m} - H_{i_1,\ldots,i_m}(\theta_0)) q_{i_1,\ldots,i_m}(\bm{T}_n) \\ - \E[(G_{i_1,\ldots,i_m} - H_{i_1,\ldots,i_m}(\theta_0)) q_{i_1,\ldots,i_m}(\bm{T}_n)] \big] \\ = \frac{1}{\sqrt{n}} \sum_{i=1}^n \big(\underbrace{\tilde\psi_i(\N_n) - J_i(\theta_0)}_{\psi_i(\N_n)} - \E[\tilde\psi_i(\N_n) - J_i(\theta_0)]\big) + o_p(1).
  \end{multline*}

  \noindent Our CLT can be applied to the first term on the last line.
\end{example}

\section{Stabilization}\label{smethpf}

This section provides high-level conditions for a CLT. Readers interested in the application to strategic network formation may wish to skip to \autoref{scltstatic} where we state primitive sufficient conditions.

Recall from \autoref{aesel} that for any $H \subseteq \N_n$, $\lambda_{\abs{H}}(r_n^{-1}\bm{T}_H, \Zeta_H)$ is the equilibrium network under the counterfactual model in which the set of nodes is $H$ rather than $\N_n$. Abusing notation, for $i \in H$, we define $i$'s {\em counterfactual node statistic}
\begin{equation}
  \psi_i(H) \equiv \psi_{\abs{H}}\big(i, r_n^{-1}\bm{T}_H, \Zeta_H, \lambda_{\abs{H}}(r_n^{-1}\bm{T}_H, \Zeta_H)\big). \label{psiH}
\end{equation}

\noindent For $H=\N_n$, this coincides with the original definition. For $H\subset \N_n$, the counterfactual interpretation is due entirely to the last argument. Contrast \eqref{psiH} with
\begin{equation}
  \psi_{\abs{H}}(i, r_n^{-1}\bm{T}_H, \Zeta_H, \bm{A}_H).
  \label{contrast}
\end{equation}

\noindent For example, if $\psi_i(\N_n)$ is the degree of node $i$ on the full network, then \eqref{contrast} is $i$'s degree in the observed subnetwork restricted to $H$, which can be directly computed from observation of $\bm{A}$. Generally, this does not equal \eqref{psiH} because the networks $\bm{A}_H$ and $\lambda_{\abs{H}}(r_n^{-1}\bm{T}_H, \Zeta_H)$ do not coincide due to strategic interactions, which induce nodes in $H$ to form different links when nodes $\N_n\backslash H$ are absent from the model. 

Recall that $d$ is the dimension of $X_1$. For any $B \subseteq \R^d$, let $\N_n(B) = \{i\in\N_n\colon r_n^{-1}X_i \in B\}$, the set of nodes with scaled positions lying in the set $B$. For $x \in \R^d$ and $R \geq 0$, let $B(x,R) = \{y \in \R^d\colon \norm{x-y} \leq R\}$.

\begin{definition}\label{rosdef}
  For any $r > 0$, $n\in\mathbb{N}$, and $i \in \N_n$, the {\em radius of stabilization} 
  \begin{equation*}
    \RR_i(n,r) \equiv \RR(i, r^{-1}\bm{T}_n, \Zeta_n, \lambda_n)
  \end{equation*}
  
  \noindent is the smallest integer-valued $R\geq 0$ such that $\psi_i(\N_n) = \psi_i(\N_n(\mathcal{B}_i))$ for all sets $\mathcal{B}_i \subseteq \R^d$ containing $B(r^{-1}X_i, R)$.
\end{definition}

\noindent This is the smallest radius $R$ such that $i$'s node statistic has the same value under all counterfactual models that drop nodes positioned outside of $i$'s neighborhood $B(r^{-1}X_i, R)$. 

The main idea is as follows. If this radius were ``small,'' then $\psi_i(\N_n)$ would primarily depend on a small fraction of nodes, specifically those positioned near $i$, in which case node statistics should be weakly dependent. To establish a law of large numbers (LLN), \cite{leung2016} requires the radius to be $O_p(1)$, analogous to a condition due to \cite{py2003}. However, a CLT seems to require stronger conditions on the distribution of the radius, in particular the thickness of its tails \citep{penrose2007gaussian}. 

\subsection{Abstract CLT}\label{sabsclt}

We first define some notation and the sequence along which we take limits. Let $\{(X_i,Z_i)\}_{i\in\mathbb{N}}$ and $\{\zeta_{ij}\colon i\in\N_n, j>i\}$ both be i.i.d.\ and independent, and assume the density $f$ of $X_1$ is bounded. We consider a sequence indexed by $n\in\mathbb{N}$ such that the $n$th element of the sequence is given by the tuple 
\begin{equation*}
  (r_n^{-1}\bm{T}_n, \Zeta_n, \lambda_n, \psi_n).
\end{equation*}

\noindent The first three components are structural primitives that determine $\bm{A}$, while the last is the node statistic function. 

For technical reasons, we need to consider sequences of models indexed by $n$ in which the network size is not $n$ but rather of the same asymptotic order as $n$. Let $\{N_n\}_{n\in\mathbb{N}}$ be a sequence of random variables independent of all model primitives such that
\begin{equation}
  N_n \sim \text{Poisson}(n).
  \label{poisson}
\end{equation}

\noindent A de-Poissonization argument discussed below requires us to consider two different network sizes: a random size $N_n+k$, where $k$ is a constant, and a non-random size $m+k$, where $m=m_n$ with $m_n/n \rightarrow c \in (0,\infty)$. Our original setup with $n$ nodes corresponds to $m=n$ and $k=0$, and no intuition is lost to the reader who only considers this case. The need to consider models with different network sizes is purely for technical reasons clarified in \autoref{spfsketch}.

Define $\mathcal{T} = \text{supp}(X_1,Z_1)$, $\mathcal{T}^k = \bigtimes_{i=1}^k \mathcal{T}$, and $\bm{T}_k = ((X_i,Z_i))_{i=1}^k$. Our first assumption is the key high-level weak dependence condition that controls the tail behavior of the distribution of the radius of stabilization, or more formally the conditional probability $\prob(\RR_1(m+k,r_n) > w \mid \bm{T}_k = \bm{t}_k)$. Since probabilities conditional on continuous random variables are not uniquely defined, we require some additional notation to phrase the assumption in terms of a {\em version} of the conditional probability. Let $m,k \in \mathbb{N}$ and $\bm{t}_k \in \mathcal{T}^k$. Construct $r_n^{-1}\bm{T}_{m+k}^\star(\bm{t}_k)$ from $r_n^{-1}\bm{T}_{m+k}$ by replacing the types of the first $k$ nodes $\bm{T}_k$ with fixed values $\bm{t}_k$. Define $\RR_1^\star(m+k,r_n;\bm{t}_k) \equiv \RR(1, r_n^{-1}\bm{T}_{m+k}^\star(\bm{t}_k), \Zeta_{m+k}, \lambda_{m+k})$. Then $\prob(\RR_1^\star(m+k,r_n;\bm{t}_k) > w)$ is a version of the conditional probability $\prob(\RR_1(m+k,r_n) > w \mid \bm{T}_k = \bm{t}_k)$.

\begin{assump}[Exponential Stabilization]\label{main-bp-exp-stab}
  For any $k \in \{1,\dots,4\}$ and $k'\in\{1,2\}$, there exist $n_0,\epsilon>0$ and $\eta \in (0,1]$ such that 
  \begin{align*}
    &\limsup_{w\rightarrow\infty} w^{-\eta} \max\{ \log \tau_{b,\epsilon}(w), \log \tau_p(w)\} < 0, \quad\text{where} \\
    &\tau_{b,\epsilon}(w) = \sup_{n>n_0} \sup_{m \in ((1-\epsilon)n, (1+\epsilon)n)} \sup_{\bm{t}_k \in \mathcal{T}^k} \prob\big( \RR_1^\star(m+k,r_n;\bm{t}_k) > w \big), \\
    &\tau_p(w) = \sup_{n>n_0} \sup_{\bm{t}_{k'} \in \mathcal{T}^{k'}} \prob\big( \RR_1^\star(N_n+k',r_n;\bm{t}_{k'}) > w \big).
  \end{align*}
\end{assump} 

\noindent For $\eta=1$, this says that the distribution of the radius of stabilization has an exponential tail under models with $N_n+k$ and $m+k$ nodes, uniformly over the types of the first $k$ nodes.\footnote{The first $k$ is an arbitrary choice of $k$ nodes. Note that nodes are exchangeable in our model since types and random-utility shocks are i.i.d.\ and the equilibrium selection mechanism is permutation-equivariant.} Technically the assumption allows for slower than exponential decay with $\eta < 1$, but it is otherwise analogous to conditions used by \cite{penrose2005normal} and \cite{penrose2007gaussian}. Our method of proof for the result below is based on theirs but with important differences in setup discussed in \autoref{srelit}.

Contrast \autoref{main-bp-exp-stab} with the more familiar concept of $m$-dependence, which states that an observation is only correlated with nodes in some non-random neighborhood of known radius $m$. Stabilization generalizes this to allow neighborhoods to be node-specific, random, and complex functions of the primitives. In general, bounding the size of this set is far from trivial, and one of our main contributions is to demonstrate that branching processes can be used for this purpose to derive primitive sufficient conditions, as discussed in \autoref{smainres}.

The next assumption imposes a moment condition. Similar to the previous assumption, we require additional notation to phrase the condition in terms of a version of a conditional expectation. Let $k\in\mathbb{N}$, $\bm{t}_k \in \mathcal{T}^k$, and $H \subseteq \mathbb{N}$ be a finite set containing $\{1, \ldots, k\}$. Construct $r_n^{-1}\bm{T}_H^\star(\bm{t}_k)$ from $r_n^{-1}\bm{T}_H$ by replacing the types of the first $k$ nodes $\bm{T}_k$ with fixed values $\bm{t}_k$. Define $\psi_1^\star(H; \bm{t}_{k}) \equiv \psi_{\abs{H}}(1, r_n^{-1}\bm{T}_H^\star(\bm{t}_k), \Zeta_H, \lambda_{\abs{H}})$, so that $\E[\psi_1^\star(H; \bm{t}_{k})]$ is a version of the conditional expectation $\E[\psi_1(H) \mid \bm{T}_k = \bm{t}_k]$. Finally, let $\norm{x}_\infty$ denote the entry-wise maximum of a vector or matrix $x$.

\begin{assump}\label{main-bp-mom}
  (a) There exist $p>2,\epsilon>0$, $M<\infty$, and $n_0 \in \mathbb{N}$ such that for all $n>n_0$, $k\in\{1,2,3\}$, $m \in ((1-\epsilon)n, (1+\epsilon)n)$, and $\bm{t}_k\in\mathcal{T}^k$,
  \begin{equation*}
    \max\big\{ \E[\norm{\psi_1^\star(\N_{m+k}; \bm{t}_k)}_\infty^p], \E[\norm{\psi_1^\star(\N_{N_n+k} \cap (H_n\cup\{1\}); \bm{t}_k)}_\infty^p] \big\} < M
  \end{equation*}

  \noindent for any sequence of sets $H_n \subseteq \mathbb{N}$. (b) For any $r>0$, $n\in\mathbb{N}$, $i\in\N_n$, and $x \in \R^d$,
  \begin{align*}
    &\psi_n(i, r^{-1}\bm{T}_n, \Zeta_n, \bm{A}) = \psi_n\big(i, ((r^{-1}X_j+x,Z_j))_{j=1}^n, \Zeta_n, \bm{A}\big) \quad\text{and} \nonumber \\ &\lambda_n(r^{-1}\bm{T}_n, \Zeta_n) = \lambda_n\big((r^{-1}X_j+x,Z_j))_{j=1}^n, \Zeta_n\big).
  \end{align*} 
\end{assump} 

\noindent Part (a) requires node statistics to have bounded $p>2$ moments, uniformly over the types of $k$ nodes. We provide primitive conditions in \autoref{sbdmoms} for the case of subnetwork counts (\autoref{esubcts}). 

Part (b) says that node statistics and the selection mechanism are invariant to additive shifts in scaled positions $\{r^{-1}X_i\}_{i=1}^n$, which holds if these only enter through scaled distances $r^{-1}\norm{X_i-X_j}$. Because $X_i$ and $Z_i$ may be arbitrarily dependent, this still allows $X_i$ (unscaled) to enter the model as a subvector of $Z_i$. In Examples \ref{elinear} and \ref{esheng}, the latent index only depends on types through scaled distances or $Z_i$, so requiring the selection mechanism and node statistics to satisfy the same property does not apparently rule out any economically interesting applications.\footnote{For example, \autoref{main-bp-mom}(b) disallows network moments such as $n^{-1} \sum_{i=1}^n \sum_{j=1}^n r^{-1} X_j A_{ij}$, but since $X_i$ and $Z_i$ may be arbitrarily correlated, it allows for $n^{-1} \sum_{i=1}^n \sum_{j=1}^n X_j A_{ij}$.}

Let $\bm{\Sigma}_n = n^{-1} \var(\sum_{i=1}^n \psi_i(\N_n))$, $\lambda_\text{min}(\bm{\Sigma}_n)$ be its smallest eigenvalue, and $\bm{I}$ denote the $d_\psi$-dimensional identity matrix.

\begin{theorem}\label{master-clt}
  Under Assumptions \ref{main-bp-exp-stab} and \ref{main-bp-mom}, $\sup_n \norm{\bm{\Sigma}_n}_\infty < \infty$. Further suppose that $\liminf_{n\rightarrow\infty} \lambda_\text{min}(\bm{\Sigma}_n) > 0$. Then
  \begin{equation}
    \bm{\Sigma}_n^{-1/2} \frac{1}{\sqrt{n}} \sum_{i=1}^n \big( \psi_i(\N_n) - \E[\psi_i(\N_n)] \big) \dlimarrow \N(\zero,\bm{I}).  \label{thenormal}
  \end{equation}
\end{theorem}
\begin{proof}
  See \autoref{smaster}.
\end{proof}

The proof does not rely on the setup or assumptions in \autoref{smodel}, other than the distributional assumptions at the top of \autoref{smodel} and the requirement that $\lambda_n(\cdot)$ is a deterministic, permutation-equivariant function of the structural primitives. The other assumptions will be used in the next section to verify \autoref{main-bp-exp-stab} for strategic network formation.\footnote{\autoref{master-clt} can be applied to other network models. For instance, in \cite{leung2019inference}, $\lambda_n(\cdot)$ is the reduced-form mapping that takes as input the structural primitives and outputs both the network and the outcome of a second-stage social interactions model given the network.}

\subsection{Outline of Proof}\label{spfsketch}

{\bf Step 1.} We first establish a CLT for the ``Poissonized'' model in which the number of nodes is $N_n$ defined in \eqref{poisson}, so-called because $\{X_i\}_{i=1}^{N_n}$ has the same distribution as a Poisson point process with intensity function $nf(\cdot)$ \citep[][Proposition 1.5]{penrose2003}. Specifically, we show 
\begin{equation}
  \tilde{\bm{\Sigma}}_n^{-1/2} \frac{1}{\sqrt{n}}\left( \sum_{i=1}^{N_n} \psi_i(\N_{N_n}) - \E\left[ \sum_{i=1}^{N_n} \psi_i(\N_{N_n}) \right] \right) \dlimarrow \N(\zero,\bm{I}), \label{poistep}
\end{equation}

\noindent where $\tilde{\bm{\Sigma}}_n = n^{-1} \var(\sum_{i=1}^{N_n} \psi_i(\N_{N_n}))$. This is simpler to prove directly because, unlike $\{X_i\}_{i=1}^n$, the Poisson process $\{X_i\}_{i=1}^{N_n}$ possesses a well-known spatial independence property: for any disjoint subsets $S_1,S_2 \subseteq \R^d$, $\abs{\{X_i\}_{i=1}^{N_n} \cap S_1} \indep \abs{\{X_i\}_{i=1}^{N_n} \cap S_2}$. 

We prove \eqref{poistep} by adapting a theorem due to \cite{penrose2005normal} (see our \autoref{poi-clt}). We proceed by first partitioning the support of $X_1$ into cubes $Q_1, \dots, Q_{V_n}$ of slowly growing volume and representing the moment $n^{-1} \sum_{i=1}^{N_n} \psi_i(\N_{N_n})$ as a double sum over cubes and nodes within cubes $n^{-1} \sum_{i=1}^{V_n} \sum_{j=1}^{N_i} \psi_{ij}(\N_{N_n})$, where $N_i$ is the number of nodes positioned in $Q_i$ and $ij$ indexes the $j$th node in $Q_i$. Spatial independence of the Poisson process implies independence of node positions across cubes. However, node statistics are complex functionals of the process and hence are not generally independent across cubes.

Since stabilization implies that node statistics $\psi_{ij}(\N_n)$ are primarily determined by nodes relatively proximate to $ij$, and hence cubes relatively proximate to $Q_i$, it can be shown that the dependence structure of $\{\sum_{j=1}^{N_i} \psi_{ij}(\N_{N_n})\}_{i=1}^{V_n}$ is ``approximately'' characterized by a certain ``dependency graph.'' This is a network in which, roughly speaking, two observations are linked if and only if they are dependent. In our case, observations are cubes, which we connect if and only if they are relatively proximate. Careful construction of the cubes ensures a small approximation error for the dependency graph characterization, and applying a CLT for dependency graphs \citep{chen2004normal} delivers the result. For precise details, see \autoref{spois}.

\bigskip
\noindent {\bf Step 2.} Since $N_n/n \plimarrow 1$, $\{X_i\}_{i=1}^n$ and $\{X_i\}_{i=1}^{N_n}$ should be ``similar,'' so given \eqref{poistep}, we expect a similar result for the original model. The second ``de-Poissonization'' step of the proof, which follows \cite{penrose2007gaussian}, shows that this intuition is correct, provided we properly adjust $\tilde{\bm{\Sigma}}_n$ downward to obtain the correct variance $\bm{\Sigma}_n$ (see our \autoref{pen2.3}). This is needed because $N_n$ contributes additional randomness to the asymptotic distribution. 

More specifically, define the {\em add-one cost}
\begin{equation}
  \Xi_n = \psi_{n+1}(\N_{n+1}) + \sum_{i=1}^n \big( \psi_i(\N_{n+1}) - \psi_i(\N_n) \big). \label{add1}
\end{equation}

\noindent This is the aggregate counterfactual impact on the total $\sum_{i=1}^n \psi_i(\N_n)$ from adding a new node labeled $n+1$ to the model. The first term is the direct effect of adding $n+1$, which is its own node statistic. The second term is the indirect effect, which is the new node's impact on the statistics of all other nodes. A key step of the proof establishes that
\begin{multline*}
  n^{-1/2} \left( \sum_{i=1}^{N_n} \psi_i(\N_{N_n}) - \E\left[ \sum_{i=1}^{N_n} \psi_i(\N_{N_n}) \right] \right) \\ = n^{-1/2} \left( \sum_{i=1}^n \psi_i(\N_n) - \E\left[ \sum_{i=1}^n \psi_i(\N_n) \right] \right) \\ + n^{-1/2}(N_n-n)\E[\Xi_{N_n}] + o_p(1).
\end{multline*}

\noindent This may be viewed as a first-order expansion in the number of nodes, comparing $N_n$ to $n$. The ``derivative'' is $\E[\Xi_{N_n}]$ since it captures the change in moments as a result of a unit increment in the number of nodes. By \eqref{poistep}, the left-hand side is asymptotically normal, and by the Poisson CLT, so is $n^{-1/2}(N_n-n)\E[\Xi_{N_n}]$. Because $N_n$ is independent of all other primitives, we can then establish that the first term on the right-hand side is asymptotically normal by an argument using characteristic functions; for details see the end of \autoref{sdepois}.

\subsection{Related Literature}\label{srelit}

The proof is closely based on arguments in \cite{penrose2005normal} and \cite{penrose2007gaussian}, whose results pertain to geometric graphs without strategic interactions. The innovation in \autoref{master-clt} is primarily conceptual, namely, the recognition that an appropriate modification of stabilization allows us to adapt their results to econometric models. Our main technical innovation will be discussed in \autoref{scltstatic}, namely the use of branching processes to derive primitive conditions for stabilization in models with strategic interactions. \cite{leung2016} uses branching processes to establish an LLN, while we tackle the more difficult task of proving a CLT, which requires us to establish new tail bounds for the radius of stabilization (\autoref{dyDC}).

The setup and assumptions used in Penrose's work are not directly applicable to our setting, so we cannot simply verify their conditions. However, we show that their proofs can be translated to our setting, which differs in three main aspects. The first is the definition of the radius of stabilization. We reformulate the definition in terms of counterfactual models \eqref{psiH} and require invariance of $i$'s node statistic to the removal of nodes outside of $\N_n(\mathcal{B}_i)$. Existing definitions demand invariance to the removal {\em and addition} of new nodes, but invariance to addition is typically violated in the models we study due to strategic interactions. Second, $X_i$ may be correlated with $Z_i$ in our setup, whereas the literature requires independence, but this turns out to have little effect on the proofs. Third, our model includes pair-specific shocks $\zeta_{ij}$, which pose little problem due to their high degree of independence. These are independent across pairs, in contrast to type pairs $((X_i,Z_i),(X_j,Z_j))$ which are correlated across pairs sharing a common node, for example $(i,j)$ and $(i,k)$.

\section{CLT for Network Formation}\label{scltstatic}

Exponential stabilization (\autoref{main-bp-exp-stab}) provides a high-level formulation of weak dependence. This section derives primitive conditions for the network formation model in \autoref{smodel}, so throughout this section we work under its setup. We begin in \autoref{sstrnbh} by introducing key definitions used in \autoref{sdepstr} to explain two sources of cross-sectional dependence induced by the model. These motivate the weak dependence conditions stated in \autoref{sSOI} and \autoref{scoord}. In \autoref{smainres} we present the main result, that these conditions imply exponential stabilization. We outline the method of proof in \autoref{smeth2}. 

\subsection{Strategic Neighborhood}\label{sstrnbh}

Recall the definition of the joint surplus from \eqref{model}, and let
\begin{multline}
  D_{ij} = \ind\big\{ \sup_s V(r_n^{-1}\norm{X_i-X_j},s,Z_i,Z_j,\zeta_{ij}) > 0 \big\} \\ \times \ind\big\{ \inf_s V(r_n^{-1}\norm{X_i-X_j},s,Z_i,Z_j,\zeta_{ij}) \leq 0 \big\}. \label{Dij}
\end{multline}

\noindent This is an indicator for whether the potential link $A_{ij}$ is {\em non-robust}. If $\inf_s V(r_n^{-1}\norm{X_i-X_j},s,Z_i,Z_j,\zeta_{ij}) > 0$, then $A_{ij}=1$, and the link is {\em robust} in that the joint surplus is positive regardless of what other links are formed. This is because $\bm{A}$ enters $V(\cdot)$ only through $S_{ij}$. Likewise, if $\sup_s V(r_n^{-1}\norm{X_i-X_j},s,Z_i,Z_j,\zeta_{ij}) \leq 0$, then $A_{ij}=0$, and the link is {\em robustly absent} in that the joint surplus is negative regardless of what other links are formed. In either case, $D_{ij}=0$. If instead $D_{ij}=1$, then $A_{ij}$ may be 0 or 1, and the potential link is non-robust in that the sign of the joint surplus is responsive to links formed by others.

Let $\bm{D}$ be the network of non-robustness indicators with $ij$th entry $D_{ij}$. Let $C_i$ denote $i$'s component in $\bm{D}$, recalling from \autoref{sintro} that a component is a connected subnetwork that is disconnected from the rest of the network. Let $\bm{\Pi}$ be the network of robust link indicators with $ij$th entry 
\begin{equation*}
  \Pi_{ij} = \ind\big\{ \inf_s V(r_n^{-1}\norm{X_i-X_j},s,Z_i,Z_j,\zeta_{ij}) > 0 \big\}.
\end{equation*}

\noindent Recall that $\N_{\bm{\Pi}}(i,1)$ denotes $i$'s 1-neighborhood in $\bm{\Pi}$, which includes $i$ itself. A crucial concept for what follows is a node's {\em strategic neighborhood}, given by
\begin{equation}
  C_i^+ = \bigcup \left\{\N_{\bm{\Pi}}(j,1)\colon j \in C_i\right\}. \label{stratneigh}
\end{equation}

\noindent This adds to $C_i$ the set of all nodes that possess a robust link to some member of $C_i$.

\begin{example}\label{enonrobust}
  Consider \autoref{elinear}, and suppose $\theta_2 > 0$. Then $A_{ij}$ is robust if $\theta_1 -\theta_3 (r_n^{-1}\norm{X_i-X_j})^2 + \zeta_{ij} > 0$ and robustly absent if $\theta_1  + \theta_2 - \theta_3 (r_n^{-1}\norm{X_i-X_j})^2 + \zeta_{ij} \leq 0$. The non-robust indicator is
  \begin{equation*}
    D_{ij} = \ind\big\{ -\theta_2 < \theta_1 - \theta_3 (r_n^{-1}\norm{X_i-X_j})^2 + \zeta_{ij} \leq 0 \big\}.
  \end{equation*}

  \noindent As the strength of strategic interactions $\theta_2$ increases, so does the right-hand side, and hence, the likelihood of non-robustness. In the case of no strategic interactions ($\theta_2=0$), there are no non-robust potential links.

  We can compute $C_i$ as follows. Initialize $C_i = \{i\}$, add $i$'s neighbors in the network $\bm{D}$ to the set, and then iteratively add neighbors of neighbors in the manner of a breadth-first search until there are no new nodes to add. To compute $C_i^+$, we set $\Pi_{ij} = \ind\{ \theta_1 - \theta_3 (r_n^{-1}\norm{X_i-X_j})^2 + \zeta_{ij} > 0 \}$ for all $i\neq j$ and add to $C_i$ all nodes that are neighbors under $\bm{\Pi}$ of some member of $C_i$.
\end{example}

The previous example illustrates how $D_{ij}$ is increasing in the strength of strategic interactions. Stronger interactions then imply that $C_i^+$ is a larger set for any given realization of the primitives. This suggests that, {\em when strategic neighborhoods are likely small in size, we expect weaker strategic interactions and hence weaker cross-sectional dependence.} We elaborate on this point in the next subsection.

\begin{figure}[ht]
  \centering
  \includegraphics[scale=0.45]{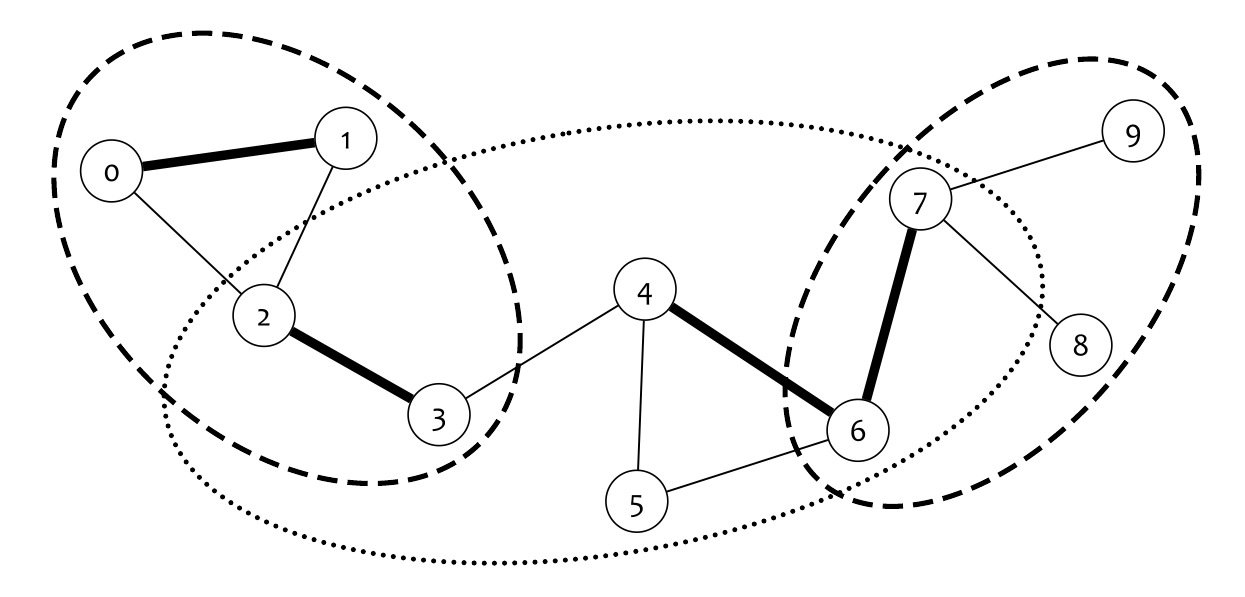}
  \caption{A network with three strategic neighborhoods.}\label{fsn} 
\end{figure}

\begin{example} \label{exstrngb}
  Let the primitives $(r_n^{-1}\bm{T}_n, \Zeta_n)$ be given, which fully determine the realizations of non-robustness indicators $\bm{D}$ and robust link indicators $\bm{\Pi}$ defined in \autoref{enonrobust}. Suppose their realizations are depicted in \autoref{fsn}, where thin lines represent non-robust potential links (those of $\bm{D}$), solid lines represent robust links (those of $\bm{\Pi}$), and the absence of a line between two nodes represents a robustly absent link. The observed network $\bm{A}$ is a subnetwork of the depicted network because if $D_{ij}=1$, then $A_{ij}$ is either 1 or 0, while if $D_{ij} = 0$, then the link between $i,j$ is either robustly absent, in which case $A_{ij} = 0$, or robust ($\Pi_{ij}=1$), in which case $A_{ij}=1$. The components of $\bm{D}$ are the three subnetworks obtained by removing the solid lines: $\{0,1,2\}$, $\{3,4,5,6\}$, and $\{7,8,9\}$. The strategic neighborhoods are the three circled subnetworks $\{0,1,2,3\}$, $\{2,\ldots,7\}$, and $\{6,7,8,9\}$, obtained by adding to each component the set of nodes with solid lines to any member of the component. Notice that strategic neighborhoods may have nodes in common, whereas components necessarily partition $\N_n$.
\end{example}

\subsection{Dependence Structure}\label{sdepstr}

Models of strategic network formation induce two forms of cross-sectional dependence, one due to link interdependencies directly induced by strategic interactions, and the other due to equilibrium selection. To illustrate, consider a model with $n=9$ nodes and realized structural primitives such that \autoref{fsn} depicts the resulting robust and non-robust potential links. Consider the impact on the realized network $\bm{A}$ of a hypothetical intervention that perturbs the value of $\zeta_{56}$. 

\subsubsection{Best-Response Chains}\label{sbrc}

The perturbation changes the joint surplus $V_{56}$ to some new value $V_{56}^{(1)}$, both evaluated under $\bm{A}$. If the signs of these values differ, then $A_{56}$ is no longer pairwise stable, so suppose in response the pair updates their potential link $A_{56}$ to $A_{56}^{(1)} = \ind\{V_{56}^{(1)} > 0\}$, resulting in the network $\bm{A}_1$. This update may affect the joint surplus of other potential links, for instance $V_{67}$ via $S_{67}$. However, because $A_{67}$ is a robust link, $V_{67}$ is always positive, whether evaluated under $\bm{A}$ or $\bm{A}_1$, so there is no subsequent update to $A_{67}$. On the other hand, $A_{45}$ is non-robust ($D_{45}=1$), so the update to $A_{56}$ may change the value of $V_{45}$ evaluated under $\bm{A}$ to a new value $V_{45}^{(2)}$ evaluated under $\bm{A}_1$ that has a different sign. Suppose in response the pair updates $A_{45}$ to $A_{45}^{(2)} = \ind\{V_{45}^{(2)} > 0\}$, resulting in a new network $\bm{A}_2$. This in turn affects the joint surplus of any other node pair containing nodes 4 and 5. The only such pairs with non-robust potential links are $(3,4)$ and $(5,6)$ since $D_{34}=D_{56}=1$, so {\em only} those potential links may update in response to the update to $A_{45}$. 

Suppose we iterate these best-response dynamics indefinitely. At no point do potential links other than those of pairs $(3,4)$, $(4,5)$, and $(5,6)$ update during the process. This is because they are the only pairs connected to nodes 5 or 6, the target of the initial perturbation, through a chain of non-robust potential links, or more formally a path in $\bm{D}$. These paths constitute the furthest extent that best-response chains extending from nodes 5 and 6 can travel. By definition, the component of $\bm{D}$ containing nodes 5 and 6 contains all such paths. 

The previous example pertains to interventions that perturb non-robust potential links, specifically $A_{56}$. Interventions that perturb robust links, such as $A_{67}$, trigger best-response dynamics in all strategic neighborhoods containing the link, which would be the two right-most strategic neighborhoods in \autoref{fsn}.

These examples demonstrate that the sizes of components $C_i$ and strategic neighborhoods $C_i^+$ (since they contain the components) are indicative of the degree of cross-sectional dependence. The key idea is that {\em if strategic neighborhoods are relatively small, then best-response chains cannot extend too far from the initial perturbation, corresponding to weaker dependence.} To control the length of best-response chains, in \autoref{sSOI}, we state a ``subcriticality'' condition that ensures that component sizes are asymptotically bounded. From the discussion following \autoref{enonrobust}, this should hold if strategic interactions are sufficiently weak. 

To derive the condition, we employ a well-known technique used in random graph theory for bounding the size of a component, which is to construct a branching process (see \autoref{introbp} for a formal definition) whose size stochastically dominates that of the component \citep[e.g.][]{bollobas2012simple}. The basic idea is to explore each $C_i$ via breadth-first search by starting at $i$, branching to its neighbors, neighbors of its neighbors, and so on. This is akin to growing a branching process, a model of population growth in which individuals in a given generation independently produce a random number of offspring, which corresponds to a node's neighborhood size. Subcriticality ensures that the average number of offspring is less than one, in which case the size of the process does not diverge and $\abs{C_i}$ is asymptotically bounded. 

While this is enough to establish an LLN \citep{leung2016}, a CLT additionally requires the distribution of $\abs{C_i}$ to have exponential tails. We utilize a tail bound for subcritical branching processes to obtain the desired result; see \autoref{Dexptail} which is proven by \cite{leung2019compute} using an argument due to \cite{turova_asymptotics_2012}. This is a key ingredient for verifying \autoref{main-bp-exp-stab}.

\subsubsection{Coordination}\label{scoord0}

The second source of dependence is due to equilibrium selection. Let us shut down the first source of dependence by supposing the perturbation to $\zeta_{56}$ is small enough not to change the sign of $V_{56}$, so there is no change to the network at any point in the best-response dynamics outlined above. However, if equilibrium selection is governed by a mechanism different from best-response dynamics, the perturbation may still substantially affect the network structure. 

Suppose $\abs{\mathcal{E}(r_n^{-1}\bm{T}_n, \Zeta_n)} = 2$. Since the perturbation does not change $V_{56}$, the set of equilibria is identical before and after the perturbation. However, one can construct a selection mechanism $\lambda_n(\cdot)$ that outputs one equilibrium under the structural primitives before the perturbation but outputs the other after the perturbation. Hence, the perturbation can alter potential links involving nodes external to $C_5$ even though any best-response chain must be limited to this component. This illustrates how, {\em under unrestricted equilibrium selection, all nodes may coordinate on the same ``signal'' ($\zeta_{56}$), resulting in strongly dependent potential links whose realizations all depend on this random variable}. 

In \autoref{scoord}, we state a condition that ensures coordination is ``decentralized'' in that strategic neighborhoods ``separately select'' their pairwise stable subnetworks. A consequence is that any perturbation to a node pair's types or random-utility shocks only affects the selection of the equilibrium subnetwork on strategic neighborhoods involving the pair. For example, in \autoref{fsn}, it will be the case that a perturbation to $\zeta_{56}$ only changes the equilibrium subnetwork on $\{2, \ldots, 7\}$ since no other neighborhood contains $\{5,6\}$. Combined with the subcriticality condition that ensures strategic neighborhood sizes are asymptotically bounded, this implies that the perturbation can only shift equilibrium selection within a bounded subset of nodes, ensuring weak cross-sectional dependence. This will allow us to construct a radius of stabilization that is asymptotically bounded.

\subsection{Strength of Interactions}\label{sSOI}

We next state a condition that controls the sizes of strategic neighborhoods, and hence length of best-response chains described in \autoref{sbrc}, by restricting the magnitude of strategic interactions. We measure strategic interaction strength by
\begin{align}
  p_{r_n}(X_i,Z_i,X_j,Z_j) =\,\, &\prob\big( \sup_s V(r_n^{-1}\norm{X_i-X_j},s,Z_i,Z_j,\zeta_{ij})>0 \mid X_i, Z_i, X_j, Z_j \big) \nonumber\\ -\,\, &\prob\big( \inf_s V(r_n^{-1}\norm{X_i-X_j},s,Z_i,Z_j,\zeta_{ij})>0 \mid X_i, Z_i, X_j, Z_j \big). \label{L_r}
\end{align}

\noindent This is the effect on link formation of changing $S_{ij}$ from its ``lowest'' to its ``highest'' possible value, conditional on types. In other words, it is the maximal change in linking probability induced by the strategic component of $V(\cdot)$. 

Let $\Phi_z(\cdot \,|\, x)$ be the conditional distribution of $Z_i$ given $X_i=x$. Recall that $f$ is the density of $X_1$, $\mathcal{T} = \text{supp}(X_1,Z_1)$, and $d_z$ is the dimension of $Z_i$. For any $h\colon \R^d \times \R^{d_z} \rightarrow \R$, define the mixed norm 
\begin{equation*}
  \norm{h}_\mathbf{m} = \sup_{x\in\R^d} \left( \int_{\R^{d_z}} h(x,z)^2 \,\text{d}\Phi^*(z) \right)^{1/2}, 
\end{equation*}

\noindent where $\Phi^*$ is a measure defined in the next assumption. Finally, for $\bar{f} = \sup_{x\in\R^d} f(x)$ and $\kappa$ in \eqref{r_n}, let
\begin{multline*}
  h_{r_n}(x,z) = n \int_{\R^d} \left( \int_{\R^{d_z}} p_{r_n}(x,z,x',z')^2 \,\text{d}\Phi_z(z' \mid x') \right)^{1/2} f(x') \,\text{d}x' \quad\text{and}\\ h^*(x,z) = \kappa \bar{f}\, \int_{\R^d} \left( \int_{\R^{d_z}} p_1(x,z;x',z')^2 \,\text{d}\Phi^*(z') \right)^{1/2} \,\text{d}x'.
\end{multline*}

\begin{assump}[Subcriticality]\label{dfrag}
  (a) There exists a measure $\Phi^*$ on $\R^{d_z}$ such that for all $(x,z) \in \mathcal{T}$ and $n\in\mathbb{N}$,
  \begin{equation*}
    n\int_{\R^d} \int_{\R^{d_z}} p_{r_n}(x,z;x',z') \,\text{d}\Phi_z(z' \mid x') f(x') \,\text{d}x' \leq \kappa \bar{f} \int_{\R^d} \int_{\R^{d_z}} p_1(x,z;x',z') \,\text{d}\Phi^*(z') \,\text{d}x',
  \end{equation*}
  
  \noindent and $h_{r_n}(x,z) \leq h^*(x,z)$. (b) $\norm{h^*}_\mathbf{m} < 1$.
\end{assump}

\noindent Part (a) is a regularity condition corresponding to Assumption 2 of \cite{leung2019compute} that replaces $\Phi_z(\cdot \mid x')$ with a ``dominating'' measure $\Phi^*$. The substantive requirement is (b), which is a slightly stronger version of Assumption 6 of \cite{leung2016}. Its formulation is as primitive as possible at this level of generality, but we may obtain further insight by specializing to a particular joint surplus function. The next example demonstrates how the condition restricts the magnitude of strategic interactions.

\begin{example}\label{esubcrit}
  Consider \autoref{elinear}, and suppose $X_i \sim \mathcal{U}([0,1]^2)$ and $\zeta_{ij} \sim \N(0,1)$. By a change of variables $x' \mapsto x+r_n(x'-x)$ as in \eqref{sparsecalc},
  \begin{multline*}
    n\int_{\R^d} \int_{\R^{d_z}} p_{r_n}(x,z;x',z') \,\text{d}\Phi_z(z' \mid x') f(x') \,\text{d}x' \\ = n \int_{\R^d} \prob\big( -\theta_2 < \theta_1 - \theta_3 (r_n^{-1}\norm{x-x'})^2 + \zeta_{12} \leq 0 \big) \ind\{x' \in [0,1]^2\} \,\text{d}x' \\ \leq \kappa \int_{\R^d} \prob\big( -\theta_2 < \theta_1 - \theta_3 \norm{x-x'}^2 + \zeta_{12} \leq 0 \big) \,\text{d}x'.
  \end{multline*}

  \noindent Notice the second line equals $h_{r_n}(x,z)$ and the third line $h^*(x,z)$. Letting $\Phi(\cdot)$ denote the normal CDF,
  \begin{align*}
    h^*(x,z) &= \kappa \int_{\R^2} \left[ \Phi(\theta_1+\theta_2-\theta_3\norm{u}^2) - \Phi(\theta_1-\theta_3\norm{u}^2) \right] \,\text{d}u \\
	       &= 2\kappa\pi \int_0^\infty w\left[ \Phi(\theta_1+\theta_2 - \theta_3w^2) - \Phi(\theta_1-\theta_3w^2) \right] \,\text{d}w
  \end{align*}

  \noindent by a change of variables to polar coordinates, where $\pi$ is the transcendental constant. Let $q(x) = x\Phi(x)+\phi(x)$, where $\phi(\cdot)$ is the normal PDF. The integral has a closed-form expression
  \begin{equation*}
    h^*(x,z) = \frac{\kappa\pi}{\theta_3} \big( q(\theta_1+\theta_2) - q(\theta_1) \big).
  \end{equation*}

  \noindent \autoref{dfrag} requires this to be less than one, which holds if $\theta_2$ is sufficiently small in magnitude. This is transparently a restriction on the strength of strategic interactions. 
\end{example}

The example illustrates how $\norm{h^*}_\mathbf{m} < 1$ {\em constitutes the network formation analog of well-known weak dependence conditions for linear spatial or temporal autoregressive models}, which require the magnitude of the autoregressive coefficient to be bounded below one. For instance, in the linear-in-means model of peer effects, it is assumed that the endogenous peer effect satisfies $\abs{\beta}<1$ \citep{bramoulle2009identification}. For nonlinear models such as ours, the analogous condition is necessarily more complicated to state. Equation (9) of \cite{de_jong_dynamic_2011} states the condition for dynamic binary choice time series models, which shares some visual similarities with our assumption.

The connection between \autoref{dfrag} and the discussion in \autoref{sdepstr} regarding the size of $C_i$ is that $\eqref{L_r} = \E[D_{ij} \mid X_i, Z_i, X_j, Z_j]$, so by Jensen's inequality and the change of variables argument in \eqref{sparsecalc}, $h^*(x,z) \geq \E[\sum_j D_{ij} \mid X_i=x, Z_i=z]$. The expectation of the right-hand side is the expected degree in $\bm{D}$, so by \autoref{dfrag}, this is upper bounded by one in expectation. This implies that, as we explore $C_i$ by branching from $i$ to its neighbors in $\bm{D}$, neighbors of neighbors, and so on, each node has fewer than one neighbor on average. That is, a node at each step is typically replaced by fewer than one node in the next step, so the process is below the replacement rate, and $\abs{C_i}$ should be asymptotically bounded.

\subsection{Coordination}\label{scoord}

The second weak dependence condition restricts the selection mechanism, ruling out coordination of the type described in \autoref{scoord0}. It requires selection to be ``decentralized'' in that each strategic neighborhood $C^+$ selects its pairwise stable subnetwork based on the types and random-utility shocks of nodes in $C^+$ alone. In this sense, there is no coordination across disjoint strategic neighborhoods.

Before stating the assumption, we need to clarify why such a restriction is possible. Initially, it may appear incoherent because the pairwise stability of a subnetwork on $H \subseteq \N_n$ can depend on the state of the network outside of $H$. If nodes $\N_n\backslash H$ are removed from the model, we would expect nodes in $H$ to adjust their links in best response. This suggests that coordination across subsets of nodes is generally unavoidable. 

However, if $H$ is specifically a strategic neighborhood, then under \autoref{S}, its subnetwork is in fact pairwise stable regardless of the state of the ambient network. Consider \autoref{fsn} and a counterfactual intervention that removes all nodes from the network outside of $C_5^+ = \{2, \ldots, 7\}$. The intervention has no impact on the pairwise stability of $A_{23}$ and $A_{67}$ since these are robust links. It also has no impact on the pairwise stability of non-robust potential links between nodes in $C_5^+$, for example $A_{45}$ since neither nodes 4 nor 5 are linked to $\N_n\backslash C_5^+$ in $\bm{A}$ and strategic interactions are local by \autoref{S}. That is, the non-robust potential links in $C_5^+$ are separated from $\N_n\backslash C_5^+$ by a ``buffer'' of robust links, so the intervention does not trigger a chain of best-responses that affects $\bm{A}_{C_5^+}$. This illustrates the following result.

\begin{proposition}[\cite{leung2016}, Proposition 1]\label{psac}
  Under \autoref{S}, for any $r>0$, $n\in\mathbb{N}$, and strategic neighborhood $C^+$,
  \begin{equation*}
    \mathcal{E}(r^{-1}\bm{T}_{C^+}, \Zeta_{C^+}) = \{\bm{A}_{C^+}: \bm{A} \in \mathcal{E}(r^{-1}\bm{T}_n,\Zeta_n)\}.
  \end{equation*}
\end{proposition}

\noindent The left-hand side is the set of pairwise stable networks under the counterfactual model in which the set of nodes is $C^+$, rather than $\N_n$. The right-hand side takes the set of pairwise stable networks on $\N_n$ and restricts them to $C^+$. The proposition asserts the two sets are equivalent, which says that the pairwise stability of subnetworks on $C^+$ only depends on the structural primitives of nodes in $C^+$. We emphasize that this property is unique to strategic neighborhoods and is not true for arbitrary subsets of nodes. \cite{leung2019compute} exploits \autoref{psac} to devise an algorithm that computes $\mathcal{E}(r_n^{-1}\bm{T}_n,\Zeta_n)$ in polynomial time under \autoref{dfrag}. We utilize it to establish a CLT.

Recall that our objective is to impose the assumption that strategic neighborhoods ``separately select'' their own pairwise stable subnetworks. \autoref{psac} ensures that it is coherent to refer to the pairwise stability of a subnetwork on a strategic neighborhood in isolation from the rest of the network. However, a second concern is that strategic neighborhoods do not necessarily partition $\N_n$, so such an assumption may still appear incoherent. For instance, in \autoref{fsn}, the left and middle strategic neighborhoods share nodes 2 and 3 in common, so it is not clear how the two neighborhoods can separately select equilibria.

Recall that components of $\bm{D}$ do partition $\N_n$, while strategic neighborhoods are obtained by adding nodes that are robustly linked to components. Then {\em the link between any node pair that lies in multiple strategic neighborhoods must necessarily be robust and therefore have the same realization under any pairwise stable equilibrium.} That is, any pairwise stable subnetwork on the left strategic neighborhood in \autoref{fsn} sets $A_{23}=1$, as does any pairwise stable subnetwork on the middle strategic neighborhood. Consequently, it is a logically coherent operation to select an equilibrium subnetwork for each strategic neighborhood and then take a ``union'' to obtain the overall network $\bm{A}$ \cite[see \S SA.4.2 of][for a detailed elaboration of this idea]{leung2019compute}.

We are now prepared to state the assumption. Let $\lambda_n(r_n^{-1}\bm{T}_n, \Zeta_n)\big|_H$ be the restriction of the range of $\lambda_n(\cdot)$ to subnetworks on $H$. For example, under \autoref{aesel}, $\lambda_n(r_n^{-1}\bm{T}_n, \Zeta_n) = \bm{A}$, so $\lambda_n(r_n^{-1}\bm{T}_n, \Zeta_n)\big|_H = \bm{A}_H$.

\begin{assump}[Decentralized Selection]\label{nocoord} 
  For any $r>0$, $n\in\mathbb{N}$, and strategic neighborhood $C^+$ constructed under the structural primitives $(r^{-1}\bm{T}_n, \Zeta_n)$, we have $\lambda_n(r^{-1}\bm{T}_n, \Zeta_n)\big|_{C^+} = \lambda_{\abs{C^+}}(r^{-1}\bm{T}_{C^+}, \Zeta_{C^+})$. 
\end{assump}

\noindent This corresponds to Assumption 7 of \cite{leung2016}. It is important to understand the difference between the left- and right-hand sides of the equality. On the left, we have the model involving all nodes $\N_n$; the selection mechanism produces a network $\bm{A}$, and we take its subnetwork $\bm{A}_{C^+}$. On the right, we have the counterfactual model involving only nodes in $C^+$, and the selection mechanism $\lambda_{\abs{C^+}}(\cdot)$ produces a network $\lambda_{\abs{C^+}}(r^{-1}\bm{T}_{C^+}, \Zeta_{C^+})$. The assumption asserts that the two outputs are the same.

\cite{sheng2014} does not require this assumption because she considers a setting with many small independent networks, which means equilibrium selection is necessarily independent across network observations. \autoref{nocoord} is the single-network analog of this requirement, requiring selection to operate separately across latent strategic neighborhoods. 

The assumption rules out selection mechanisms in which all nodes coordinate through a common signal, such as a single node's type. Coordination is only allowed to occur within strategic neighborhoods. In the special case where there exists a unique equilibrium on $\N_n$, for instance if there are no strategic interactions, the assumption holds trivially. More generally, the condition is satisfied by variants of myopic best-response dynamics, which are widely used in the theoretical and econometric literature on dynamic network formation \citep[e.g.][]{jackson2010,mele2017}. 

\begin{example}[Best-Response Dynamics]\label{embrd}
  An example of myopic best-response dynamics is the following. Arbitrarily order all node pairs and begin at an arbitrary network $\bm{A}_0$. At step $t$, update the previous network $\bm{A}_{t-1}$ by setting the $ij$th component to $\bm{1}\{V(r_n^{-1}\norm{X_i-X_j}, S_{ij}^{t-1}, Z_i, Z_j, \zeta_{ij}) > 0\}$, where $S_{ij}^{t-1}=S_n(i,j,r_n^{-1}\bm{T}_n,\bm{A}_{t-1})$, the network statistics evaluated at the prior network. Repeat for all pairs of nodes to obtain $\bm{A}_t$. Repeat this process until convergence to a network $\bm{A} \in \mathcal{E}(r_n^{-1}\bm{T}_n, \Zeta_n)$.\footnote{The usual method for proving the existence of a pairwise stable network is to establish non-existence of ``closed cycles,'' which implies that myopic best-response dynamics always converge to an equilibrium. See for example Proposition 2.1 of \cite{sheng2014}.} This constitutes a selection mechanism $\lambda_n(\cdot)$ since it maps structural primitives to an equilibrium network. To ensure that $\lambda_n(\cdot)$ is permutation-equivariant (\autoref{aesel}), we may suppose that potential links in $\bm{A}_0$ and the ordering of node pairs are only functions of the types of the nodes in the pair.\footnote{More generally, the initial network and ordering of node pairs may be determined by any permutation-equivariant functions of $(r_n^{-1}\bm{T}_n, \Zeta_n)$.}

  To understand how this satisfies \autoref{nocoord}, first consider the ideal scenario in which the data consists of two independent network observations formed by these dynamics. The equilibrium on each network may be generated in two equivalent ways. First, the dynamics may be run separately on each network. Second, the two networks may be concatenated into one, defining the joint surplus between pairs of nodes in different networks as $-\infty$, and the dynamics may be run on the entire entity. These produce the same output given the initial network and node pair ordering.

  Now suppose the data consists of a single network comprised of two strategic neighborhoods. If the neighborhoods do not share a pair of nodes in common, then this is the same situation as the two-network case. If they do, then both ways of running the dynamics still produce the same result because any pair of nodes shared by both neighborhoods must form a robust link. This reasoning immediately extends to an arbitrary number of strategic neighborhoods, so \autoref{nocoord} holds.
\end{example}

\subsection{Main Results}\label{smainres}

The last assumption we require is a regularity condition.

\begin{assump}[Regularity]\label{dreg}
  Either $p_r(X_1,Z_1;X_2,Z_2)=0$ a.s.\ for any $r>0$ in a neighborhood of zero, or $\inf\{ \liminf_{n\rightarrow\infty} n\, \E[p_{r_n}(x,z;X_2,Z_2)]\colon (x,z) \in \mathcal{T}\} > 0$.
\end{assump}

\noindent The case $p_r(X_1,Z_1;X_2,Z_2)=0$ corresponds to a model without strategic interactions, which is only mentioned for completeness. In the more interesting case, the assumption essentially requires that strategic interactions are sufficiently nontrivial for all nodes in that $p_{r_n}(X_1,Z_1;X_2,Z_2)$ is at least order $n^{-1}$. This is a mild requirement that is typically satisfied because the two probabilities in \eqref{L_r} are upper and lower bounds on the probability of link formation, and the upper bound is order $n^{-1}$ under sparsity by \eqref{sparsecalc}.

\begin{example}\label{edreg}
  Consider \autoref{esubcrit}. Following the derivation there,
  \begin{equation*}
    n\, \E[p_{r}(x,z;X_2,Z_2)] = 2nr^2 \pi \int_0^\infty w \left[ \Phi(\theta_1+\theta_2 - \theta_3 w^2) - \Phi(\theta_1 - \theta_3 w^2) \right] \,\text{d}w.
  \end{equation*}

  \noindent If $\theta_2=0$, this corresponds to the case $p_r(X_1,Z_1;X_2,Z_2)=0$ a.s. Otherwise, \autoref{dreg} holds if $\theta_2 \neq 0$ since $nr_n^2 = \kappa > 0$ by \eqref{r_n}.
\end{example}

\begin{theorem}\label{primclt}
  Assumptions \ref{Vc}--\ref{dyklocal} and \ref{dfrag}--\ref{dreg} imply \autoref{main-bp-exp-stab}.
\end{theorem}
\begin{proof}
  See \autoref{sprimclt}. We sketch the proof below.
\end{proof}

\noindent From Theorems \ref{master-clt} and \ref{primclt} we immediately obtain the following CLT for strategic network formation.

\begin{cor}\label{cprim0}
  Under Assumptions \ref{Vc}--\ref{dyklocal} and \ref{main-bp-mom}--\ref{dreg}, $\sup_n \norm{\bm{\Sigma}_n}_\infty < \infty$, and if additionally \\ $\liminf_{n\rightarrow\infty} \lambda_\text{min}(\bm{\Sigma}_n) > 0$, then \eqref{thenormal} holds.
\end{cor}

\noindent At this level of generality, these conditions are close to as primitive as possible, but with additional structure we can derive lower-level conditions. 

\begin{cor}\label{cprim}
  Consider the network formation model in \autoref{esubcrit}, and assume \autoref{main-bp-mom}(b) and the following weak dependence conditions hold.
  \begin{enumerate}[(a)]
    \item (Strength of interactions) For $q(x)$ defined in the example,
      \begin{equation*}
	\frac{\kappa\pi}{\theta_3} \big( q(\theta_1+\theta_2) - q(\theta_1) \big) < 1.
      \end{equation*}

    \item (Decentralized selection) The equilibrium is selected via myopic best-response dynamics as in \autoref{embrd}.
  \end{enumerate}

  \noindent If $\sum_{i=1}^n \psi_i(\N_n)$ is a vector of connected subnetwork counts (\autoref{esubcts}), then \eqref{thenormal} holds.
\end{cor}

\begin{proof}
  We verify the conditions of \autoref{cprim0}. \autoref{Vc} holds because
  \begin{equation*}
    \bar{\Phi}_\zeta(\bar{V}^{-1}(\delta,0)) = \bar{\Phi}_\zeta(\theta_3\delta^2 - \theta_1 - \max\{\theta_2,0\}),
  \end{equation*}

  \noindent which decays to zero exponentially with $\delta$ since $\bar{\Phi}_\zeta$ is the complementary CDF of the standard normal distribution. \autoref{S} holds by choice of $S_{ij}$. Assumptions \ref{aesel} and \ref{nocoord} hold by (b). \autoref{dyklocal} holds because we consider connected subnetwork counts (see \autoref{eshenginf}). \autoref{main-bp-mom}(a) follows from \autoref{lbdmoms}. \autoref{dfrag} follows from (a) (see \autoref{esubcrit}). Finally, \autoref{dreg} holds, as shown in \autoref{edreg}. 
\end{proof}

\subsection{Method of Proof}\label{smeth2}

We present the method of proof for \autoref{primclt} for the case of the average degree, whose node statistic is $\psi_i(\N_n) = \sum_{j=1}^n A_{ij}$. The approach can be applied to other network models with strategic interactions.\footnote{A previous version of the paper presents applications to dynamic network formation \citep{leung2019normal}, and \cite{leung2019inference} applies the methodology to games on networks.} We construct an upper bound on the radius of stabilization that has a distribution with sufficiently thin tails. 

\bigskip
\noindent {\bf Step 1.} Recalling \eqref{psiH}, we construct a set of nodes $J_i \subseteq \N_n$ positioned near $i$ such that 
\begin{equation}
  \psi_i(\N_n) = \psi_i(J_i). \label{Jigoal}
\end{equation}

\noindent That is, $i$'s node statistic is invariant to the counterfactual removal of nodes outside of $J_i$. The challenge is to find a set that is relatively small so that the distribution of the set's size has thin tails. Supposing such a set could be found, the radius of stabilization $\RR_i(n,r_n)$ (\autoref{rosdef}) would be upper bounded by the radius of the smallest ball centered at $r_n^{-1}X_i$ containing the positions of $J_i$ (plus one):
\begin{equation}
  \RR_i(n,r_n) \leq \tilde\RR_i(n,r_n) \equiv \max_{j \in J_i} r_n^{-1}\norm{X_i-X_j} + 1. \label{ros}
\end{equation}

Since $\psi_i(\N_n) = \sum_{j=1}^n A_{ij}$ is simply the 1-neighborhood size, an initial guess for $J_i$ might be $i$'s 1-neighborhood. However, this does not generally satisfy \eqref{Jigoal}. To see why, consider node $i=6$ in \autoref{fsn} whose 1-neighborhood is contained in $\{4,5,7\}$. Since $D_{56}=1$, the potential link $A_{56}$ is non-robust and therefore may differ under the counterfactual that removes node 3 from the model, even though 3 is not $i$'s neighbor.

Instead we take $J_i = C_i^+$. The pairwise stability of $\bm{A}_{C_i^+}$ is invariant to the removal of nodes outside of $C_i^+$ (see \autoref{exstrngb} and \autoref{psac}), and under \autoref{nocoord}, $\bm{A}_{C_i^+}$ remains the selected equilibrium subnetwork on $C_i^+$ after removal of $\N_n\backslash C_i^+$. That is, for $H = C_i^+$, $\bm{A}_H = \lambda_n(r_n^{-1}\bm{T}_n, \Zeta_n)|_H = \lambda_{\abs{H}}(r_n^{-1}\bm{T}_H, \Zeta_H)$. Since $C_i^+$ contains $i$'s 1-neighborhood, \eqref{Jigoal} holds.

\bigskip
\noindent {\bf Step 2.} Next we show that $\abs{J_i}$ has a distribution with exponential tails. First, \autoref{introbp} establishes that $\abs{J_i}$ is stochastically dominated by the size of a certain branching process when $J_i$ is constructed from components and $K$-neighborhoods of dyadic networks. \autoref{Jtails} shows that the size of the branching process has a distribution with exponential tails using Assumptions \ref{dfrag} and \ref{dreg}.

In the case of average degree, our choice of $J_i=C_i^+$ is the union of a component of $\bm{D}$ and 1-neighborhoods in $\bm{\Pi}$, as defined in \autoref{scltstatic}, which are dyadic networks. For more general $K$-neighborhood node statistics satisfying \autoref{dyklocal}, the construction of $J_i$ is more complicated but still involves similar objects. \autoref{constructJT} in \autoref{sexpstab} shows how to construct $J_i$ for a general node statistic satisfying \autoref{dyklocal}. The lemmas in \autoref{sexpstab} can then be applied to obtain the desired tail bounds.

\bigskip
\noindent {\bf Step 3.} We translate the tail bound for $\abs{J_i}$ into one for $\tilde\RR_i(n,r_n)$. Intuitively, if $\abs{J_i}$ is small, then so will be $\tilde\RR_i(n,r_n)$ since nodes are homophilous in positions (\autoref{Vc}), so each $j\in J_i$ will typically be close to $i$ in terms of distance $r_n^{-1}\norm{X_i-X_j}$. \autoref{dyDC} provides the formal argument.

\section{Applications to Inference}\label{sinfer}

We discuss two inference procedures applicable to network data generated by the model in \autoref{smodel}. Our results provide the first formal justification for their use in the subsequent applications.

Since our objective is to establish a CLT, thus far we have made no distinction between what aspects of the model are known or observed by the econometrician. In what follows, the only requirement is that the econometrician must be able to compute the relevant network moments, but it is application-specific which structural primitives need to be observed for this to be possible. 

In the network statistics applications, only $\bm{A}$ needs to be observed. In structural applications, $V(\cdot)$ is typically known up to a vector of parameters, $\bm{A}$ is observed, and $\zeta_{ij}$ is unobserved, but any subvector of $(X_i, Z_i)$ could potentially be unobserved. In particular, the procedures can be implemented when positions are unobserved, as in the literature on latent-space models. Typically the distribution of the unobserved component of types, conditional on observables, is assumed known up to a vector of parameters.

\subsection{Network Statistics}\label{snetstats}

Define $\mu_0 = \E[\psi_1(\N_n)]$ and $\hat\mu = n^{-1} \sum_{i=1}^n \psi_i(\N_n)$. Consider testing the null hypothesis 
\begin{equation*}
  H_0\colon \mu_0=\mu.
\end{equation*}

\noindent This is relevant for the reporting of stylized facts in the networks literature. Such facts are obtained by computing various statistics from $\bm{A}$ but are seldom accompanied by formal uncertainty quantification due to a lack of available methods. \cite{leung2018dependence} discusses two particular examples: testing for nontrivial clustering (his \S3.3) and testing for a power law degree distribution (his \S3.4). We next discuss two generic tests justified by our CLT. Confidence regions for $\mu_0$ can be obtained by test inversion. 

\bigskip
\noindent {\bf Single large network.} If the sample consists of a single network, we may apply the resampling procedure due to \cite{song_ordering-free_2016} and \cite{leung2018dependence}. Let $\alpha$ be the desired level of the test, $R_n = (n/2)^{4/3}$ rounded to the nearest integer, $\Pi$ be the set of all bijections (permutation functions) on $\{1, \dots, n\}$, and $\pi = (\pi_r)_{r=1}^{R_n}$ be i.i.d.\ uniform draws from $\Pi$. Let $\hat{\bm{V}} = n^{-1} \sum_{i=1}^n (\psi_i(\N_n) - \hat\mu)(\psi_i(\N_n) - \hat\mu)'$, the sample variance. Define the test statistic
\begin{equation*}
  T_U(\mu;\pi) = \frac{1}{\sqrt{d_\psi R_n}} \sum_{r=1}^{R_n} (\psi_{\pi_r(1)}(\N_n) - \mu)' \hat{\bm{V}}^{-1} (\psi_{\pi_r(2)}(\N_n) - \mu),
\end{equation*}

\noindent recalling that $d_\psi$ is the dimension of the range of $\psi_n(\cdot)$. Let $z_{1-\alpha}$ be the $1-\alpha$ quantile of the standard normal distribution. The test rejects if and only if 
\begin{equation}
  T_U(\mu;\pi) > z_{1-\alpha}. \label{deprob}
\end{equation}

Validity of this test hinges on the high-level weak dependence condition that $\hat\mu$ is $\sqrt{n}$-consistent. \cite{leung2016} provides an LLN for $\hat\mu$ but not a rate of convergence. Our paper is the first to provide primitive conditions for $\sqrt{n}$-consistency in the context of strategic network formation.

For intuition on the importance of $\sqrt{n}$-consistency for the test's validity, consider a simpler test statistic studied by \cite{leung2018dependence}:
\begin{equation*}
  \tilde{T}_M(\mu;\pi) = \frac{1}{\sqrt{R_n^M}} \sum_{r=1}^{R_n^M} \hat{\bm{V}}^{-1/2} \big(\psi_{\pi_r(1)}(\N_n)-\mu\big)
\end{equation*}

\noindent for $R_n^M = \sqrt{n}$ rounded to the nearest integer. To understand its asymptotic behavior, we add and subtract its mean conditional on the data:
\begin{multline*}
  \frac{1}{\sqrt{R_n^M}} \sum_{r=1}^{R_n^M} \hat{\bm{V}}^{-1/2} \big(\psi_{\pi_r(1)}(\N_n)-\E[\psi_{\pi_r(1)}(\N_n) \mid \{\psi_i(\N_n)\}_{i=1}^n]\big) \\ + \frac{1}{\sqrt{R_n^M}} \sum_{r=1}^{R_n^M} \hat{\bm{V}}^{-1/2} \big(\E[\psi_{\pi_r(1)}(\N_n) \mid \{\psi_i(\N_n)\}_{i=1}^n] - \mu\big).
\end{multline*}

\noindent Conditional on the data $\{\psi_i(\N_n)\}_{i=1}^n$, the permutations $\pi_r$ are independent, so the first term is an average of $R_n^M$ conditionally independent observations and can be shown to be asymptotically $\N(\bm{0},\bm{I})$ under $H_0$. The second is a bias term that can be shown to equal $(R_n^M/n)^{1/2} \hat{\bm{V}}^{-1/2} \sqrt{n}(\hat\mu-\mu)$. Hence, under $H_0$ and $\sqrt{n}$-consistency, the bias is order $(R_n^M/n)^{1/2} = o(1)$.

\bigskip
\noindent {\bf Multiple large networks.} A drawback of the previous procedure is that it is inefficient, having a rate of convergence slower than $\sqrt{n}$ since $R_n^M = o(n)$. If the sample consists of sufficiently many independent large networks, then more powerful methods are available from the cluster-robust inference literature. Consider a sequence of $L$ independent networks indexed by $n$, where $L$ is fixed with respect to $n$ and each network $\ell=1,\dots,L$ has size $n_\ell$ satisfying $n_\ell/n \rightarrow c_\ell \in (0,\infty)$ as $n\rightarrow\infty$. Let $\hat{\mu}^\ell = n_\ell^{-1} \sum_{i=1}^{n_\ell} \psi_i(\N_{n_\ell})$, the network moment computed on network $\ell$. Assume there exists a universal population moment $\mu_0$ such that $\E[\hat{\mu}^\ell] = \mu_0 + o(n^{-1/2})$ for all $\ell$. We seek to test the null $H_0\colon \mu_0=\mu$.

We consider the randomization test proposed by \cite{canay2017randomization}. For $S_{n,\ell} = \sqrt{n_\ell}(\hat{\mu}^\ell - \mu)$ and $S_n = (S_{n,\ell})_{\ell=1}^L$, define the Wald statistic
\begin{equation*}
  T(S_n) = \left( \frac{1}{\sqrt{L}} \sum_{\ell=1}^L S_{n,\ell}' \right) \left( \frac{1}{L} \sum_{\ell=1}^L S_{n,\ell}S_{n,\ell}' \right)^{-1} \left( \frac{1}{\sqrt{L}} \sum_{\ell=1}^L S_{n,\ell} \right).
\end{equation*}

\noindent We obtain critical values from the randomization distribution $\{T(\pi S_n)\colon \pi \in \{-1,1\}^L\}$ where $\pi S_n = ( \pi_\ell S_{n,\ell} )_{\ell=1}^L$ for $\pi = (\pi_\ell)_{\ell=1}^L$. Let $\alpha$ be the desired level of the test, $q = 2^L(1-\alpha)$ rounded up to the nearest integer, and $c_{L,1-\alpha}$ be the $q$th largest value of $\{T(\pi S_n)\colon \pi \in \{-1,1\}^L\}$. The test rejects if and only if 
\begin{equation*}
  T(S_n) > c_{L,1-\alpha}. 
\end{equation*}

The test is asymptotically level $\alpha$ under the high-level condition that the limit distribution of the vector of network moments is asymptotically normal. We provide the first primitive sufficient conditions in the literature for strategic network formation. Intuitively, under $H_0$ and asymptotic normality, $\{S_{n,\ell}\}_{\ell=1}^L$ are independent draws from a mean-zero, approximately normal distribution. Hence, multiplying these draws by $\pm 1$ does not change the asymptotic distribution of $T(S_n)$, which is the key justification for the validity of randomization tests.

\subsection{Structural Inference}\label{ssturinf}

We revisit \autoref{eshenginf}, which concerns inference on structural parameters using moment inequalities proposed by \cite{sheng2014}. Recalling the setup there, let $\theta_0$ denote the true parameters of $V(\cdot)$. To test the hypothesis $H_0\colon \theta_0=\theta$, we test the moment inequality
\begin{equation*}
  \mu_0 \equiv \E[\psi_1(\N_n)] \leq \zero \quad \text{where}\quad \psi_{i_1}(\N_n) = \sum_{i_2=1}^n \dots \sum_{i_m=1}^n (G_{i_1,\ldots,i_m} - H_{i_1,\ldots,i_m}(\theta_0)) q_{i_1,\ldots,i_m}(\bm{T}_n),
\end{equation*}

When the data consists of a single large network, we can employ a test due to \cite{leung2018dependence}. Let $\mu_k$ be the $k$th component of $\mu_0$, $\psi_{ik}(\N_n)$ the $k$th component of $\psi_i(\N_n)$, and $T_{U,k}(\mu_k; \pi)$ the U-type statistic defined in the previous subsection but computed with scalar data $\{\psi_{ik}(\N_n)\}_{i=1}^n$. Also let $\hat\mu_k$ be the $k$th component of $\hat\mu$ and $\hat{\bm{V}}_{kk}$ the $k$th diagonal entry of $\hat{\bm{V}}$. Define the test statistic
\begin{align*}
  &Q_n(\pi) = \max_{1\leq k\leq d_\psi} \big\{ T_{U,k}(0; \pi) - \mu_k^*\ind\{\hat\mu_k < 0\} \big\} \quad\text{where} \\
  &\mu_k^* = \hat\mu_k \hat{\bm{V}}_{kk}^{-1} \frac{1}{\sqrt{d_\psi R_n}} \sum_{r=1}^{R_n} \left( \psi_{\pi_r(1),k}(\N_n) + \psi_{\pi_r(2),k}(\N_n) \right) - \sqrt{\frac{R_n}{d_\psi}} \hat{\bm{V}}_{kk}^{-1} \hat\mu_k^2.
\end{align*}

\noindent Let $\tilde\pi_1, \dots, \tilde\pi_L$ be i.i.d.\ with the same distribution as $\pi$. Let $q = L(1-\alpha)$ rounded up to the nearest integer and $c_{L,1-\alpha}$ be $q$th largest value of $\{\max_{1\leq k\leq d_\psi} T_{U,k}(\hat\mu_k;\tilde\pi_\ell)\}_{\ell=1}^L$. The test rejects if and only if 
\begin{equation*}
  Q_n(\pi) > c_{L,1-\alpha}. 
\end{equation*}

Theorem 2 of \cite{leung2018dependence} provides conditions under which the test is asymptotically level $\alpha$ under $H_0$. The main assumption that needs to be verified is $\sqrt{n}$-consistency of $\hat\mu$, which is a consequence of our CLT. Like the test in the previous subsection, this procedure is inefficient with a slower than $\sqrt{n}$-rate of convergence. To construct more powerful tests, we require either a consistent estimate of $\bm{\Sigma}_n$ or a valid resampling procedure, topics we leave to future research.

\section{Simulation Study}\label{smc}

We conduct a simulation study to assess the quality of the normal approximation and finite-sample performance of the inference procedures in \autoref{sinfer}. We simulate data according to the model in \autoref{elinear} with $\lVert \cdot\rVert$ equal to the Euclidean norm, $\theta = (1, 0.25, 1)$, $X_i \stackrel{iid}\sim \mathcal{U}([0,1]^2)$, and $\zeta_{ij} \stackrel{iid}\sim \mathcal{N}(0,1)$. To satisfy \autoref{dfrag}, we set $\kappa = 0.8\theta_3\pi^{-1}(q(\theta_1+\theta_2) - q(\theta_1))^{-1}$ following the notation in \autoref{esubcrit}, which implies $\norm{h^*}_{\mathbf{m}} = 0.8 < 1$. To satisfy \autoref{nocoord}, we select the equilibrium network using myopic best-response dynamics (\autoref{embrd}) starting from the network $\bm{\Pi}$ defined in \autoref{sstrnbh}. The model generates networks with fairly realistic properties; the largest component comprises about 75 percent of the network, while the average degree and clustering are respectively about 4.3 and 0.4.

We consider equality tests involving the following network moment. Define the node statistic $\psi_i(\N_n) = \psi_i^c - \psi_i^f$ where 
\begin{equation*}
  \psi_i^c = \frac{\sum_{j\neq k\neq i} A_{ij}A_{jk}A_{ik}}{\sum_{j\neq k\neq i} A_{ij} A_{ik}} \quad\text{and}\quad \psi_i^f = \frac{2}{n-1} \sum_{j\neq i} A_{ij}.
\end{equation*}

\noindent Then $n^{-1} \sum_{i=1}^n \psi_i^c$ is the average clustering coefficient, a measure of triadic closure, while $n^{-1} \sum_{i=1}^n \psi_i^f$ is the link frequency. The motivation for the network moment $n^{-1} \sum_{i=1}^n \psi_i(\N_n)$ is as follows. Under the ``null'' Erd\H{o}s-R\'{e}nyi model, the moment is approximately zero. However, a well-known stylized fact is that real-world networks typically feature nontrivial clustering in that the statistic is far from zero \citep{barabasi2015}. While this fact is based only on the point estimate, the inference procedures in \autoref{sinfer} enable us to formulate this as a statistical test.

We simulate rejection rates for three different tests of hypotheses of the form
\begin{equation*}
  H_0\colon \E[\psi_1] = c\mu \quad\text{against}\quad H_1\colon \E[\psi_1] \neq c\mu
\end{equation*}

\noindent where $c \in (0,1]$ and $\mu$ is the true value of $\E[\psi_1]$, computed using 40k simulation draws. To simulate size, we set $c=1$, and to simulate power, we set $c \in \{0.8, 0.9\}$. We consider the following tests at the 5-percent level. To assess the quality of the normal approximation, we use the ``oracle'' $t$-test with test statistic $(n^{-1} \sum_{i=1}^n \psi_i - c\mu) / \sigma$, where $\sigma$ is the true standard error, computed by taking the standard deviation of the network moment across 40k simulation draws. We then consider the dependence-robust test \eqref{deprob} and randomization test from \autoref{snetstats}. The randomization test will use data from 5, 6, or 8 independent networks of identical size, while the other tests only utilize data from one of the networks. We simulate size and power using 5k simulation draws.

\begin{table}[ht]
\centering
\caption{Simulation results}
\begin{threeparttable}
\begin{tabular}{lrrrrrrrrr}
\toprule
 & \multicolumn{3}{c}{Size ($c=1$)} & \multicolumn{3}{c}{Power ($c=0.9$)} & \multicolumn{3}{c}{Power ($c=0.8$)} \\
$n$ & 250 & 500 & 1000 & 250 & 500 & 1000 & 250 & 500 & 1000 \\
\midrule
Oracle & 0.050 & 0.054 & 0.048 & 0.337 & 0.627 & 0.915 & 0.856 & 0.995 & 1.000 \\
DR & 0.069 & 0.057 & 0.058 & 0.176 & 0.263 & 0.410 & 0.530 & 0.823 & 0.981 \\
Rand (8) & 0.040 & 0.040 & 0.038 & 0.928 & 0.998 & 1.000 & 1.000 & 1.000 & 1.000 \\
Rand (6) & 0.028 & 0.029 & 0.030 & 0.672 & 0.929 & 0.997 & 0.993 & 1.000 & 1.000 \\
Rand (5) & 0.000 & 0.000 & 0.000 & 0.000 & 0.000 & 0.000 & 0.000 & 0.000 & 0.000 \\
$\mu$ & 0.393 & 0.409 & 0.416 & 0.393 & 0.409 & 0.416 & 0.393 & 0.409 & 0.416 \\
$\sigma$ & 0.026 & 0.018 & 0.013 & 0.026 & 0.018 & 0.013 & 0.026 & 0.018 & 0.013 \\
\bottomrule
\end{tabular}
\begin{tablenotes}[para,flushleft]
  \small 5000 simulations. DR corresponds to the dependence-robust test and Rand ($x$) to the randomization test using data from $x$ independent networks, all of size $n$. Oracle corresponds to the $t$-test using the true standard error $\sigma$.
\end{tablenotes}
\end{threeparttable}
\label{trej}
\end{table}

\autoref{trej} shows the results. From the first three columns, we see that the oracle $t$-test controls size well across all sample sizes, illustrating the quality of the normal approximation. The randomization test (the Rand rows in the table) controls size well across all $n$, while the dependence-robust test (the DR row in the table) exhibits some over-rejection, particularly in smaller samples. Section 6 of \cite{leung2018dependence} provides additional simulation results for the dependence-robust test under a variety of data-generating processes, including a network formation model similar to the present design. His results show that the test has good properties for different choices of the tuning parameter $R_n$.

The last six columns of \autoref{trej} compare power. The differences between the oracle and dependence-robust tests are due to the latter's slower rate of convergence. A comparison of the power of the randomization and dependence-robust tests is not apples-to-apples because the former utilizes 5 to 8 times the sample size. However, we see that the randomization test requires at least 6 networks to obtain nontrivial power. These results confirm the discussion in \autoref{sinfer}.

\section{Conclusion}\label{sconclude}

This paper develops a large-sample theory for a model of network formation with strategic interactions and homophilous agents when the data consists of a small sample of large networks or possibly a single network. We prove a general CLT under a high-level weak dependence condition and provide a general methodology for its verification. We apply the methodology to obtain primitive conditions for a CLT for moments of pairwise stable networks.

There are several important directions for future work. Concentration inequalities for stabilization would be useful, for example, for deriving lower-level conditions for uniform convergence of nonparametric or high-dimensional estimators using network data. It is also of interest to develop efficient alternatives to the inference procedures in \autoref{sinfer} in the case where the data is a single network observation.

\appendix
\numberwithin{equation}{section} 

\section{Bounded Subnetwork Moments}\label{sbdmoms}

This section formally defines subnetwork counts, introduced in \autoref{esubcts}, and verifies \autoref{main-bp-mom}(a) for these moments. Fix the subnetwork size $h\in\mathbb{N}\backslash\{1\}$, and let $a_h$ be a connected network on nodes $\{1, \dots, h\}$ (a network is connected if the path distance between any pair of nodes is finite). For $H \subset \N_n$ with $\abs{H}=h$, we say $\bm{A}_H$ is {\em isomorphic} to $a_h$ if there exists a bijection $\pi\colon \N_n\rightarrow \N_n$ such that $\pi(\bm{A})_H = a_h$, where $\pi(\bm{A})$ is the permuted adjacency matrix $(A_{\pi(i)\pi(j)})_{i,j\in\N_n}$. If $\bm{A}_H$ is isomorphic to $a_h$, we write $\bm{A}_H \cong a_h$.

Let $[n]_h$ be the set of subsets of $\N_n$ of size $h$. The subnetwork count for $a_h$ is
\begin{equation*}
  \sum_{H \in [n]_h} \bm{1}\{\bm{A}_H\cong a_h\}.
\end{equation*}

\noindent This counts ``unlabeled'' subnetworks isomorphic to $a_h$. To rewrite it as a (scaled) network moment $\sum_{i=1}^n \psi_i(\N_n)$, we observe that this is proportional to the corresponding count of ``labeled'' subnetworks. Formally, 
\begin{equation}
  h! \sum_{H \in [n]_h} \bm{1}\{\bm{A}_H\cong a_h\} = \sum_{i_1\in\N_n} \underbrace{\sum_{i_2\in\N_n} \dots \sum_{i_h\in\N_n} \bm{1}\{\bm{A}_{\{i_1,\ldots,i_h\}} \cong a_h\}}_{\psi_{i_1}(\N_n)}. \label{t3902jg}
\end{equation}

\noindent The right-hand side counts labeled subnetworks isomorphic to $a_h$, which will be our object of analysis in what follows.

Equation (5.3) of \cite{sheng2014} uses subnetwork counts to define moment inequalities. She divides the counts by $\binom{n}{h}$ since she considers a setting with many small independent networks. With a single large sparse network, the correct scaling is instead to divide by $n$, resulting in the network moment $n^{-1} \sum_{i=1}^n \psi_i(\N_n)$ for $\psi_i(\N_n)$ defined in \eqref{t3902jg}.

\begin{secprop}\label{lbdmoms}
  Under \autoref{Vc}, \autoref{main-bp-mom}(a) holds for node statistics $\psi_i(\N_n)$ of the form given in \eqref{t3902jg}.
\end{secprop}
\begin{proof}
  Per the setup of \autoref{main-bp-mom}(a), let the number of nodes be $m+k$, where $k \in \{1,2\}$ and $m$ is either a nonrandom element of $\mathbb{N}$ or equal to $N_n$. The statement of \autoref{main-bp-mom}(a) also considers models where the set of nodes is a subset $H_n \cup \{1\}$ of $\N_{m+k}$, but since this only reduces the upper bound in \eqref{Mupper} below, it is sufficient to consider $H_n = \N_{m+k}\backslash\{1\}$. Define
  \begin{equation}
    \psi_1(\N_{m+k}) = \sum_{i_1\in\N_{m+k}} \dots \sum_{i_{h-1} \in\N_{m+k}} \bm{1}\{\bm{A}_{\{1,i_1,\ldots,i_{h-1}\}} \cong a_h\}. \label{tu3902jh4w}
  \end{equation}

  \noindent Since $a_h$ is a connected network on $\{1,\dots,h\}$, any node in the network is at most path distance $h-1$ from node 1. Therefore, we can replace occurrences of $\N_{m+k}$ in \eqref{tu3902jh4w} with $\N_{\bm{A}}(1,h-1)$. With this change, \eqref{tu3902jh4w} is bounded above by $\abs{\N_{\bm{A}}(1,h-1)}^{h-1}$.

  Recall the definition of the network $\bm{M}$ from \eqref{Mnet}. Since $A_{ij} \leq M_{ij}$, $\bm{A}$ is a subnetwork of $\bm{M}$, so
  \begin{equation}
    \abs{\N_{\bm{A}}(1,h-1)}^{h-1} \leq \abs{\N_{\bm{M}}(1,h-1)}^{h-1}. \label{Mupper}
  \end{equation}

  \noindent By \autoref{Jsd}, for $m$ sufficiently large, $\abs{\N_{\bm{M}}(1,h-1)}$ is stochastically dominated by the size of a branching process $\abs{\mathfrak{X}^M_{r_n}(X_1,Z_1;h-1)}$. By \autoref{MKexptail}, the distribution of $\abs{\mathfrak{X}^M_r(x,z;h-1)}$ has exponential tails uniformly in $x,z,r$. It follows that $\abs{\N_{\bm{M}}(1,h-1)}^{h-1}$ has uniformly bounded $p$ moments for any $p>2$.
\end{proof}


\FloatBarrier
\phantomsection
\addcontentsline{toc}{section}{References}
\bibliography{clt_nf}{} 
\bibliographystyle{aer}


\newpage
\part{Supplementary Appendix}

\makeatletter
\@addtoreset{section}{part}
\makeatother
\renewcommand{\thesection}{SA.\arabic{section}} 
\setcounter{section}{0}

This supplement is organized as follows. The first section (\autoref{smaster}) concerns the proof of \autoref{master-clt}. In \autoref{spois}, we prove an abstract CLT for the ``Poissonized'' model and ``de-Poissonize'' the result in \autoref{sdepois} to obtain a CLT for the original ``binomial'' model of interest (see \autoref{spfsketch} for a sketch of the argument). \autoref{master-clt} is a simple corollary of the latter result and is stated at the end of \autoref{smaster}. 

The second section (\autoref{sprimclt}) concerns the proof of \autoref{primclt}. In \autoref{introbp}, we show that the sizes of strategic neighborhoods are stochastically dominated by those of certain branching processes. Then \autoref{sexpstab} presents a series of lemmas culminating in a tail bound for the radius of stabilization. The proof of \autoref{primclt} is stated at the end of \autoref{sprimclt} (see \autoref{smeth2} for a sketch). 

\section{\autoref{master-clt}}\label{smaster}

\subsection{Poissonization}\label{spois}

This subsection proves a CLT for the Poissonized model in which the set of nodes is $\N_{N_n}$ rather than $\N_n$, where $N_n$ is defined in \eqref{poisson}. Throughout, we assume $d_\psi=1$. Define $\tilde{\sigma}_n^2 = n^{-1}\var(\sum_{i=1}^{N_n} \psi_i(\N_{N_n}))$.

\begin{sectheorem}\label{poi-clt}
  Suppose $d_\psi=1$. Under Assumptions \ref{main-bp-exp-stab} and \ref{main-bp-mom}, $\sup_n \tilde\sigma_n^2 < \infty$, and if additionally $\liminf_{n\rightarrow\infty} \tilde{\sigma}_n^2 > 0$, then
  \begin{equation*}
    \tilde\sigma_n^{-1} \frac{1}{\sqrt{n}} \left( \sum_{i=1}^{N_n} \psi_i(\N_{N_n}) - \E\left[ \sum_{i=1}^{N_n} \psi_i(\N_{N_n}) \right] \right) \dlimarrow \N(0,1).  
  \end{equation*}
\end{sectheorem}

\noindent The proof largely follows that of Theorem 2.1 of \cite{penrose2005normal} (henceforth PY). We next present several auxiliary lemmas and provide the proof at the end of this subsection.

\begin{secremark}\label{rlpoi}
  That PY's proof essentially carries over to our setting is perhaps not obvious due to several differences between the setups. First, the definitions of exponential stabilization differ since, as discussed in \autoref{srelit}, we only require invariance of node statistics to the removal of nodes outside the radius of stabilization, whereas PY demand invariance to removal and addition. Second, PY's model has no random-utility shocks $\zeta_{ij}$, but adding these to the model turns out to not affect the argument due to their high degree of independence. Third, PY define $\bm{T}_{N_n}$ as a marked Poisson point process,\footnote{See e.g.\ \cite{last2017lectures} for the definition of a marked Poisson point process.} which is equivalent to our representation because (a) our set of positions $\bm{X}_{N_n} = (X_i)_{i=1}^{N_n}$ has the same distribution as $\mathcal{P}_{nf}$, the Poisson point process on $\R^d$ with intensity function $nf(\cdot)$ \citep[][Proposition 1.5]{penrose2003}, and (b) the ``marks'' $Z_i$ associated with each $X_i$ are independent across nodes, each only potentially correlated with its own position $X_i$. Fourth, PY consider the simpler case where $X_i \indep Z_i$ for each $i$ since this is sufficient for their applications. We allow for dependence, but this turns out not to affect the argument. Finally, PY consider moments of the form
  \begin{equation*}
    \frac{1}{n} \sum_{i=1}^{N_n} \xi\big( (X_i,Z_i), r_n^{-1}\bm{T}_{N_n} \big).
  \end{equation*}

  \noindent This differs from $\psi_n(\cdot)$ because the network is not an argument (nor is $\Zeta_{N_n}$ for reasons previously discussed). However, since the network equals $\lambda_{N_n}(r_n^{-1}\bm{T}_{N_n}, \Zeta_{N_n})$ by \autoref{aesel}, we can define
  \begin{equation*}
    \xi(i, r_n^{-1}\bm{T}_{N_n}, \Zeta_{N_n}) \equiv \psi_{N_n}\big(i, r_n^{-1}\bm{T}_{N_n}, \Zeta_{N_n}, \lambda_{N_n}(r_n^{-1}\bm{T}_{N_n}, \Zeta_{N_n})\big),
  \end{equation*}

  \noindent in which case this coincides with PY's setup except for the presence of $\Zeta_{N_n}$. 
\end{secremark} 

The proof technique is to approximate $\sum_{i=1}^{N_n} \psi_i(\N_{N_n})$ by a sum $\sum_{i=1}^n W_i$ such that the dependence structure of the summands can be characterized in terms of a ``dependency graph'' $\bm{G}$ that links dependent observations. 

\begin{definition}
  Let $\bm{G}$ be a network on $n$ nodes with self-links, meaning that $G_{ii}=1$ for all $i$. We say $\bm{G}$ is a {\em dependency graph} for data $\{W_i\}_{i=1}^n$ if for any $S_1,S_2 \subseteq \N_n$ such that $G_{ij}=0$ for all $i \in S_1$ and $j \in S_2$, we have $\{W_i\colon i\in S_1\} \indep \{W_j\colon j\in S_2\}$.
\end{definition}

Let $\norm{W}_p$ be the $L_p$-norm of $W$ and $\Phi(\cdot)$ the standard normal CDF.

\begin{seclemma}[Dependency Graph CLT]\label{py4.1}
  Let $q \in (2,3]$ and $W = \sum_{i=1}^n W_i$. Suppose $\bm{G}$ is a dependency graph for $\{W_i\}_{i=1}^n$, and let $\Gamma = \max_i \sum_{j\neq i} G_{ij}$. Further suppose $\E[W^2] = 1$, $\E[W_i]=0$, and $\norm{W_i}_q \leq \theta$ for all $i$ and some $\theta>0$. Then
  \begin{equation*}
    \sup_t\, \abs{ \prob(W \leq t) - \Phi(t) } \leq 75\,\theta^q\, \Gamma^{5(q-1)}\, n.
  \end{equation*}
\end{seclemma}
\begin{proof}
  This is Theorem 2.7 of \cite{chen2004normal}.
\end{proof}

We next introduce some notation. Recall from \autoref{main-bp-exp-stab} the definitions of $\tau_p(\cdot)$ and $\eta$. For $p$ in \autoref{main-bp-mom}(a) and $d$ the dimension of $X_1$, let $q \in (2,3]$ such that $q < p$. Choose $\rho_n = (\alpha \log n)^{1/\eta}$ for some $\alpha$ sufficiently large that, for some $C>0$ and all $n\geq 1$, 
\begin{equation}
  \rho_n^{d/p} (n \tau_p(\rho_n))^{(q-2)/(2q)} < C n^{-4}, \quad \tau_p(\rho_n) < Cn^{-3}, \quad \rho_n^d < Cn^{p/(p+2)}.
  \label{rho_n}
\end{equation}

We provide an alternative labeling of nodes by position that is more convenient for the argument that follows. Recall from \autoref{rlpoi} that $\bm{X}_{N_n}$ has the same distribution as the Poisson point process $\mathcal{P}_{nf}$. Let $s_n = r_n \rho_n$ and $V_n$ be the number of cubes of the form $Q = \prod_{i=1}^d [j_i s_n, (j_i+1) s_n)$ where $j_i \in \mathbb{Z}$ for all $i$, such that $Q$ has positive density under $f$. Label the cubes $Q_1, \dots, Q_{V_n}$, and note that they have side lengths $s_n = o(1)$. Since $f$ has bounded support as assumed at the start of \autoref{smodel}, $V_n$ is finite, and
\begin{equation}
  V_n = O(n\rho_n^{-d}). \label{V_n} 
\end{equation}

Fix any $1\leq i\leq V_n$. Note that $N_i \equiv \abs{\bm{X}_{N_n} \cap Q_i}$ is distributed $\text{Poisson}(\nu_i)$ for $\nu_i = n \int_{Q_i} f(x) \,\text{d}x$. Relabel the points of $\bm{X}_{N_n} \cap Q_i$ as $X_{i1}, \dots, X_{iN_i}$, which, conditional on $N_i$, are i.i.d.\ draws from a density $f(\cdot) / \int_{Q_i} f(x)\,\text{d}x$. Under this relabeling,
\begin{equation}
  \bm{X}_{N_n} = \medcup_{i=1}^{V_n} \{X_{ij}\}_{j=1}^{N_i}. \label{Prep}
\end{equation}

\noindent In what follows, we will often use this double-indexed labeling for nodes, writing $X_{ij}$ for node positions and $\psi_{ij}(\N_{N_n})$ for node statistics. That is, suppose under this relabeling that node $k \in \{1, \ldots, N_n\}$ is relabeled $ij$. Then $\psi_{ij}(\N_{N_n})$ equals $\psi_k(\N_{N_n})$, $k$'s node statistic under the Poissonized model with $N_n$ nodes.

The next lemma is mostly a restatement of Lemma 4.3 of \cite{penrose2005normal}, translated to our notation.\footnote{The translation works as follows: replace $\lambda$ (their notation) with $n$ (our notation), $\lambda^{-1/d}$ with $r_n$, $\kappa$ with $f$, $A_\lambda$ with $\text{supp}(f)$, $f$ with the function mapping any real number to 1, $\mathcal{P}_\lambda$ with $\bm{X}_{N_n}$, and $\xi_\lambda(X_{i,j},\mathcal{P}_\lambda)$ with $\psi_{ij}(\N_{N_n})$.} For any $B \subseteq \R^d$, define $\N_{N_n}'(B) = \{i \in \N_{N_n}'\colon X_i \in B\}$. Unlike the definition of $\N_n(B)$ prior to \autoref{rosdef}, this involves a random number of nodes $N_n$ and unscaled node positions $X_i$. 

\begin{seclemma}\label{py4.3}
  Under \autoref{main-bp-mom}(a), for $p$ given in the assumption and $1 < q < p$, there exists $C>0$ such that $\norm{\sum_{j=1}^\infty \abs{\psi_{ij}(\N_{N_n}'(S_{ni}))} \ind\{j\leq N_i\}}_q \leq C \rho_n^{d(p+1)/p}$ for any $Q_i \subseteq S_{ni} \subseteq \R^d$, $1\leq i\leq V_n$, and $n\in\mathbb{N}$.
\end{seclemma}

\noindent Since the lemma is a fairly straightforward consequence of \autoref{main-bp-mom}(a) and properties of the Poisson distribution, we omit the proof. The main difference relative to PY is the presence of random-utility shocks $\Zeta_{N_n}$. The proof only involves quantities conditional on $N_i$ and integrating over $\Zeta_{N_n}$. Because the latter is independent of $N_i$ for all $i$, their presence does not materially affect PY's original argument. 

The last lemma is used to construct a dependency graph. 

\begin{seclemma}\label{elementary} 
  For any two disjoint sets $S_1,S_2 \subseteq \R^d$, 
  \begin{equation}
    \{(X_i,X_j,Z_i,Z_j,\zeta_{ij})\colon i,j \in \N_{N_n}(S_1)\} \indep \{(X_i,X_j,Z_i,Z_j,\zeta_{ij})\colon i,j \in \N_{N_n}(S_2)\}. \label{morespatialindep}
  \end{equation} 
\end{seclemma}
\begin{proof}
This is a consequence of the spatial independence property of Poisson point processes \citep[e.g.][Theorem 5.2]{last2017lectures}. Here we provide an elementary proof from first principles. Let $N(S) = \abs{\N_{N_n}(S)}$, $T_i = (X_i,Z_i)$, $\bm{T}_m = (T_i)_{i=1}^m$, $\bm{T}_S = (T_i\colon X_i \in S)$, $\bm{T}_{m,\dots,n} = (T_i)_{i=m}^n$, and $\bm{X}_{m,\dots,n} = (X_i)_{i=m}^n$. First consider the simple case where $S_1,S_2$ partition $\text{supp}(f)$. Fix $m_1,m_2,m\in\mathbb{N}$ such that $m_1+m_2=m$. For any $\mathcal{U}_m \subseteq \mathcal{T}^m$ and $\mathcal{Z}_m\subseteq \text{supp}(\Zeta_m)$, 
  \begin{align*}
    \prob(N(S_1)&=m_1, \bm{T}_{S_1} \in \mathcal{U}_{m_1}, \Zeta_{S_1} \in \mathcal{Z}_{m_1}, N(S_2)=m_2, \bm{T}_{S_2}\in \mathcal{U}_{m_2}, \Zeta_{S_2} \in \mathcal{Z}_{m_2}) \\
		&\begin{aligned}= \prob(N_n=m) \prob(N(S_1)=m_1, \bm{T}_{S_1} \in \mathcal{U}_{m_1}, 
		&\Zeta_{S_1} \in \mathcal{Z}_{m_1}, N(S_2)=m_2, \\
		&\bm{T}_{S_2} \in \mathcal{U}_{m_2}, \Zeta_{S_2} \in \mathcal{Z}_{m_2} \mid N_n=m) \end{aligned} \\
		&\begin{aligned}= \frac{e^{-n} n^m}{m!} \frac{m!}{m_1!m_2!} \prob(\bm{X}_{m_1} \in S_1, \, &\bm{T}_{m_1} \in \mathcal{U}_{m_1}, \Zeta_{m_1} \in \mathcal{Z}_{m_1}, \bm{X}_{m_1+1,\dots,m} \in S_2, \\ 
		&\bm{T}_{m_1+1,\dots,m} \in \mathcal{U}_{m_2}, \Zeta_{m_1+1,\dots,m} \in \mathcal{Z}_{m_2}) \end{aligned} \\
		&\begin{aligned}= \frac{e^{-n \prob(X_1 \in S_1)}n^{m_1}}{m_1!} &\prob(
		  \bm{X}_{m_1} \in S_1, \bm{T}_{m_1} \in \mathcal{U}_{m_1}, \Zeta_{m_1} \in \mathcal{Z}_{m_1}) \\ 
		&\times \frac{e^{-n \prob(X_1 \in S_2)}n^{m_2}}{m_2!} \prob(\bm{X}_{m_2} \in S_2, \bm{T}_{m_2} \in \mathcal{U}_{m_2}, \Zeta_{m_2} \in \mathcal{Z}_{m_2}). \end{aligned}
  \end{align*}

  \noindent The third equality holds because types are identically distributed, and $N(S_1)=m_1$ and $N(S_2)=m_2$ imply $N_n=m$ given that $S_1,S_2$ partition the space. The last line is true because disjoint sets of nodes have independent types and random-utility shocks.

  We perform similar calculations for the marginal distribution:
  \begin{align*}
    &\begin{aligned}\prob(N(S_1)=m_1, \bm{T}_{S_1} \in \mathcal{U}_{m_1}, \Zeta_{S_1} \in \,\,&\mathcal{Z}_{m_1}) = \sum_{m\geq m_1} \prob(N_n=m) \prob(N(S_1)=m_1, \\
    &\bm{T}_{S_1} \in \mathcal{U}_{m_1}, \Zeta_{S_1} \in \mathcal{Z}_{m_1}, N(S_1^c)=m-m_1 \mid N_n=m) \end{aligned} \\
    &= \sum_{m\geq m_1} \frac{e^{-n} n^m}{m!} \frac{m!}{m_1!(m-m_1)!} \prob(\bm{X}_{m_1} \in S_1, \bm{T}_{m_1} \in \mathcal{U}_{m_1}, \Zeta_{m_1} \in \mathcal{Z}_{m_1}, \bm{X}_{m_1+1,\dots,n} \in S_1^c) \\
    &\begin{aligned}= \sum_{m\geq m_1} \frac{e^{-n \prob(X_1 \in S_1)}n^{m_1}}{m_1!} \prob(\bm{X}_{m_1} \in S_1, \, \bm{T}_{m_1} \in \,\,&\mathcal{U}_{m_1},\Zeta_{m_1} \in \mathcal{Z}_{m_1}) \\
    &\times\frac{e^{-n \prob(X_1 \in S_1^c)}n^{m-m_1}}{(m-m_1)!} \prob(\bm{X}_{m_1+1,\dots,n} \in S_1^c) \end{aligned}\\
    &\begin{aligned}= \frac{e^{-n \prob(X_1 \in S_1)}n^{m_1}}{m_1!} \prob(\bm{X}_{m_1} \in S_1, \bm{T}_{m_1} \in \mathcal{U}_{m_1}, \,\,&\Zeta_{m_1} \in \mathcal{Z}_{m_1}) \\
      &\times \underbrace{\sum_{m\geq m_1} \frac{e^{-n \prob(X_1 \in S_1^c)}(n\prob(X_1 \in S_1^c))^{m-m_1}}{(m-m_1)!}}_1. \end{aligned}
  \end{align*}

  \noindent Therefore,
  \begin{multline*}
    \prob(N(S_1)=m_1, \bm{T}_{S_1} \in \mathcal{U}_{m_1}, \Zeta_{S_1} \in \mathcal{Z}_{m_1}, N(S_2)=m_2, \bm{T}_{S_2}\in \mathcal{U}_{m_2}, \Zeta_{S_2} \in \mathcal{Z}_{m_2}) \\
    = \prob(N(S_1)=m_1, \bm{T}_{S_1} \in \mathcal{U}_{m_1}, \Zeta_{S_1} \in \mathcal{Z}_{m_1}) \prob(N(S_2)=m_2, \bm{T}_{S_2}\in \mathcal{U}_{m_2}, \Zeta_{S_2} \in \mathcal{Z}_{m_2}), 
  \end{multline*}

  \noindent which proves \eqref{morespatialindep} for the simple case where $S_1,S_2$ partition $\text{supp}(f)$. For the general case in which $S_1,S_2$ does not partition $\text{supp}(f)$, since the sets are disjoint, we can partition the support into two sets $S_1',S_2'$ such that $S_1 \subseteq S_1'$ and $S_2 \subseteq S_2'$ and apply the result for the simple case.
\end{proof}

\bigskip
\begin{proof}[Proof of \autoref{poi-clt}]
  We mostly follow PY's proof of Theorem 2.1 with some minor modifications, reorganization, and elaboration. That $\sup_n \tilde\sigma_n^2 < \infty$ is a consequence of \autoref{pen5.2} below. Abbreviate
  \begin{equation*}
    \Lambda_n = \sum_{i=1}^{N_n} \psi_i(\N_{N_n}).
  \end{equation*}
  
  \noindent We prove the following stronger result: for $q \in (2,3]$ with $q<p$ and $\eta$ from \autoref{main-bp-exp-stab}, there exists $C>0$ such that for all $n$ sufficiently large,
  \begin{equation}
    \sup_t \,\abs{ \prob\big( \var(\Lambda_n)^{-1/2} (\Lambda_n - \E[\Lambda_n]) \leq t\big) - \Phi(t) } \leq C (\log n)^{dq/\eta} n \var(\Lambda_n)^{-q/2}. \label{stronger}
  \end{equation}

  \noindent Since $\var(\Lambda_n)^{-1} = O(n^{-1})$ by assumption, the right-hand side of \eqref{stronger} tends to zero, which proves the theorem. 

  Under representation \eqref{Prep}, $\Lambda_n = \sum_{i=1}^{V_n} \sum_{j=1}^{N_i} \psi_{ij}(\N_{N_n})$. The key idea is to approximate this with a similar sum to which the dependency graph CLT (\autoref{py4.1}) more usefully applies. For any $1\leq i\leq V_n$, let $Q_i^+ = \{x\in\R^d\colon \norm{x-y} \leq r_n\rho_n \text{ for some } y \in Q_i\}$, the $r_n\rho_n$-neighborhood of $Q_i$. Recall the definition of $\N_{N_n}'(S)$ prior to \autoref{py4.3}. Define
  \begin{equation*}
    \Lambda_n' = \sum_{i=1}^{V_n} \sum_{j=1}^{N_i} \psi_{ij}(\N'_{N_n}(Q_i^+)) \quad\text{and}\quad \Psi_i = \var(\Lambda_n')^{-1/2} \sum_{j=1}^{N_i} \psi_{ij}(\N'_{N_n}(Q_i^+)), 
  \end{equation*}

  \noindent so that $\var(\Lambda_n')^{-1/2} \Lambda_n' = \sum_{i=1}^{V_n} \Psi_i$. 

  Recall $\RR_i(n,r_n)$ from \autoref{rosdef}. The relabeling introduced above indexes nodes by $ij$ for $1\leq i\leq V_n$ and $1\leq j\leq N_i$, so we instead write $\RR_{ij}(n,r_n)$. For $\rho_n$ defined prior to \eqref{rho_n}, define the event $\mathcal{E}_{ij} = \{\RR_{ij}(N_n,r_n) \leq \rho_n\}$.

  \bigskip

  \noindent {\bf Approximating $\Lambda_n$.} We derive a bound on the discrepancy between $\Lambda_n$ and $\Lambda_n'$. We first show that $\Lambda_n' = \Lambda_n$ under the event $\mathcal{E}_n = \medcap_{i=1}^{V_n} \medcap_{j=1}^\infty \mathcal{E}_{ij}$. To see this, first define $r_n^{-1}Q_i^+ = \{r_n^{-1}x\colon x \in Q_i^+\}$, and observe that $x \in B(X_{ij}, r_n\rho_n) \subseteq Q_i^+$ if and only if $r_n^{-1}x \in B(r_n^{-1}X_{ij}, \rho_n) \subseteq r_n^{-1}Q_i^+$. Hence, $\N_{N_n}'(Q_i^+) = \N_{N_n}(r_n^{-1}Q_i^+)$ for $\N_{N_n}(B)$ defined prior to \autoref{rosdef}. As a result, $\psi_{ij}(\N_{N_n}'(Q_i^+)) = \psi_{ij}(\N_{N_n}(r_n^{-1}Q_i^+))$. Moreover, for any $X_{ij} \in Q_i$ and state of the world in $\mathcal{E}_{ij}$, we have $\{k\in\N_{N_n}\colon r_n^{-1}X_k \in B(r_n^{-1}X_{ij}, \bm{R}_{ij}(N_n,r_n))\} \subseteq \N_{N_n}(r_n^{-1}Q_i^+)$, which implies $\psi_{ij}(\N_{N_n}(r_n^{-1}Q_i^+)) = \psi_{ij}(\N_{N_n})$ by \autoref{rosdef}. Since this is true for all $ij$ under $\mathcal{E}_n$, it follows that $\Lambda_n' = \Lambda_n$ under this event.

  By \autoref{py4.3}, Minkowski's inequality, and \eqref{V_n}, there exists $C'>0$ such that for all $n$,
  \begin{equation}
    \max\big\{\norm{\Lambda_n}_q, \norm{\Lambda_n'}_q\big\} \leq C V_n \rho_n^{d(p+1)/p} \leq C' n \rho_n^{d/p}. \label{4.14}
  \end{equation}

  \noindent By the H\"{o}lder and Minkowski inequalities,
  \begin{equation}
    \norm{\Lambda_n - \Lambda_n'}_2 \leq \norm{\Lambda_n - \Lambda_n'}_q \prob(\mathcal{E}_n^c)^{0.5 - 1/q} \leq (\norm{\Lambda_n}_q + \norm{\Lambda_n'}_q) \prob(\mathcal{E}_n^c)^{(q-2)/(2q)}. \label{pre4.15}
  \end{equation}

  \noindent By the law of iterated expectations, 
  \begin{equation*}
    \prob\left( \mathcal{E}_n^c \right) \leq \E\left[ \sum_{i=1}^{V_n} \sum_{j=1}^{N_i} \ind\{\mathcal{E}_{ij}^c\} \right] = \E\left[ \sum_{i=1}^{V_n} \sum_{j=1}^{N_i} \E[\ind\{\mathcal{E}_{ij}^c\} \mid \bm{X}_{N_n}, N_n] \right]
  \end{equation*}

  \noindent As in the setting of \cite{penrose2005normal}, this is an expectation of pointwise sums (i.e.\ over the points of $\bm{X}_{N_n}$) of functionals of a Poisson point process $\bm{X}_{N_n}$. We may then apply Theorem 1.6 of \cite{penrose2003} to obtain 
  \begin{multline}
    \prob\left( \mathcal{E}_n^c \right) \leq \E\left[ \sum_{i=1}^{V_n} \sum_{j=1}^{N_i} \E[\ind\{\mathcal{E}_{ij}^c\} \mid \bm{X}_{N_n}, N_n] \right] \\ = n \int_{\R^d} \prob\left( \RR_{ij}(N_n+1,r_n) > \rho_n \mid X_{ij} = x \right) f(x) \,\text{d}x \leq n\, \tau_p(\rho_n) \label{4.11}
  \end{multline} 

\noindent for $n$ sufficiently large.\footnote{The quantity $h(Y, \mathcal{P}_\lambda)$ in Theorem 1.6 corresponds to $\E[\ind\{\mathcal{E}_{ij}^c\} \mid \bm{X}_{N_n}, N_n]$, which is a function of the Poisson process $\mathcal{P}_\lambda = \bm{X}_{N_n}$ and unit $ij$'s position $Y = X_{ij}$.} By \eqref{4.14}, \eqref{pre4.15}, \eqref{4.11}, our choice of $\rho_n$ \eqref{rho_n}, and \autoref{main-bp-exp-stab}, there exist $C,C'>0$ such that for all $n$
  \begin{equation}
    \norm{\Lambda_n - \Lambda_n'}_2 \leq C n \rho_n^{d/p} ( n \tau_p(\rho_n) )^{(q-2)/(2q)} \leq C' n^{-3}.
    \label{4.15}
  \end{equation}

  \noindent {\bf Dependency Graph.} We construct a dependency graph $\bm{G}$ for the summands $\{\Psi_i\}_{i=1}^{V_n}$ of $\Lambda_n'$ by setting $G_{ij} = \ind\{\inf\{\norm{x-y}\colon x \in Q_i, y \in Q_j\} \leq 2 r_n\rho_n\}$. To see that this is a dependency graph, observe that the value of $\Psi_i$ is determined solely by the types of and random-utility shocks between nodes positioned in $Q_i^+$. Then whenever $\inf\{\norm{x-y}\colon x \in Q_i, y \in Q_j\} > 2 r_n\rho_n$, we have $\{\Psi_k\colon k \in Q_i\} \indep \{\Psi_\ell\colon \ell \in Q_j\}$ by \autoref{elementary}. Hence, $\bm{G}$ is a dependency graph.

  Note that \eqref{stronger} holds trivially when $\var(\Lambda_n) < 1$ for $n$ sufficiently large, so we henceforth consider the case where it is at least 1. We apply \autoref{py4.1} to $\{W_i\}_{i=1}^{V_n}$ for $W_i = \Psi_i - \E[\Psi_i]$. To determine a good choice of $\theta$ for applying the lemma, notice by \autoref{py4.3} that
  \begin{equation*}
    \norm{\Psi_i}_q \leq C\, \var(\Lambda_n')^{-1/2} \bigg|\bigg| \sum_{j=1}^\infty \abs{\psi_{ij}(\N'_{N_n}(Q_i^+)) \ind\{j\leq N_i\}} \bigg|\bigg|_q \leq C\, \var(\Lambda_n')^{-1/2} \rho_n^{d(p+1)/p}.
  \end{equation*}
 
  \noindent We set $\theta$ to be the right-hand side. There are at most $5^d$ other cubes at most distance $2r_n\rho_n$ from cube $Q_i$, so $\Gamma \leq 5^d$ for $\Gamma = \max_i \sum_{j\neq i} G_{ij}$. By \autoref{py4.1} and \eqref{V_n}, for $W = \sum_{i=1}^{V_n} (\Psi_i - \E[\Psi_i])$, there exists $C>0$ such that for all $n$
  \begin{equation}
    \sup_t\, \abs{ \prob(W \leq t) - \Phi(t) } \leq Cn\rho_n^{-d} \,\var(\Lambda_n')^{-q/2} \rho_n^{dq(p+1)/p} \label{pre4.18}
  \end{equation}

  \noindent {\bf Right side of \eqref{stronger}.} We next show that the right-hand side of \eqref{pre4.18} is at most the right-hand side of \eqref{stronger}. First, we bound the discrepancy between $\var(\Lambda_n')$ and $\var(\Lambda_n)$. Trivially,
  \begin{equation}
    \var(\Lambda_n) = \var(\Lambda_n') + \var(\Lambda_n - \Lambda_n') + 2 \cov(\Lambda_n', \Lambda_n-\Lambda_n'). \label{pre4.17}
  \end{equation}

  \noindent We can bound the first term on the right-hand side as follows. By construction, $\Lambda_n'$ is the sum of $V_n$ random variables that have second moments bounded by a constant times $\rho_n^{2d(p+1)/p}$ due to \autoref{py4.3} for $q=2$. Also, as previously established, the covariance between any pair of these variables is zero when the corresponding nodes are respectively positioned in cubes $Q_i$ and $Q_j$ such that any element of $Q_i$ is more than $2r_n\rho_n$ away from any element of $Q_j$. Otherwise, the covariance is bounded by a constant times $\rho_n^{2d(p+1)/p}$ by Cauchy-Schwarz. Consequently, by \eqref{V_n},
  \begin{equation}
    \var(\Lambda_n') = O(n \rho_n^{d(p+2)/p}).
    \label{4.13}
  \end{equation}

  Combining \eqref{4.15}, \eqref{pre4.17}, and \eqref{4.13}, by our choice of $\rho_n$ \eqref{rho_n} and Cauchy-Schwarz, 
  \begin{equation}
    \abs{\var(\Lambda_n) - \var(\Lambda_n')} \leq C n^{-2}.
    \label{4.17}
  \end{equation}

  \noindent Given $\var(\Lambda_n) \geq 1$ from above, this implies that, for $n$ sufficiently large, $\var(\Lambda_n') \geq \var(\Lambda_n)/2$. Since $q<p$,
  \begin{equation}
    \eqref{pre4.18} \leq C n \var(\Lambda_n)^{-q/2} \rho_n^{dq}, \label{4.18}
  \end{equation}

  \noindent which delivers the bound on the right-hand side of \eqref{stronger}, given $\rho_n = (\alpha \log n)^{1/\eta}$.

  \bigskip

  \noindent {\bf Left side of \eqref{stronger}.} However, \eqref{pre4.18} and \eqref{4.18} do not immediately deliver \eqref{stronger} because they apply to $W = \var(\Lambda_n')^{-1/2} (\Lambda_n' - \E[\Lambda_n'])$ rather than the target $\var(\Lambda_n)^{-1/2} (\Lambda_n - \E[\Lambda_n])$. Let $\tilde W = \var(\Lambda_n')^{-1/2} (\Lambda_n - \E[\Lambda_n])$. Using the Lipschitz property of $\Phi(\cdot)$ and \eqref{4.18}, elementary calculations yield
  \begin{equation}
    \sup_t\, \abs{\prob(\tilde W \leq t) - \Phi(t)} \leq C\beta + Cn\var(\Lambda_n)^{-q/2} \rho_n^{dq} + \prob(\abs{\tilde W - W} \geq \beta)
    \label{4.19}
  \end{equation}

  \noindent for any $\beta>0$ \citep[][eq.\ (4.19)]{penrose2005normal}. Also,
  \begin{equation}
    \abs{\tilde W - W} \leq \var(\Lambda_n')^{-1/2} \left( \abs{\Lambda_n - \Lambda_n'} + \E[\abs{\Lambda_n - \Lambda_n'}] \right). \label{g2j03}
  \end{equation}

  \noindent Since \eqref{4.15} implies $\E[\abs{\Lambda_n - \Lambda_n'}] \leq C n^{-3}$, we have \eqref{g2j03} $\leq Cn^{-3}$ under $\mathcal{E}_n$. On the other hand, $\prob(\mathcal{E}_n^c) \leq Cn^{-2}$ by \eqref{4.11} and \eqref{rho_n}. So choosing $\beta = Cn^{-3}$ in \eqref{4.19}, we obtain
  \begin{equation}
    \sup_t\, \abs{\prob(\tilde W \leq t) - \Phi(t)} \leq Cn\var(\Lambda_n)^{-q/2} \rho_n^{dq} + Cn^{-2}.
    \label{4.20}
  \end{equation}
  
  With these pieces in place, we may now prove \eqref{stronger}. We have
  \begin{multline}
    \sup_t\, \abs{\prob\big(\var(\Lambda_n)^{-1/2} (\Lambda_n - \E[\Lambda_n]\big) \leq t) - \Phi(t)} \\ \leq \sup_t \bigg| \prob\left(\tilde W \leq t\left( \frac{\var(\Lambda_n)}{\var(\Lambda_n')} \right)^{1/2} \right) - \Phi\left( t \left( \frac{\var(\Lambda_n)}{\var(\Lambda_n')} \right)^{1/2} \right) \bigg| \\
    + \sup_t \bigg| \Phi\left( t \left( \frac{\var(\Lambda_n)}{\var(\Lambda_n')} \right)^{1/2} \right) - \Phi(t) \bigg|. \label{4.21}
  \end{multline}

  \noindent For $n$ large, $\var(\Lambda_n') \geq \var(\Lambda_n)/2 \geq 0.5$ since $\var(\Lambda_n) \geq 1$. By \eqref{4.17}, there exists $C > 0$ such that for all $n>0$,
  \begin{equation}
    \bigg| t \left( \frac{\var(\Lambda_n)}{\var(\Lambda_n')} \right)^{1/2} - t \bigg| \leq \abs{t} \bigg| \frac{\var(\Lambda_n)}{\var(\Lambda_n')} - 1 \bigg| \leq C n^{-2} \abs{t}. \label{ghawi3409}
  \end{equation}

  \noindent Letting $\phi(\cdot)$ denote the normal density, for all $s\leq t$, $\abs{\Phi(s)-\Phi(t)} \leq (t-s) \max_{u \in [s,t]} \phi(u)$. Then using \eqref{ghawi3409},
  \begin{equation*}
    \sup_t \bigg| \Phi\left( t \left( \frac{\var(\Lambda_n)}{\var(\Lambda_n')} \right)^{1/2} \right) - \Phi(t) \bigg| \leq Cn^{-2} \sup_t \big\{ \abs{t} \max_{u \in [t-tCn^{-2}, t+tCn^{-2}]} \phi(u) \big\} \leq Cn^{-2}.
  \end{equation*}

  \noindent Combined with \eqref{4.20} and \eqref{4.21}, we obtain
  \begin{equation*}
    \sup_t\, \abs{ \prob\big( \var(\Lambda_n)^{-1/2} (\Lambda_n - \E[\Lambda_n]) \leq t\big) - \Phi(t) } \leq C (\log n)^{dq/\eta} n \var(\Lambda_n)^{-q/2} + O(n^{-2}).
  \end{equation*}

  \noindent By \eqref{4.13} and \eqref{4.17}, $\var(\Lambda_n) = O(n\rho_n^{d(p+2)/p})$, so by \eqref{rho_n} and the fact that $q \in (2,3]$, the first term on the right-hand side dominates the second term, and \eqref{stronger} follows.
\end{proof}

\subsection{de-Poissonization}\label{sdepois}

As in the previous subsection, assume $d_\psi=1$. The next result shows that, given a CLT for the Poissonized model with node set $\N_{N_n}$, a CLT also holds for the original ``binomial'' model of interest with node set $\N_n$. Define $\sigma_n^2 = n^{-1} \var(\sum_{i=1}^n \psi_i(\N_n))$.

\begin{sectheorem}\label{pen2.3} 
  Suppose $d_\psi=1$. Under Assumptions \ref{main-bp-exp-stab} and \ref{main-bp-mom}, $\sup_n \sigma_n^2 < \infty$. Further suppose that $\liminf_{n\rightarrow\infty} \sigma_n^2 > 0$. Then $\liminf_{n\rightarrow\infty} \tilde\sigma_n^2 > 0$ and
  \begin{equation*}
    \sigma_n^{-1} \frac{1}{\sqrt{n}} \sum_{i=1}^n \big(\psi_i(\N_n) - \E[\psi_i(\N_n)]\big) \dlimarrow \N(0,1).  
  \end{equation*}
\end{sectheorem}

The proof is stated at the end of this subsection and uses the next four lemmas. \autoref{3.1py} establishes an asymptotic equivalence between $\bm{X}_n$ and $\bm{X}_{N_n}$ and is taken directly from \cite{py2003}. \autoref{pen3.7} proves that the node statistics of two nodes positioned a fixed distance apart are asymptotically uncorrelated. \autoref{pen5.1} proves that add-one costs of two nodes at a fixed distance apart are asymptotically uncorrelated. Finally, \autoref{pen5.2} shows that various moments are uniformly bounded. 

\begin{secremark}[Related Literature]\label{rlde} 
  These results are adaptations of lemmas in \cite{penrose2007gaussian}, so we compare and contrast our setup with his. First, the differences discussed in \autoref{rlpoi} persist here. Second, the proofs of Lemmas \ref{pen3.7} and \ref{pen5.1} need some modification because our setups differ, most notably the definition of stabilization. Third, to prove the lemmas, we define a coupling that follows the ideas in \cite{penrose2007gaussian} and \cite{py2003}, but the details are different because we need to accommodate random-utility shocks $\zeta_{ij}$, dependence between $X_i$ and $Z_i$, and differences in notation. Fourth, Penrose considers a generalization of $\psi_n(\cdot)$ that is a {\em measure} in that it takes as an argument a subset of $\R^d$. This generalization is unnecessary for our purposes.
\end{secremark}

The lemmas that follow utilize the following coupling whose essential idea is from \cite{py2003}, p.\ 298. As in \autoref{eselection}, define i.i.d.\ types $\{(X_i,Z_i)\}_{i\in\mathbb{N}}$ and i.i.d.\ random-utility shocks $\{\zeta_{ij}\colon i\neq j, \{i,j\} \subseteq \mathbb{N}\}$, with the two sets mutually independent. For any $n\in\mathbb{N}$, let $N_n'$ be an independent copy of $N_n$ (defined in \autoref{spfsketch}), independent of all other primitives. 
\begin{itemize}
  \item For any $m\in\mathbb{N}$, let each node $i\in\N_m$ be assigned type $T_i = (X_i, Z_i)$ and each node pair $i,j\in\N_m$ assigned random-utility shock $\zeta_{ij}$. Define
    \begin{equation}
      \bm{X}_m = (X_5, X_7, X_9, \dots, X_{2m+3}).
      \label{Xm}
    \end{equation}

    \noindent This will be the array of node positions, which notably excludes even-labeled nodes and nodes $1, \ldots, 4$. The latter are reserved for particular use in the proof of \autoref{pen3.7} below, while the even-labeled nodes will be used in the next bullet point. Observe that \eqref{Xm} has the same distribution as the more conventional-looking array $(X_1, X_2, \dots, X_m)$. 
    
  \item Fix any $x,y\in\R^d$. Let $F_x$ be the set of points in $\R^d$ closer to $x$ than $y$ and $F_y$ the remaining set of points (break ties arbitrarily). Let $\mathcal{P}_{nf}^{*,x}$ be the restriction of $\tilde{\bm{X}}_{N_n} \equiv (X_5, X_7, X_9, \dots, X_{2N_n+3})$ to $F_x$. Let $\mathcal{Q}_{nf}^{*,x}$ be the restriction of $\tilde{\bm{X}}_{N_n'} \equiv (X_6, X_8, X_{10} \dots, X_{2N_n'+4})$ to $F_y$. Construct $\mathcal{P}_{nf}^x$ by concatenating the two vectors $\mathcal{P}_{nf}^{*,x}$ and $\mathcal{Q}_{nf}^{*,x}$. As the notation suggests, for any $n$, the set consisting of the components of $\mathcal{P}_{nf}^x$ has the same distribution as an inhomogeneous Poisson point process on $\R^d$ with intensity function $nf(\cdot)$. This follows from Poisson superposition \citep[][Theorem 3.3]{last2017lectures} since $F_x$ and $F_y$ partition $\R^d$. Hence, $\mathcal{P}_{nf}^{x}$ has the same distribution as $\bm{X}_{N_n}$ used in \autoref{spois}.
    
  \item Similarly, let $\mathcal{P}_{nf}^{*,y}$ be the restriction of $\tilde{\bm{X}}_{N_n}$ to $F_y$ and $\mathcal{Q}_{nf}^{*,y}$ the restriction of $\tilde{\bm{X}}_{N_n'}$ to $F_x$. Construct $\mathcal{P}_{nf}^y$ by concatenating the two vectors $\mathcal{P}_{nf}^{*,y}$ and $\mathcal{Q}_{nf}^{*,y}$. Then $\mathcal{P}_{nf}^{y}$ has the same distribution as $\bm{X}_{N_n}$. 
\end{itemize}

\noindent The point of this construction is that, by switching the odds and evens on the two half-sets $F_x$ and $F_y$, we have $\mathcal{P}_{nf}^{x} \indep \mathcal{P}_{nf}^{y}$. As a consequence, we obtain the following independence property. Let $r_n^{-1}\mathcal{P}_{nf}^x = \{r_n^{-1}X\colon X \in \mathcal{P}_{nf}^x\}$. For any $S_1,S_2 \subseteq \R^d$ with $S_1 \cap S_2 = \emptyset$, define $\N_x(S_1) = \{k \in \mathbb{N}\colon r_n^{-1}X_k \in r_n^{-1}\mathcal{P}_{nf}^x \cap S_1\}$ and $\N_y(S_2)$ similarly. Then
\begin{equation}
  \{(X_i,X_j,Z_i,Z_j,\zeta_{ij})\colon i,j \in \N_x(S_1)\} \indep \{(X_i,X_j,Z_i,Z_j,\zeta_{ij})\colon i,j \in \N_y(S_2)\}. \label{morespatialindep2}
\end{equation}

\noindent This is similar to \eqref{morespatialindep} but using two ``interleaved'' Poisson processes $\mathcal{P}_{nf}^{x}$ and $\mathcal{P}_{nf}^{y}$. The interleaving idea of taking two different ``halves'' on $F_x$ and $F_y$ of two independent Poisson processes (in our case, the odds and the evens) is taken from \cite{py2003}. Our construction only differs by building the two processes from the odd and even elements of the same countable set of random vectors. 

For any $x,y\in\R^d$, $r>0$, and $\bm{X} \subseteq \R^d$, let $\tau_{x,r}y = x+r^{-1}(y-x)$ and $\tau_{x,r}\bm{X} = \{x+r^{-1}(y-x)\colon y \in \bm{X}\}$. This operator dilates the point set $\bm{X}$ relative to point $x$, whereas in the main paper, we use $r^{-1}\bm{X}=\tau_{\zero,r}\bm{X}$, which specifically dilates $\bm{X}$ relative to the origin. For technical reasons, it will be convenient to consider dilations relative to other locations. The next coupling lemma shows that the ``binomial process'' $\tau_{x,r_n}\bm{X}_n$ and Poisson process $\tau_{x,r_n}\mathcal{P}_{nf}^x$ coincide on neighborhoods of $x$.

\begin{seclemma}\label{3.1py}
  For any $R > 0$, $x,y \in \R^d$, and sequence $\{\ell_n\}_{n\in\mathbb{N}}$ with $\ell_n/n\rightarrow 1$,
  \begin{equation*}
    \prob\big( \tau_{x,r_n}\bm{X}_{\ell_n} \cap B(x,R) = \tau_{x,r_n}\mathcal{P}_{nf}^x \cap B(x,R) \big) \rightarrow 1.
  \end{equation*}
\end{seclemma}
\begin{proof}
  This is Lemma 3.1 of \cite{py2003}.
\end{proof}

\begin{secremark} 
  The coupling will be used to show that the node statistic of a node positioned at $x$ is asymptotically uncorrelated with that of a node positioned at $y$, for any fixed $x,y\in\R^d$ (\autoref{pen3.7}). The idea is that by stabilization, $x$'s statistic is primarily determined by nodes near $x$ and likewise for $y$, and by \autoref{3.1py} and stabilization, nodes near $x$ will be elements of $\mathcal{P}_{nf}^x$ and nodes near $y$ of $\mathcal{P}_{nf}^y$, which are independent. Note this does not mean node statistics are essentially i.i.d.\ because in the continuum limit, there are many nodes near $x$ and many near $y$.
\end{secremark}

Recall that $\Phi_z(\cdot \mid x)$ is the conditional distribution of $Z_1$ given $X_1=x$. Fix any $w,x,y,z \in \text{supp}(f)$, $z_1 \in \text{supp}(\Phi_z(\cdot \mid x))$, $z_2 \in \text{supp}(\Phi_z(\cdot \mid y))$, $z_3 \in \text{supp}(\Phi_z(\cdot \mid x+r_nz))$, and $z_4 \in \text{supp}(\Phi_z(\cdot \mid y+r_nw))$. Consider a modification of the coupling construction in which we {\em redefine}
\begin{align*}
  &X_1\equiv x, X_2\equiv y, X_3\equiv x+r_nz, X_4\equiv y+r_nw, \quad\text{and} \\
  &Z_k\equiv z_k \quad\text{for all } k=1, \ldots, 4.
\end{align*}

\noindent That is, we replace the types associated with the first four positions in the coupling construction with non-random values. We refer to this as ``{\bf DGP $(\star)$}'' in what follows. Note that the DGP implicitly depends on $n$ and the eight non-random vectors defined above. It will be used to construct objects analogous to quantities defined prior to Assumptions \ref{main-bp-exp-stab} and \ref{main-bp-mom}(a) involving a $\star$ superscript.

For any finite $J \subseteq \mathbb{N}$, $i\in J$, and $v\in\R^d$, define 
\begin{equation*}
  \psi_{i,v}(J) = \psi_{\abs{J}}\big(i, \tau_{v,r_n}\bm{T}_J, \Zeta_J, \lambda_{\abs{J}}(\tau_{v,r_n}\bm{T}_J, \Zeta_J)\big)
\end{equation*}

\noindent where $\tau_{v,r_n}\bm{T}_J = ((\tau_{v,r_n}X_i,Z_i))_{i\in J}$ and $\Zeta_J = (\zeta_{ij})_{i\neq j; i,j\in J}$. In other words, $\psi_{i,v}(J)$ is $i$'s node statistic under a model in which the node set is $J$ and positions are dilated relative to $v$. For $v=\zero$, $\psi_{i,v}(J)$ is a representation of $\psi_i(J)$ defined in \eqref{psiH}. Recall from \eqref{morespatialindep2} the definition of $\N_x(S_1)$, and abbreviate $\N_x \equiv \N_x(\R^d)$. Then for $u \in \{x,y\} \subseteq \R^d$, $\psi_{i,x}(\N_u)$ corresponds to $i$'s statistic under the Poissonized model where the number of nodes $N_n$ is random. When $u=m\in\mathbb{N}$, it corresponds to $i$'s statistic under the usual ``binomial'' model with non-random node set $\N_m$. Finally, for any $a,b \in \mathbb{N}$ with $a\leq b$, define $\N_{a;b} = \{a, \ldots, b+a-1\}$.

\begin{seclemma}\label{pen3.7}
  Let $\{\ell_n\}_{n\in\mathbb{N}}$ and $\{m_n\}_{n\in\mathbb{N}}$ be sequences of natural numbers such that $\ell_n/n$ and $m_n/n$ tend to one. Define
  \begin{align*}
    &\begin{aligned} V_1(x,y) = \big( &\psi_{1,x}(\N_{5;\ell_n} \cup \{1\}), \psi_{1,x}(\N_{5;\ell_n} \cup \{1,2\}), \psi_{1,x}(\N_{5;\ell_n} \cup \{1,3\}), \\
			 &\psi_{1,x}(\N_{5;\ell_n} \cup \{1,2,3\}), \psi_{1,x}(\N_{5;m_n} \cup \{1\}), \psi_{1,x}(\N_{5;m_n} \cup \{1,2\}), \\
			 &\psi_{1,x}(\N_{5;m_n} \cup \{1,2,3\}), \psi_{2,y}(\N_{5;m_n} \cup \{2\}), \psi_{2,y}(\N_{5;m_n} \cup \{1,2\}), \\
			 &\psi_{2,y}(\N_{5;m_n} \cup \{1,2,3\}), \psi_{2,y}(\N_{5;m_n} \cup \{1,\dots,4\}) \big) \quad\text{and} \end{aligned} \\
			 &\begin{aligned} V_2(x,y) = \big( &\psi_{1,x}(\N_x \cup \{1\}), \psi_{1,x}(\N_x \cup \{1\}), \psi_{1,x}(\N_x \cup \{1,3\}), \\
			 &\psi_{1,x}(\N_x \cup \{1,3\}), \psi_{1,x}(\N_x \cup \{1\}), \psi_{1,x}(\N_x \cup \{1\}), \\
			 &\psi_{1,x}(\N_x \cup \{1,3\}), \psi_{2,y}(\N_y \cup \{2\}), \psi_{2,y}(\N_y \cup \{2\}), \\
			 &\psi_{2,y}(\N_y \cup \{2\}), \psi_{2,y}(\N_y \cup \{2,4\}) \big). \end{aligned}
  \end{align*}

  \noindent Under DGP $(\star)$ and \autoref{main-bp-exp-stab}, 
  \begin{equation*}
    \norm{V_1(x,y) - V_2(x,y)}_\infty \plimarrow 0.
  \end{equation*}

  \noindent where $\norm{\cdot}_\infty$ denotes the component-wise maximum. Furthermore, the terms in $V_2(x,y)$ with $x$ subscripts are independent of the terms with $y$ subscripts by \eqref{morespatialindep2}.
\end{seclemma}

\begin{secremark} 
  The vector $V_1(x,y)$ involves only binomial models, whereas $V_2(x,y)$ involves only Poissonized models. By the coupling construction, $\N_x$ and $\N_y$ involve only nodes with labels 5 or greater, and by definition, $\N_{5;\ell_n}$ involves only nodes with labels 5 through $\ell_n+5$. Then $\psi_{1,x}(\N_{5;\ell_n} \cup \{1,2,3\})$, for example, involves nodes 1, 2, 3, and 5 through $\ell_n+5$, so it depends on $(X_1,Z_1)\equiv (x,z)$ and $(X_2,Z_2)\equiv (y,z_2)$, and $(X_3,Z_3) \equiv (x+r_nz, z_3)$, as well as the types of nodes 5 through $\ell_n+5$, which are random.
\end{secremark}

\begin{secremark} 
  In this lemma, we have four nodes of interest: two ``focal'' nodes $1,2$ and two ``ancillary'' nodes $3,4$. Since $r_n\rightarrow 0$, node $3$ is local to $1$ and $4$ is local to $2$ in terms of positions, while $1$ and $2$ are a fixed distance $\norm{x-y}$ apart and in this sense non-local. The lemma shows that, in various models with different node sets, node $1$'s statistic is asymptotically uncorrelated with node $2$'s statistic since all terms in $V_2(x,y)$ involving node 1 are independent of those involving node 2 by the coupling construction. Node $1$'s statistic asymptotically only depends on nodes $1$ and $3$ and not on nodes $2$ and $4$, and vice versa for node $2$.
\end{secremark}

\begin{proof}[Proof of \autoref{pen3.7}]
  This lemma is analogous to Lemma 3.7 of \cite{penrose2007gaussian}, but the arguments need modification for reasons discussed in \autoref{rlde}. The key step is establishing an analog of his (3.4), which is accomplished in \eqref{cm55} below.

  We will only prove convergence of arguably the most complicated of the eleven components in $V_1(x,y)$ and $V_2(x,y)$, namely that under DGP $(\star)$, for any $\epsilon>0$,
  \begin{equation}
    \prob\big( \norm{ \psi_{2,y}(\N_{5;m_n} \cup \{1,\dots,4\}) - \psi_{2,y}(\N_y \cup \{2,4\}) }_\infty > \epsilon \big) \rightarrow 0. \label{r132g09huraeb}
  \end{equation}

  \noindent Convergence of the other components follows from the same arguments. We will make use of the following fact, which is a consequence of \autoref{main-bp-mom}(b):
  \begin{align}
    \psi_m(i, \tau_{z,r_n}\bm{T}_m, \bm{\zeta}_m, \lambda_m) &\equiv \psi_m\big(i, \tau_{z,r_n}\bm{T}_m, \bm{\zeta}_m, \lambda_m(\tau_{z,r_n}\bm{T}_m, \bm{\zeta}_m)\big) \nonumber\\ &= \psi_m(i, \tau_{z',r_n}\bm{T}_m, \bm{\zeta}_m, \lambda_m) \quad\forall z,z'\in\R^d, i \in \N_m, m\in\mathbb{N}. \label{xidiff2}
  \end{align}

  Maintain DGP $(\star)$ in what follows. Let $\bm{X}_{5;m_n} \equiv \{X_i\}_{i=5}^{m_n+4}$, and define for any $R>0$
  \begin{equation*}
    \mathcal{E}_n(R) = \big\{ \tau_{y,r_n}\bm{X}_{5;m_n} \cap B(y,R) = \tau_{y,r_n}\mathcal{P}_{nf}^y \cap B(y,R) \big\}
  \end{equation*}

  \noindent and 
  \begin{multline*}
    \RR^*_{2} = \max\big\{ \RR(2, \tau_{y,r_n}\bm{T}_{\N_{5;m_n} \cup \{1,\ldots,4\}}, \Zeta_{\N_{5;m_n} \cup \{1,\ldots,4\}}, \lambda_{m_n+4}), \\ \RR(2, \tau_{y,r_n}\bm{T}_{\N_y \cup \{2,4\}}, \Zeta_{\N_y \cup \{2,4\}}, \lambda_{\abs{\N_y}+2}) \big\}.
  \end{multline*}

  \noindent The latter is the larger of the radii of stabilization (\autoref{rosdef}) corresponding to the node statistics $\psi_{2,y}(\N_{5;m_n} \cup \{1,\dots,4\})$ and $\psi_{2,y}(\N_y \cup \{2,4\})$, respectively.

  For any $R>0$, define $\N(y,R) = \{j\in\mathbb{N}\colon \tau_{y,r_n}X_j \in B(y,R)\}$, the set of nodes whose dilated positions lie within the $R$-ball of $y$. This is analogous to the definition of $\N_n(\mathcal{B}_i)$ in \autoref{rosdef}. By \autoref{rosdef}, if $y=\zero$, then under the event $\{R > \RR^*_{2}\}$,
  \begin{equation}
    \psi_{2,y}\big( \N_{5;m_n} \cup \{1,\ldots,4\} \big) = \psi_{2,y}\big( (\N_{5;m_n} \cup \{1,\ldots,4\}) \cap \N(y,R) \big). \label{f302hjfverbh}
  \end{equation}
  
  \noindent By \eqref{xidiff2}, this also holds for $y\neq\zero$. 
  
  Recall that under DGP $(\star)$, the positions of nodes $1$, $3$, and $4$ are $x$, $x+r_nz$, and $y+r_nw$, respectively. Since $\tau_{y,r_n}(\{x, x+r_nz, y+r_nw\}) = \{y+r_n^{-1}(x-y), y+r_n^{-1}(x+r_nz - y), w+y\}$, the intersection of this set and $B(y,R)$ is $\{w+y\}$ for $n,R$ sufficiently large. Then  
  \begin{equation}
    \eqref{f302hjfverbh} = \psi_{2,y}\big( (\N_{5;m_n} \cup \{2,4\}) \cap \N(y,R) \big) \label{gt2u093nj}
  \end{equation}

  \noindent for such $n,R$ under $\{R > \RR^*_{2}\}$. Furthermore, under the event $\mathcal{E}_n(R) \cap \{R > \RR^*_{2}\}$,
  \begin{equation*}
    \eqref{gt2u093nj} = \psi_{2,y}\big( (\N_y \cup \{2,4\}) \cap \N(y,R) \big) = \psi_{2,y}(\N_y \cup \{2,4\}), 
  \end{equation*} 

  \noindent where the first equality uses the event $\mathcal{E}_n(R)$, which allows us to replace the ``binomial'' model with the ``Poissonized'' model, and the second equality follows from the argument for \eqref{f302hjfverbh}. Therefore, by the law of total probability,
  \begin{multline}
    \prob\big( \norm{ \psi_{2,y}(\N_{5;m_n} \cup \{1,\dots,4\}) - \psi_{2,y}(\N_y \cup \{2,4\}) }_\infty > \epsilon \big) \\ \leq \prob(\mathcal{E}_n(R)^c) + \prob(\RR^*_{2} > R), \label{cm55}
  \end{multline}

  \noindent where $\mathcal{E}_n(R)^c$ is the complement of $\mathcal{E}_n(R)$. It remains to show that the right-hand side can be made arbitrarily small by choosing $n,R$ large enough.

  We first claim that by \autoref{main-bp-exp-stab}, $\RR^*_{2} = O_p(1)$. This is not quite immediate from the assumption because, for instance, $\RR^*_{2}$ is defined in terms of the point sets $\tau_{y,r_n}\bm{X}_{5;m_n}$ and $\tau_{y,r_n}\mathcal{P}_{nf}^y$ rather than $r_n^{-1}\bm{X}_{5;m_n}=\tau_{\zero,r_n}\bm{X}_{5;m_n}$ and $r_n^{-1}\mathcal{P}_{nf}^y = \tau_{\zero,r_n}\mathcal{P}_{nf}^y$ as in the assumption. Invoking \eqref{xidiff2} with $z=y$ and $z' = \zero$, we have
  \begin{multline*}
    \psi_{m_n+4}(2, \tau_{y,r_n}\bm{T}_{\N_{5;m_n} \cup \{1,\ldots,4\}}, \Zeta_{\N_{5;m_n} \cup \{1,\ldots,4\}}, \lambda_{m_n+4}) \\ = \psi_{m_n+4}(2, \tau_{\zero,r_n}\bm{T}_{\N_{5;m_n} \cup \{1,\ldots,4\}}, \Zeta_{\N_{5;m_n} \cup \{1,\ldots,4\}}, \lambda_{m_n+4}).
  \end{multline*}
  
  \noindent Consequently, 
  \begin{multline*}
    \RR(2, \tau_{y,r_n}\bm{T}_{\N_{5;m_n} \cup \{1,\ldots,4\}}, \Zeta_{\N_{m_n+4} \cup \{1,\ldots,4\}}, \lambda_{m_n+4}) \\ = \RR(2, \tau_{\zero,r_n}\bm{T}_{\N_{5;m_n} \cup \{1,\ldots,4\}}, \Zeta_{\N_{m_n+4} \cup \{1,\ldots,4\}}, \lambda_{m_n+4}). 
  \end{multline*}
  
  \noindent Since $r_n^{-1}z = \tau_{\zero,r_n}z$, the right-hand side of the previous display is $O_p(1)$ by \autoref{main-bp-exp-stab}. By the same argument, 
  \begin{equation*}
    \RR(2, \tau_{y,r_n}\bm{T}_{\N_y \cup \{2,4\}}, \Zeta_{\N_y \cup \{2,4\}}, \lambda_{\abs{\N_y}+2}) = \RR(2, \tau_{\zero,r_n}\bm{T}_{\N_y \cup \{2,4\}}, \Zeta_{\N_y \cup \{2,4\}}, \lambda_{\abs{\N_y}+2}) = O_p(1).
  \end{equation*}

  \noindent Hence, $\RR^*_{2} = O_p(1)$, as desired.

  It follows that, for any $\epsilon>0$, we can choose $R$ large enough such that for all $n$ sufficiently large, $\prob(\RR^*_{2} > R) < \epsilon/2$. Furthermore, by \autoref{3.1py}, for any such $R$, we can choose $n$ large enough such that $\prob(\mathcal{E}_n(R)^c) < \epsilon/2$. Combining these facts with \eqref{cm55}, we obtain \eqref{r132g09huraeb}.
\end{proof}

The next lemma shows that the add-one costs $\Xi_\ell$ and $\Xi_m$ defined in \eqref{add1} are asymptotically uncorrelated for $\ell,m$ large. 

\begin{seclemma}\label{pen5.1}
  Let $\{\nu_n\}_{n\in\mathbb{N}} \subseteq \mathbb{N}$ be a diverging sequence such that $\nu_n/n \rightarrow 0$. Under Assumptions \ref{main-bp-exp-stab} and \ref{main-bp-mom},
  \begin{equation*}
    \sup_{n-\nu_n \leq \ell < m \leq n+\nu_n} \big|\, \E[\Xi_\ell \Xi_m] - \E[\Xi_{N_n}]^2 \big| \rightarrow 0.
  \end{equation*}
\end{seclemma}
\begin{proof}
  We follow the proof of Lemma 5.1 of \cite{penrose2007gaussian}. The main difference is that we use our \autoref{pen3.7} in place of their Lemma 3.7. For $i\leq m$, let 
  \begin{multline*}
    \psi_m = \psi_m(m, r_n^{-1}\bm{T}_m, \Zeta_m, \lambda_m) \\ \text{and}\quad \Delta_{i,m} = \psi_{m+1}(i, r_n^{-1}\bm{T}_{m+1}, \Zeta_{m+1}, \lambda_{m+1}) - \psi_m(i, r_n^{-1}\bm{T}_m, \Zeta_m, \lambda_m).
  \end{multline*}

  \noindent Then
  \begin{align}
    \E[\Xi_\ell &\Xi_m] = \E\left[ \left( \psi_{\ell+1} + \sum_{i=1}^\ell \Delta_{i,\ell} \right) \left( \psi_{m+1} + \sum_{j=1}^\ell \Delta_{j,m} + \Delta_{\ell+1,m} + \sum_{k=\ell+2}^m \Delta_{k,m} \right) \right] \nonumber\\
			       &\begin{aligned}= \E[\psi_{\ell+1}\psi_{m+1}] &+ \ell\,\E[\Delta_{1,\ell}\psi_{m+1}] + \ell\,\E[\psi_{\ell+1}\Delta_{1,m}] + \ell(\ell-1)\E[\Delta_{1,\ell}\Delta_{2,m}] \\
			       &+ \ell\,\E[\Delta_{1,\ell}\Delta_{1,m}] + \E[\psi_{\ell+1}\Delta_{\ell+1,m}] + \ell\,\E[\Delta_{1,\ell}\Delta_{\ell+1,m}] \\
			       &+ (m-\ell-1)\E[\psi_{\ell+1}\Delta_{\ell+2,m}] + \ell(m-\ell-1)\E[\Delta_{1,\ell}\Delta_{\ell+2,m}]. \end{aligned} \label{g32q9bh0ni}
  \end{align}

  \noindent By the law of iterated expectations, 
  \begin{equation*}
    \E[\Xi_{N_n}] = \E\left[ \psi_{N_n+1} + \sum_{i=1}^{N_n} \E[\Delta_{i,N_n} \mid \bm{X}_{N_n}, N_n] \right].
  \end{equation*}
  
  \noindent This is the expectation of a pointwise sum (i.e.\ over the points of $\bm{X}_{N_n}$) of functionals of a Poisson point process $\bm{X}_{N_n}$. We may then apply Corollary 3.2.3 of \cite{schneider2008stochastic} and use a change of variables to obtain
  \begin{multline}
    \E[\Xi_{N_n}] = \int_{\R^{d_z}} \int_{\R^d} \E^\star\left[ \psi_{1,x}(\N_x \cup \{1\}) \right] \,\text{d}\Phi_z(z_1 \mid x) f(x) \,\text{d}x \\
    + \kappa \int_{\R^{d_z}} \int_{\R^{d_z}} \int_{\R^d} \int_{\R^d} \E^\star\left[ \psi_{1,x}(\N_x \cup \{1,3\}) - \psi_{1,x}(\N_x \cup \{1\}) \right] \\ 
    \times \,\text{d}\Phi_z(z_3 \mid x+r_nz) \,\text{d}\Phi_z(z_1 \mid x)  f(x+r_nz) f(x) \,\text{d}z \,\text{d}x, \label{schneid}
  \end{multline}
    
  \noindent where $\E^\star[\cdot]$ denotes an expectation taken under DGP $(\star)$ defined prior to \autoref{pen3.7}. Recall that under this DGP, $\psi_{1,x}(\N_x \cup \{1\})$ is a function of $(x,z_1)$ and $\psi_{1,x}(\N_x \cup \{1,3\})$ is a function of $(x,z_1)$ and $(x+r_nz,z_3)$. 

  We need to derive the ``limits'' of each of the terms in \eqref{g32q9bh0ni} and show they match with the corresponding term of $\E[\Xi_{N_n}]^2$. We only demonstrate for the first two terms as the argument for the others is similar. 

  \bigskip
  
  \noindent {\bf Term 1.} This term provides the simplest illustration of the logic. Using \eqref{xidiff2} and $m>\ell\geq 1$, the first term $\E[\psi_{\ell+1}\psi_{m+1}]$ equals
  \begin{equation*}
    \int_{\R^{d_z}} \int_{\R^{d_z}} \int_{\R^d} \int_{\R^d} \E^\star\big[ \psi_{1,x}(\N_{5;\ell} \cup \{1\}) \psi_{2,y}(\N_{5;m-1} \cup \{1,2\}) \big] \,\text{d}\Phi_z(z_1 \mid x) \,\text{d}\Phi_z(z_2 \mid y) f(x) f(y) \,\text{d}x \,\text{d}y,
  \end{equation*}

  \noindent where $\E^\star[\cdot]$ denotes an expectation taken under DGP $(\star)$ defined prior to \autoref{pen3.7}. Recall that under this DGP, $\psi_{1,x}(\N_{5;\ell} \cup \{1\})$ is a function of $(x,z_1)$ and $\psi_{2,y}(\N_{5;m-1} \cup \{1,2\})$ a function of $(x,z_1)$ and $(y,z_2)$. Since $\ell,m$ are order $n$, by \autoref{pen3.7}, 
  \begin{equation*}
    \abs{ \psi_{1,x}(\N_{5;\ell} \cup \{1\}) \psi_{2,y}(\N_{5;m-1} \cup \{1,2\}) - \psi_{1,x}(\N_x \cup \{1\}) \psi_{2,y}(\N_y \cup \{2\}) } \plimarrow 0 
  \end{equation*}

  \noindent and $\psi_{1,x}(\N_x \cup \{1\}) \indep \psi_{2,y}(\N_y \cup \{2\})$. By the Vitali convergence theorem, 
  \begin{multline*}
    \abs{ \E^\star\left[ \psi_{1,x}(\N_{5;\ell} \cup \{1\}) \psi_{2,y}(\N_{5;m-1} \cup \{1,2\}) \right] \\ - \E^\star\left[ \psi_{1,x}(\N_x \cup \{1\}) \right] \E^\star\left[ \psi_{2,y}(\N_y \cup \{2\}) \right] } \rightarrow 0.
  \end{multline*}

  \noindent Then using \autoref{main-bp-mom} and the dominated convergence theorem,
  \begin{equation*}
    \bigg| \E[\psi_{\ell+1}\psi_{m+1}] - \left( \int_{\R^{d_z}} \int_{\R^d} \E^\star\left[ \psi_{1,x}(\N_x \cup \{1\}) \right] \,\text{d}\Phi_z(z_1 \mid x) f(x) \,\text{d}x \right)^2 \bigg| \rightarrow 0.
  \end{equation*}

  \noindent {\bf Term 2.} Consider the second term on the right-hand side of \eqref{g32q9bh0ni}. For $m>\ell\geq 1$,
  \begin{multline*}
    \ell\,\E[\Delta_{1,\ell}\psi_{m+1}] = \ell\,\E\big[ \left( \psi_{\ell+1}(1, r_n^{-1}\bm{T}_{\ell+1}, \Zeta_{\ell+1}, \lambda_{\ell+1}) - \psi_\ell(1, r_n^{-1}\bm{T}_\ell, \Zeta_\ell, \lambda_{\ell}) \right) \\
    \cdot \psi_{m+1}(m+1, r_n^{-1}\bm{T}_{m+1}, \Zeta_{m+1}, \lambda_{m+1}) \big] \\
    = \ell\,\E\left[ \left( \psi_{1,\zero}(\N_{5;\ell-1} \cup \{1,3\}) - \psi_{1,\zero}(\N_{5;\ell-1} \cup \{1\}) \right) \psi_{2,\zero}(\N_{5;m-2} \cup \{1,2,3\}) \right]. 
  \end{multline*}

  \noindent By \eqref{xidiff2} and a change of variables, this equals
  \begin{multline}
    \ell\, r_n^d \int_{\R^{d_z}} \int_{\R^{d_z}} \int_{\R^{d_z}} \int_{\R^d} \int_{\R^d} \int_{\R^d} \E^\star\left[ \big( \psi_{1,x}(\N_{5;\ell-1} \cup \{1,3\}) \right. \\ \left. - \psi_{1,x}(\N_{5;\ell-1} \cup \{1\}) \big) \psi_{2,y}(\N_{5;m-2} \cup \{1,2,3\}) \right] \\ \times \,\text{d}\Phi_z(z_1 \mid x) \,\text{d}\Phi_z(z_2 \mid y) \,\text{d}\Phi_z(z_3 \mid x+r_nz) f(x) f(y) f(x+r_nz) \,\text{d}z \,\text{d}x \,\text{d}y, \label{chofva}
  \end{multline}

  \noindent where $\E^\star[\cdot]$ denotes an expectation taken under DGP $(\star)$ defined prior to \autoref{pen3.7}. Recall that under this DGP, $\psi_{1,x}(\N_{5;\ell-1} \cup \{1,3\})$, $\psi_{1,x}(\N_{5;\ell-1} \cup \{1\})$, and $\psi_{2,y}(\N_{5;m-2} \cup \{1,2,3\})$ are functions of $(x,z_1)$, $(x+r_nz,z_3)$, and/or $(y,z_2)$.  Following the argument for Term 1, by \autoref{pen3.7}, \autoref{main-bp-mom}, and the Vitali convergence theorem,
  \begin{multline}
    \E^\star\left[ \big( \psi_{1,x}(\N_{5;\ell-1} \cup \{1,3\}) - \psi_{1,x}(\N_{5;\ell-1} \cup \{1\}) \big) \psi_{2,y}(\N_{5;m-2} \cup \{1,2,3\}) \right] \\ -
    \E^\star\left[ \psi_{1,x}(\N_x \cup \{1,3\}) - \psi_{1,x}(\N_x \cup \{1\}) \right] \E^\star\left[ \psi_{2,y}(\N_y \cup \{2\}) \right] \rightarrow 0. \label{integrand}
  \end{multline}

  Next, to apply the dominated convergence theorem, we show that the integrand of \eqref{chofva} is uniformly bounded by an integrable function. Under DGP $(\star)$ and the event
  \begin{equation*}
    \norm{z} > \RR(1, \tau_{x,r_n}\bm{T}_{\N_{5;\ell-1} \cup \{1,3\}}, \Zeta_{\N_{5;\ell-1} \cup \{1,3\}}, \lambda_{\ell+1}),
  \end{equation*}

  \noindent the distance between the scaled positions of nodes 1 and 3 (respectively $r_n^{-1}x$ and $r_n^{-1}(x+r_nz)$) exceeds the right-hand side since $r_n^{-1}\norm{x - (x+r_nz)} = \norm{z}$. By \autoref{rosdef}, this implies $\psi_{1,x}(\N_{5;\ell-1} \cup \{1,3\}) = \psi_{1,x}(\N_{5;\ell-1} \cup \{1\})$, in which case the integrand of \eqref{chofva} is zero. By \autoref{main-bp-mom} and the H\"{o}lder and Minkowski inequalities, there exist $C,\epsilon>0$ such that for $n$ large, 
  \begin{multline*}
    \abs{\E^\star\left[ \big( \psi_{1,x}(\N_{5;\ell-1} \cup \{1,3\}) - \psi_{1,x}(\N_{5;\ell-1} \cup \{1\}) \big) \psi_{2,y}(\N_{5;m-2} \cup \{1,2,3\}) \right]} \\ \leq
    \big( \norm{\psi_{1,x}(\N_{5;\ell-1} \cup \{1,3\})}_p^\star + \norm{\psi_{1,x}(\N_{5;\ell-1} \cup \{1\})}_p^\star \big) \cdot \norm{\psi_{2,y}(\N_{5;m-2} \cup \{1,2,3\})}_p^\star \\ \times \prob^\star\big(\RR(1, \tau_{x,r_n}\bm{T}_{\N_{5;\ell-1} \cup \{1,3\}}, \Zeta_{\N_{5;\ell-1} \cup \{1,3\}}, \lambda_{\ell+1}) > \norm{z}\big)^{1-2/p} \leq C \tau_{b,\epsilon}(\norm{z})^{1-2/p},
  \end{multline*}

  \noindent where $p$ is defined in \autoref{main-bp-mom} and, for any random variable $W$, $\norm{W}_p^\star \equiv \sup_{x,z_1,y,z_2,z,z_3} \E^\star[\abs{W}^p]^{1/p}$, and $\prob^\star$ denotes the distribution under DGP $(\star)$. By \autoref{main-bp-exp-stab},
  \begin{multline*}
    C \ell r_n^d \int_{\R^d} \int_{\R^d} \int_{\R^d} \tau_{b,\epsilon}(\norm{z})^{1-2/p} f(x) f(y) f(x+r_nz) \,\text{d}x \,\text{d}y \,\text{d}z \\ \leq C \kappa \sup_w f(w) \int_{\R^d} \tau_{b,\epsilon}(\norm{z})^{1-2/p} \,\text{d}z < \infty, 
  \end{multline*}

  \noindent which verifies the dominance condition. Therefore, by \eqref{chofva}, \eqref{integrand} and dominated convergence, 
  \begin{multline*}
    \bigg| \ell\,\E[\Delta_{1,\ell}\psi_{m+1}] - \kappa\int_{\R^{d_z}} \int_{\R^d} \E^\star\left[ \psi_{2,y}(\N_y \cup \{2\}) \right] \,\text{d}\Phi_z(z_2 \mid y) f(y) \,\text{d}y \\ \times \int_{\R^{d_z}} \int_{\R^{d_z}} \int_{\R^d} \int_{\R^d} \E^\star\left[ \psi_{1,x}(\N_x \cup \{1,3\}) - \psi_{1,x}(\N_x \cup \{1\}) \right] \\ \times \,\text{d}\Phi_z(z_1 \mid x) \,\text{d}\Phi_z(z_3 \mid x+r_nz) f(x) f(x+r_nz) \,\text{d}z \,\text{d}x \bigg| \rightarrow 0.
  \end{multline*}
\end{proof}

The last lemma shows that add-one costs have bounded moments. 

\begin{seclemma}\label{pen5.2}
  Let $\{\nu_n\}_{n\in\mathbb{N}} \subseteq \mathbb{N}$ be a diverging sequence such that $\nu_n/n \rightarrow 0$. Under Assumptions \ref{main-bp-exp-stab} and \ref{main-bp-mom}, for $\tilde\sigma_n^2$ defined in \autoref{spois},
  \begin{equation*}
    \limsup_{n\rightarrow\infty} \sup_{n-\nu_n \leq m \leq n+\nu_n} \E[\Xi_m^2] < \infty \quad\text{and}\quad \sup_n \max\big\{\E[\Xi_{N_n}], \tilde\sigma_n^2\big\} < \infty.
  \end{equation*}
\end{seclemma}
\begin{proof}
  Recall the characterization of $\E[\Xi_{N_n}]$ in \eqref{schneid}. The first term on the right-hand side of this equation is uniformly bounded in $n$ by \autoref{main-bp-mom}. For the second term, consider the quantity inside $\E^\star[\cdot]$. Under DGP $(\star)$ and the event $\norm{z} > \RR(1, \tau_{x,r_n}\bm{T}_{\N_x \cup \{1,3\}}, \Zeta_{\N_x \cup \{1,3\}}, \lambda_{\abs{\N_x}+2})$, since node 1 has position $x$ and node 3 has position $x+r_nz$, their scaled distance satisfies $r_n^{-1}\norm{x-(x+r_nz)} > \RR(1, \tau_{x,r_n}\bm{T}_{\N_x \cup \{1,3\}}, \Zeta_{\N_x \cup \{1,3\}}, \lambda_{\abs{N_x}+2})$. Hence, by definition of the radius of stabilization, $\psi_{1,x}(\N_x \cup \{1,3\}) = \psi_{1,x}(\N_x \cup \{1\})$. Therefore, by \autoref{main-bp-mom} and the H\"{o}lder and Minkowski inequalities, there exist $C,\epsilon>0$ such that for $n$ sufficiently large, the integrand of the second term on the right-hand side is bounded above in absolute value by
  \begin{multline}
    \big( \norm{\psi_{1,x}(\N_x \cup \{1,3\})}_p^\star + \norm{\psi_{1,x}(\N_x \cup \{1\})}_p^\star \big) f(x) f(x+r_nz) \\ \times \prob^\star\big(\RR(1, \tau_{x,r_n}\bm{T}_{\N_x \cup \{1,3\}}, \Zeta_{\N_x \cup \{1,3\}}, \lambda_{\abs{N_x}+2}) > \norm{z}\big)^{1-1/p} \\ \leq C  f(x) \sup_w f(w) \tau_p(\norm{z})^{1-1/p}, \label{gfj092qbhw}
  \end{multline}

  \noindent where for any random variable $W$, $\norm{W}_p^\star \equiv \sup_{x,z_1,z,z_3} \E^\star[\abs{W}^p]^{1/p}$, and $\prob^\star$ and $\E^\star[\cdot]$ are respectively the probability and expectation taken under DGP $(\star)$. The right-hand side of \eqref{gfj092qbhw} does not depend on $n$ and has a finite integral over $x,z$ by \autoref{main-bp-exp-stab}. Therefore $\sup_n \E[\Xi_{N_n}] < \infty$.

  The result for $\E[\Xi_m^2]$ can be established using similar arguments; see Lemma 5.2 of \cite{penrose2007gaussian}. Since the effort of translating the proof to our notation is the same as what was done for the previous two lemmas, we omit the argument.

  Finally, similar to \eqref{schneid}, 
  \begin{multline*}
    \tilde\sigma_n^2 = \int_{\R^{d_z}} \int_{\R^d} \E^\star\left[ \psi_{1,x}(\N_x \cup \{1\})^2 \right] \,\text{d}\Phi_z(z_1 \mid x) f(x) \,\text{d}x \\ + \kappa \int_{\R^{d_z}} \int_{\R^{d_z}} \int_{\R^d} \int_{\R^d} \bigg( \E^\star\left[ \psi_{1,x}(\N_x \cup \{1,3\}) \psi_{3,x}(\N_x \cup \{3,1\}) \right] \\ - \E^\star\left[ \psi_{1,x}(\N_x \cup \{1\}) \right] \E^\star\left[ \psi_{3,x}(\N_{x+r_nz} \cup \{3\}) \right] \bigg)\\ \times \,\text{d}\Phi_z(z_1 \mid x) \,\text{d}\Phi_z(z_3 \mid x+r_nz) f(x) f(x+r_nz) \,\text{d}z \,\text{d}x. 
  \end{multline*}

  \noindent The first term is uniformly bounded by \autoref{main-bp-mom}. The second term is also uniformly bounded by an argument similar to \eqref{gfj092qbhw}. 
\end{proof}

\begin{proof}[Proof of \autoref{pen2.3}]
  We largely follow the proof of Theorem 2.12 of \cite{penrose2003} and Theorem 2.3 in \cite{penrose2007gaussian}, although the characteristic function argument at the end is different since we avoid deriving the limiting variance. Abbreviate $\alpha_n = \E[\Xi_{N_n}]$ and $\Lambda_m = \sum_{i=1}^m \psi_i(\N_m)$. We first prove that
  \begin{equation}
    \E\left[ n^{-1}\left( \Lambda_{N_n} - \Lambda_n - (N_n-n)\alpha_n \right)^2 \right] \rightarrow 0. \label{2.43}
  \end{equation}

  \noindent The left-hand side equals
  \begin{multline}
    \sum_{m\colon |m-n|\leq n^{3/4}} \E\left[ n^{-1}\left( \Lambda_m - \Lambda_n - (m-n)\alpha_n \right)^2 \right] \prob(N_n=m) \\
    + \E\left[ n^{-1}\left( \Lambda_{N_n} - \Lambda_n - (N_n-n)\alpha_n \right)^2 \mathbf{1}\{\abs{N_n-n} > n^{3/4}\} \right]. \label{2.44}
  \end{multline}

  \noindent Let $\varepsilon>0$. For $n$ sufficiently large and $m \in [n, n+n^{3/4}]$,
  \begin{multline*}
    \E\left[ \left( \Lambda_m - \Lambda_n - (m-n)\alpha_n \right)^2 \right] 
    = \E\left[ \left( \sum_{\ell=n}^{m-1} (\Xi_\ell - \alpha_n) \right)^2 \right] \\
    = \sum_{\ell=n}^{m-1} \E\left[ (\Xi_\ell - \alpha_n)^2 \right] + \sum_{\ell=n}^{m-1} \sum_{\substack{k=n\\k\neq\ell}}^{m-1} \E\left[ (\Xi_\ell - \alpha_n) (\Xi_k - \alpha_n) \right] \\
    \leq \varepsilon \sqrt{n}(m-n) + \varepsilon (m-n)^2, 
  \end{multline*}

  \noindent where the inequality follows from Lemmas \ref{pen5.1} and \ref{pen5.2}. A similar argument also applies to $m \in [n-n^{3/4}, n]$. Then the first term of \eqref{2.44} is bounded by
  \begin{equation*}
    n^{-1}\E\left[ \varepsilon\, n^{1/2}\abs{N_n-n} + \varepsilon(N_n-n)^2 \right] \leq n^{-1}\varepsilon  \left( n^{1/2} \left( \E\abs{N_n-n}^2 \right)^{1/2} + n \right) = 2\varepsilon. 
  \end{equation*}

  Following the argument in the proof of Theorem 2.3 in \cite{penrose2007gaussian}, the second term of \eqref{2.44} can be bounded by a constant times
  \begin{equation*}
    n^{-1}n^{(2p+1)/p} \prob\left( \abs{N_n-n} > n^{3/4} \right)^{1-2/p'} 
  \end{equation*}

  \noindent for some $p' \in (2,p)$ and $p$ defined in \autoref{main-bp-mom}. This argument uses \autoref{main-bp-mom} and \autoref{py4.3}. The previous display tends to zero since the probability decays exponentially with $n$ by Lemma 1.4 of \cite{penrose2003}. This establishes \eqref{2.43}.

  Next consider the identity
  \begin{equation}
    n^{-1/2}\Lambda_{N_n} = n^{-1/2}\Lambda_n + n^{-1/2}(N_n-n)\alpha_n + n^{-1/2}\left( \Lambda_{N_n} - \Lambda_n - (N_n-n)\alpha_n \right). \label{tiden}
  \end{equation}

  \noindent The variance of the last term tends to zero by \eqref{2.43}. The second term on the right-hand side has variance $\alpha_n^2$ and is independent of the first term by the coupling construction. Therefore,
  \begin{equation}
    \frac{1}{n} \var\left( \Lambda_{N_n} \right) = \frac{1}{n} \var\left( \Lambda_n \right) + \alpha_n^2 + o(1) \quad\Rightarrow\quad \sigma_n^2 = \tilde{\sigma}_n^2 - \alpha_n^2 + o(1), \label{vardecomp}
  \end{equation}

  \noindent and $\sup_n \sigma_n^2 < \infty$ since $\sup_n \max\{\alpha_n,\tilde\sigma_n^2\} < \infty$ by \autoref{pen5.2}. 
  
  Since $\sigma_n^2$ has strictly positive limit infimum by assumption, so does $\tilde{\sigma}_n^2$. We can then rewrite \eqref{tiden} as
  \begin{equation}
    n^{-1/2}\tilde\sigma_n^{-1}(\Lambda_{N_n}-\E[\Lambda_{N_n}]) - o_p(1) = \frac{\sigma_n}{\tilde\sigma_n} n^{-1/2}\sigma_n^{-1}(\Lambda_n -\E[\Lambda_{N_n}]) + \frac{\alpha_n}{\tilde\sigma_n} n^{-1/2}(N_n-n). \label{tiden2}
  \end{equation}

  \noindent Let $\tilde\varphi_n(\cdot)$ be the characteristic function of the left-hand side, $\varphi_n(\cdot)$ that of $n^{-1/2}\sigma_n^{-1}(\Lambda_n-\E[\Lambda_{N_n}])$, and $\varphi_p(\cdot)$ that of $n^{-1/2}(N_n-n)$. By independence of $N_n$, due to the coupling construction, the characteristic function of the right-hand side of \eqref{tiden2} is
  \begin{equation*}
    \varphi_n\left( \frac{\sigma_n}{\tilde\sigma_n} x \right) \varphi_p\left( \frac{\alpha_n}{\tilde\sigma_n} x \right), \quad\text{implying}\quad \varphi_n\left( y \right) = \tilde\varphi_n\left( \frac{\tilde\sigma_n}{\sigma_n}y \right) \varphi_p\left( \frac{\alpha_n}{\sigma_n} y \right)^{-1} 
  \end{equation*} 

  \noindent for $y = \sigma_n \tilde\sigma_n^{-1} x$. By \autoref{poi-clt} and the Poisson CLT,
  \begin{equation*}
    \tilde\varphi_n(x) \rightarrow e^{-x^2/2} \quad\text{and}\quad \varphi_p(x) \rightarrow e^{-x^2/2}. 
  \end{equation*} 

  \noindent Furthermore, this convergence is uniform over $x$ in a compact set \citep[][Theorem 15.23]{klenke2013probability}. As previously discussed, $\sup_n \max \{\alpha_n, \tilde\sigma_n\} < \infty$ and $\liminf_{n\rightarrow\infty} \min\{\sigma_n, \tilde\sigma_n\} > 0$. Thus, for any $y\in\R^d$, there exists a compact set $K$ such that for $n$ sufficiently large, $y\tilde\sigma_n / \sigma_n $ and $y\alpha_n / \sigma_n$ lie in $K$. Combining these facts,
  \begin{equation*}
    \varphi_n(y) = \frac{\text{exp}\left\{ -\frac{\tilde\sigma_n^2}{\sigma_n^2} \frac{y^2}{2} \right\} + o(1)}{\text{exp}\left\{ -\frac{\alpha_n^2}{\sigma_n^2} \frac{y^2}{2} \right\} + o(1)} = \frac{\text{exp}\left\{ -\frac{\tilde\sigma_n^2-\alpha_n^2}{\sigma_n^2} \frac{y^2}{2} \right\} + o(1)}{1 + o(1)} \rightarrow e^{-y^2/2}
  \end{equation*}
  
  \noindent for any $y\in\R$, where convergence follows from \eqref{vardecomp}. Therefore,
  \begin{equation*}
    n^{-1/2} \sigma_n^{-1} \left( \Lambda_n - \E\left[ \Lambda_{N_n} \right] \right) \dlimarrow \N(0,1). 
  \end{equation*}

  \noindent Finally, in the previous expression, we can replace $\E[\Lambda_{N_n}]$ with $\E[\Lambda_n]$ since \eqref{2.43} implies that $\E[n^{-1/2}(\Lambda_{N_n} - \Lambda_n - (N_n-n)\alpha_n)] \rightarrow 0$.
\end{proof}

\subsection{Proof of \autoref{master-clt}}

Recall that $d_\psi$ is the dimension of the range of $\psi_n(\cdot)$. For the case $d_\psi=1$, the result follows directly from \autoref{pen2.3}. For the case $d_\psi>1$, let $t \in \R^{d_\psi}\backslash\{\zero\}$. The 1-dimensional statistic $\tilde{\psi}_i(\N_n) \equiv t'\psi_i(\N_n)$ satisfies Assumptions \ref{main-bp-exp-stab} and \ref{main-bp-mom} since $\psi_n(\cdot)$ satisfies them, so the result follows from the Cram\'{e}r-Wold device. \qed

\section{\autoref{primclt}}\label{sprimclt}

Because we are verifying \autoref{main-bp-exp-stab}, we consider a model with nodes $\N_m$ and sparsity parameter $r_n$ with $m,n\in\mathbb{N}$. Let $\bm{A} = \lambda(r_n^{-1}\bm{T}_m, \Zeta_m)$, per \autoref{aesel}. From the same primitives $(r_n^{-1}\bm{T}_m, \Zeta_m)$, construct $\bm{D} \equiv \bm{D}(m)$, the network on $\N_m$ with no self links and $ij$th entry \eqref{Dij}. Let $C_i \equiv C_i(m)$ be $i$'s component in $\bm{D}(m)$, and define $C_i^+ \equiv C_i^+(m)$ as in \eqref{stratneigh}. Let $\bm{M} \equiv \bm{M}(m)$ be the network on $\N_m$ with no self links and $ij$th entry
\begin{equation}
  M_{ij} = \ind\big\{ \sup_s V(r_n^{-1}\norm{X_i-X_j}, s, Z_i, Z_j, \zeta_{ij}) > 0 \big\}. \label{Mnet}
\end{equation}

\noindent This is the network of potential links that are not robustly absent. Finally, for any $J \subseteq \N_m$ and $i \in J$, abbreviate
\begin{equation*}
  \psi_i(J) \equiv \psi_{\abs{J}}(i, r_n^{-1}\bm{T}_J, \Zeta_J, \lambda_{\abs{J}}) \equiv \psi_{\abs{J}}\big(i, r_n^{-1}\bm{T}_J, \Zeta_J, \lambda_{\abs{J}}(r_n^{-1}\bm{T}_J, \Zeta_J)\big).
\end{equation*}

The following lemma constructs for each node $i$ a set $J_i$ satisfying \eqref{Jigoal}. 

\begin{seclemma}[\cite{leung2016}, Lemma 1]\label{constructJT}
  Under Assumptions \ref{S}, \ref{dyklocal}, and \ref{nocoord}, for any $m,n\in\mathbb{N}$ and $i\in\N_m$, $\psi_i(\N_m) = \psi_i(J_i)$ for
  \begin{equation*}
    J_i \equiv J_i(m) = \bigcup \left\{ C_j^+(m)\colon j \in \N_{\bm{M}(m)}(i,K) \right\}. 
  \end{equation*}
\end{seclemma}

Let $\tilde{C}_i \equiv \tilde{C}_i(m) = \medcup_{j \in C_i(m)} \N_{\bm{M}(m)}(j,1)$, and note that $C_i^+(m) \subseteq \tilde{C}_i(m)$. It will be more convenient to stochastically bound the following set which contains $J_i$:
\begin{equation}
  \tilde{J}_i \equiv \tilde{J}_i(m) = \bigcup \left\{ \tilde{C}_j(m)\colon j \in \N_{\bm{M}(m)}(i,K) \right\}.
  \label{Ji}
\end{equation}

\subsection{Branching Processes}\label{introbp}

We next define branching processes used to stochastically bound $\abs{\tilde C_i}$ and $\abs{\N_{\bm{M}}(i,K)}$. These are used to construct a process whose size stochastically dominates $\abs{\tilde{J}_i}$ from \eqref{Ji}. A branching process is a discrete-time stochastic process in which each time period is associated with a set of elements that are the ``offspring'' of elements of the previous time period. The process starts at time zero with only one element. 

The use of branching processes to bound component and neighborhood sizes is a well-known technique in random graph theory \citep[e.g.][]{bollobas2012simple}. The basic idea is to explore the component of a node $i$ using a breadth-first search where at each step we uncover the neighbors of nodes uncovered in the previous step. If the network formation model is dyadic (as is the case for $\bm{D}$ and $\bm{M}$), then the set of nodes uncovered at a given step can be stochastically bounded by the number of offspring produced at a given time period of an appropriate branching process. If the branching process is subcritical in that number of offspring declines over time, then the size of $i$'s component will be stochastically bounded.

Let $x \in \R^d$, $z \in \R^{d_z}$, $r > 0$, and $\bar{f} = \sup_{x\in\R^d} f(x)$. Recall the definition of $\Phi^*$ from \autoref{dfrag} and $p_1(x,z;x',z')$ from \eqref{L_r}. Denote by $\mathfrak{X}^D_{r}(x,z)$ the multi-type Galton-Walton branching process \citep[see e.g.][\S 2.1]{bollobas_phase_2007} with type space $\R^d \times \R^{d_z}$ starting at a particle of type $(x,z)$ with the following two properties. First, each particle of type $(x',z')$ is replaced in the next generation by a set of particles (its ``offspring'') distributed as a Poisson point process on $\R^d \times \R^{d_z}$ with intensity
\begin{equation}
  \text{d}\varphi_{r}(x',z'; x'',z'') = \kappa\bar{f} (1+r) p_1(x',z';x'',z'') \,\text{d}\Phi^*(z'') \,\text{d}x''. \label{Dintens}
\end{equation}

\noindent Second, conditional on particle types at generation $t-1$, the Poisson point processes that create the particles of generation $t$ are independent across generation $t-1$ particles. 

We can interpret $\mathfrak{X}^D_{r}(x,z)$ as a tree network. The root of the tree, corresponding to the first generation of the process, consists of a single node (particle) of type $(x,z)$. Its neighbors are its offspring, which constitute the second generation and have types distributed according to the Poisson point process described above. Conditional on the types of the second generation, each node in the second generation of type $(x',z')$ independently generates neighbors with types distributed as the Poisson point process above. These constitute the third generation. The process is repeated indefinitely. Let $\abs{\mathfrak{X}^D_r(x,z)}$ denote the number of particles ultimately generated, which may be infinite. We write $(x',z') \in \mathfrak{X}^D_r(x,z)$ to mean that $(x',z')$ is a type of a particle generated at some point in the process, viewing $\mathfrak{X}^D_r(x,z)$ as the set of all such particles.

We next define a ``fixed-depth'' branching process, which is required to terminate at a predetermined generation. Recalling the definitions in \autoref{Vc}, let 
\begin{equation}
\bar{p}_r(x;x') = \prob\big(\bar{V}(r^{-1}\norm{x-x'},\zeta_{ij}) > 0\big). \label{barp}
\end{equation}

\noindent Let $\mathfrak{X}^M_r(x,z;K)$ denote the branching process on $\R^d \times \R^{d_z}$ starting at a particle of type $(x,z)$ that terminates at generation $K+2$, with the following two properties. First, any particle of type $(x',z')$ is replaced in the next generation by offspring with types distributed as a Poisson point process on $\R^d \times \R^{d_z}$ with intensity
\begin{equation}
  \text{d}\bar{\varphi}_r(x',z'; x'',z'') = \kappa \bar{f} (1+r) \bar{p}_1(x';x'') \,\text{d}\Phi^*(z'') \,\text{d}x''. \label{Mintens}
\end{equation}

\noindent Second, conditional on the types of generation $t-1$, the point processes that constitute generation $t$ are independent. Thus, the process is generated the same way as $\mathfrak{X}^D_r(x,z)$, except the intensity measure is different, and once the $(K+1)$-th generation is born, no further offspring are generated. Define $\abs{\mathfrak{X}^M_r(x,z)}$ and $(x',z') \in \mathfrak{X}^M_r(x,z)$ as we did for $\mathfrak{X}^D_r(x,z)$.

We use $\abs{\mathfrak{X}^D_{r_n}(X_i,Z_i)}$ and $\abs{\mathfrak{X}^M_{r_n}(X_i,Z_i;K)}$ to stochastically bound $\abs{C_i}$ and $\abs{\N_{\bm{M}}(i,K)}$, respectively. The main idea is that the number of offspring of a particle stochastically dominates the degree of a node with the same type.

\begin{seclemma}\label{Jsd}
  Let $(x,z) \in \mathcal{T} \equiv \text{supp}(X_1,Z_1)$.
  \begin{enumerate}[(a)]
    \item Under \autoref{Vc}, for any $m,n\in\mathbb{N}$ sufficiently large and $i \in \N_m$, $\abs{\N_{\bm{M}(m)}(i,K)}$ is stochastically dominated by $\abs{\mathfrak{X}^M_{r_n}(x,z;K)}$. The same result holds if we replace $m$ with $N_n+1$.

    \item Under Assumptions \ref{Vc}, \ref{dfrag}, and \ref{dreg}, part (a) holds if we replace $\abs{\N_{\bm{M}(m)}(i,K)}$ with $\abs{C_i(m)}$ and $\abs{\mathfrak{X}^M_{r_n}(x,z;K)}$ with $\abs{\mathfrak{X}^D_{r_n}(x,z)}$.
  \end{enumerate}
\end{seclemma}
\begin{proof}
  Part (b) follows from Lemma SA.3.1 of \cite{leung2019compute}. The key argument in the lemma is that we can stochastically bound the degree distribution of any node $j$ in $\bm{D}(m)$, which is conditionally binomial, with the Poisson point process with intensity \eqref{Dintens}. The case $m = N_n+1$ is simpler since the degree distribution in this case is exactly Poisson, so it follows from essentially the same arguments. The corresponding results for part (a) are simpler still because we terminate the branching process after $K+1$ generations rather than branching indefinitely. Hence, they also follow from similar arguments. For part (a), the analogous proof of Lemma SA.3.1 of \cite{leung2019compute} only needs his Assumption SA.3.2(a), which is implied by our \autoref{Vc}. Our Assumptions \ref{dfrag} or \ref{dreg} are unnecessary for (a) since these are used to control branching processes with an infinite number of periods.
\end{proof}

We lastly construct a ``hybrid'' branching process 
\begin{equation}
  \mathfrak{X}^H_{r_n}(X_i,Z_i;K)
  \label{hatfrak}
\end{equation}

\noindent whose size stochastically dominates $\abs{\tilde{J}_i}$ for sufficiently large $m,n$. First, for any $(x,z) \in \mathcal{T}$, define the smaller process 
\begin{equation*}
  \mathfrak{X}_{r_n}^+(x,z),
\end{equation*}

\noindent constructed by first generating $\mathfrak{X}^D_{r_n}(x,z)$ and then for each particle (in any generation) of type $(x',z') \in \mathfrak{X}^D_{r_n}(x,z)$ generating a branching process $\mathfrak{X}^M_{r_n}(x',z';1)$ independently across particles conditional on $\mathfrak{X}^D_{r_n}(x,z)$. Using \autoref{Jsd}, $\abs{\mathfrak{X}_{r_n}^+(X_i,Z_i)}$ stochastically dominates $\abs{\tilde{C}_i}$. 

Next construct \eqref{hatfrak} by generating the fixed-depth process $\mathfrak{X}^M_{r_n}(X_i,Z_i;K)$ and then for each particle (in any generation) of type $(x',z')$ in $\mathfrak{X}^M_{r_n}(X_i,Z_i;K)$ generating the smaller process $\mathfrak{X}_{r_n}^+(x',z')$ independently across particles conditional on $\mathfrak{X}^M_{r_n}(X_i,Z_i;K)$. Using \autoref{Jsd}, $\abs{\mathfrak{X}^H_{r_n}(X_i,Z_i;K)}$ stochastically dominates $\abs{\tilde{J}_i(m)}$, and therefore $\abs{J_i(m)}$, as desired. The previous claims are also true if we replace $m$ with $N_n+1$.

\subsection{Tail Bounds}\label{sexpstab}

The first two results respectively show that $\abs{\mathfrak{X}^D_{r_n}(x,z)}$ and $\abs{\mathfrak{X}^M_{r_n}(x,z;K)}$, defined in \autoref{introbp}, have distributions with exponential tails. Let $\bar{f} = \sup_x f(x)$, and for $\alpha,r>0$, define $g^\alpha_{r}(x,z) = \E[\alpha^{\abs{\mathfrak{X}^D_{r}(x,z)}}]$ and
\begin{equation*}
  \psi_{r}(x,z) = \kappa \bar{f} (1+r) \int_{\R^d} \left( \int_{\R^{d_z}} p_1(x,z;x',z')^2 \,\text{d}\Phi^*(z') \right)^{1/2} \,\text{d}x'. 
\end{equation*}

\begin{seclemma}\label{Dexptail}
  Under Assumptions \ref{Vc}, \ref{dfrag}, and \ref{dreg}, there exists $\alpha>1$ such that, for $r>0$ sufficiently small, $\sup_{(x,z)\in\mathcal{T}} g^\alpha_{r}(x,z) < \infty$.
\end{seclemma}
\begin{proof}
  We apply Lemma SA.3.3 of \cite{leung2019compute}, which directly delivers the conclusion. His Assumption SA.3.1 corresponds to our \autoref{dfrag}. His Assumption SA.3.2 requires $\inf_{(x,z)\in \mathcal{T}} \psi_0(x,z) > 0$ and $\sup_{(x,z)\in \mathcal{T}} \psi_0(x,z) < \infty$. \autoref{Vc} implies the second inequality, and \autoref{dreg} implies the first because
  \begin{equation*}
    \lim_{n\rightarrow\infty} \psi_{r_n}(x,z) \geq \lim_{n\rightarrow\infty} n \,\E[p_{r_n}(x,z;X_2,Z_2)].
  \end{equation*}
\end{proof}

\begin{seclemma}\label{MKexptail}
  Under \autoref{Vc}, $\sup_{r\leq \kappa} \sup_{(x,z) \in \mathcal{T}} \E[\alpha^{\abs{\mathfrak{X}^M_r(x,z;K)}}] < \infty$ for any $\alpha > 1$ and $K\in\mathbb{N}$.
\end{seclemma}
\begin{proof}
  Let $\bar{g}^{\alpha}_{r}(K)$ be the function that maps $(x,z) \in \mathcal{T}$ to $\E[\alpha^{\abs{\mathfrak{X}^M_r(x,z;K)}}]$ and $\bar{T}_{r}$ the functional that maps $h\colon \R^d \times \R^{d_z} \rightarrow \R$ to 
  \begin{equation*}
    (\bar{T}_{r}h)(x,z) = \int_{\R^d} \int_{\R^{d_z}} h(x',z') \text{d}\bar{\varphi}_{r}(x,z; x',z'), 
  \end{equation*}
  
  \noindent where $\bar{\varphi}_r(\cdot)$ is defined in \eqref{Mintens}. Define $\Gamma\colon h \mapsto \alpha\,\text{exp}\{\bar{T}_{r}(h-1)\}$. A standard branching process argument yields the recursion
  \begin{multline}
    \bar{g}^{\alpha}_{r}(K)(x,z) \\
    = \alpha\, \text{exp}\left\{ \kappa \bar{f} (1+r) \int_{\R^d} \int_{\R^{d_z}} (\bar{g}^{\alpha}_{r}(K-1)(x',z')-1)\bar{p}_r(x;x') \,\text{d}\Phi^*(z') \,\text{d}x' \right\} \\
    = \alpha\,\text{exp}\{(\bar{T}_{r}(\bar{g}^{\alpha}_{r}(K-1)-1))(x,z)\} = \alpha (\Gamma \bar{g}^{\alpha}_{r}(K-1))(x,z). \label{grecur}
  \end{multline}

  \noindent See for example Chapter 6 of \cite{mode1971multitype}. Then for any $(x,z)\in\mathcal{T}$,
  \begin{equation}
    \bar{g}^{\alpha}_{r}(K)(x,z) = \Gamma^K \alpha, \label{gak}
  \end{equation}

  \noindent by \eqref{grecur} and the fact that $\abs{\mathfrak{X}^M_r(x,z;0)} = 1$. 

  By \eqref{grecur} and \eqref{gak},
  \begin{equation}
    \Gamma\alpha = \alpha e^{\alpha-1} \text{exp}\left\{ \kappa (1+r) \bar{f} \int_{\R^d} \int_{\R^{d_z}} \bar{p}_r(x;x') \,\text{d}\Phi^*(z') \,\text{d}x' \right\}. \label{fj3902bne}
  \end{equation}

  \noindent By \autoref{Vc}, $\sup_x \int_{\R^d} \bar{p}_1(x;x') \,\text{d}x' < \infty$ since $\bar{p}_1(x;x') \leq \bar{\Phi}_\zeta( \bar{V}^{-1}(\norm{x-x'}, 0))$ decays exponentially in $\norm{x-x'}$. Consequently, there exists $C<\infty$ such that, for any $r\leq \kappa$, \eqref{fj3902bne} is bounded by $\alpha e^{\alpha-1} C$. Likewise, $\Gamma^2\alpha < \alpha e^{\alpha e^{\alpha-1} C-1} C < \infty$, and repeating this argument, we obtain $\sup_{(x,z) \in \mathcal{T}} \Gamma^K\alpha < \infty$. Combined with \eqref{gak}, this proves the claim.
\end{proof}

Using the previous two lemmas, we establish an exponential tail bound for the distribution of $\abs{\mathfrak{X}^H_{r_n}(x,z;K)}$ defined in \eqref{hatfrak}. 

\begin{seclemma}\label{JCN} 
  Under Assumptions \ref{Vc}, \ref{dfrag}, and \ref{dreg}, there exists $n_0 \in \mathbb{N}$ such that 
  \begin{equation*}
    \limsup_{w\rightarrow\infty} w^{-1} \log \left( \sup_{n>n_0} \sup_{(x,z) \in \mathcal{T}} \prob\big( \abs{\mathfrak{X}^H_{r_n}(x,z;K)} > w \big) \right) < 0.
  \end{equation*}
\end{seclemma}
\begin{proof}
  Let $\Gamma_n(x,z) = \medcup_{(x',z') \in \mathfrak{X}^M_{r_n}(x,z;K)} \mathfrak{X}^D_{r_n}(x',z')$, where conditional on particle types in $\mathfrak{X}^M_{r_n}(x,z;K)$, $\{\mathfrak{X}^D_{r_n}(x',z')\colon (x',z') \in \mathfrak{X}^M_{r_n}(x,z;K)\}$ are independent branching processes with intensity \eqref{Dintens}. Then $\mathfrak{X}^H_{r_n}(x,z;K) = \medcup_{(x',z') \in \Gamma_n(x,z)} \mathfrak{X}^M_{r_n}(x',z';1)$, where, conditional on $\Gamma_n(x,z)$, $\{\mathfrak{X}^M_{r_n}(x',z';1)\colon (x',z') \in \Gamma_n(x,z)\}$ are independent branching processes with intensity \eqref{Mintens}.
  
  By the law of total probability, for any $w' > 0$,
  \begin{multline}
    \prob\left(\abs{\mathfrak{X}^H_{r_n}(x,z;K)} > w \right) \leq \prob\left( \abs{\Gamma_n(x,z)} > w' \right) \\ 
    + \prob\bigg( \sum_{(x',z') \in \Gamma_n(x,z)} \abs{\mathfrak{X}^M_{r_n}(x',z';1)} > w \medcap \abs{\Gamma_n(x,z)} \leq w' \bigg). \label{weo}
  \end{multline}

  \noindent For any $\alpha>1$, the second term on the right-hand side equals
  \begin{multline}
    \E\bigg[ \prob\bigg( \sum_{(x',z') \in \Gamma_n(x,z)} \abs{\mathfrak{X}^M_{r_n}(x',z';1)} > w \,\bigg|\, \Gamma_n(x,z) \bigg) \mathbf{1}\big\{ \abs{\Gamma_n(x,z)} \leq w' \big\} \bigg] \\ 
    \leq \alpha^{-w}\, \E\bigg[ \prod_{(x',z') \in \Gamma_n(x,z)} \E\left[ \alpha^{\abs{\mathfrak{X}^M_{r_n}(x',z';1)}} \,\big|\, \Gamma_n(x,z) \right] \mathbf{1}\big\{ \abs{\Gamma_n(x,z)} \leq w' \big\} \bigg] \\
    \leq \alpha^{-w}\, \left( \sup_{n>n_0} \sup_{(x',z') \in \mathcal{T}} \E\left[ \alpha^{\abs{\mathfrak{X}^M_{r_n}(x',z';1)}} \right] \right)^{w'}, \label{weo2}
  \end{multline}

  \noindent where the second line uses Markov's inequality and conditional independence of the branching processes, and the third line uses independence to drop the conditioning.

  By \autoref{MKexptail}, we can choose $\alpha$ and $n_0$ such that the supremum term in \eqref{weo2} is bounded above by some $C > 1$. Then setting $w' = c\,w$ for any $c>0$,
  \begin{equation*}
    \eqref{weo2} < \text{exp}\left\{ -w \log\alpha + c\,w \log C \right\}.
  \end{equation*}

  \noindent Choosing $c$ such that $\log\alpha > c \log C$, the second element on the right-hand side of \eqref{weo} is bounded above by $e^{-\beta w}$ for some $\beta > 0$. By a similar argument,
  \begin{multline}
    \prob\left( \abs{\Gamma_n(x,z)} > w' \right) \leq \alpha^{-w'} \left( \sup_{n>n_0} \sup_{(x',z') \in \mathcal{T}} \E\left[ \alpha^{\abs{\mathfrak{X}^D_{r_n}(x',z')}} \right] \right)^{w''} \\ + \prob\left( \mathfrak{X}^M_{r_n}(x,z;K) > w'' \right). \label{weo3}
  \end{multline}

  \noindent By \autoref{Dexptail}, we can choose $\alpha,n_0$ such that the supremum term is bounded by a finite constant. Setting $w'' = c'w'$ for some small enough $c'>0$, the first term on the right-hand side can be bounded above by $e^{-\beta'w'} = e^{-\beta'cw}$ for some $\beta'>0$, recalling $w' = c\,w$. By \autoref{MKexptail}, the second term on the right-hand side of \eqref{weo3} is $O(e^{-\beta''w''})$ for some $\beta''>0$, uniformly over all $n$ sufficiently large. We have therefore established that \eqref{weo} is $O(e^{-\gamma w})$ for some $\gamma>0$.
\end{proof}

Recall the definition of $\abs{\tilde{J}_i(m)}$ from \eqref{Ji}. We next obtain an exponential tail bound for its distribution using the previous lemmas. Let $m,k\in\mathbb{N}$ and $\bm{t}_k \in \mathcal{T}^k$. Recall that $\tilde{J}_i(m+k)$ is a function of node types $r_n^{-1}\bm{T}_{m+k}$. Construct $\tilde{J}_i^\star(m+k; \bm{t}_k)$ from $\tilde{J}_i(m+k)$ by replacing the node types with $r_n^{-1}\bm{T}_{m+k}^\star(\bm{t}_k)$ defined prior to \autoref{main-bp-exp-stab}, so that $\prob(\abs{\tilde{J}_1^\star(m+k; \bm{t}_k)} > w)$ is a version of the conditional probability $\prob(\abs{\tilde{J}_1(m+k)} > w \mid \bm{T}_k=\bm{t}_k)$.

\begin{seclemma}\label{Jtails}
  Under Assumptions \ref{Vc}, \ref{dfrag}, and \ref{dreg}, for any $k \in \{1,\dots,4\}$ and $k'\in\{1,2\}$, there exist $n_0,\epsilon>0$ such that $\limsup_{w\rightarrow\infty} w^{-1} \max\{ \log \tilde\tau_{b,\epsilon}(w), \log \tilde\tau_p(w)\} < 0$, where
  \begin{align*}
    &\tilde\tau_{b,\epsilon}(w) = \sup_{n>n_0} \sup_{m \in ((1-\epsilon)n, (1+\epsilon)n)} \sup_{\bm{t}_k \in \mathcal{T}^k} \prob\big( \abs{\tilde{J}_1^\star(m+k; \bm{t}_k)} > w \big), \\
    &\tilde\tau_p(w) = \sup_{n>n_0} \sup_{\bm{t}_{k'} \in \mathcal{T}^{k'}} \prob\big( \abs{\tilde{J}_1^\star(N_n+k'; \bm{t}_{k'})} > w \big).
  \end{align*}
\end{seclemma}
\begin{proof}
  We prove the result for $\tilde\tau_{b,\epsilon}(w)$ and $k=3$. The argument for the other cases is similar. Fix $\bm{t}_k \in \mathcal{T}^k$, and construct the networks $\bm{M}$ and $\bm{D}$ under the data-generating process with structural primitives $(r_n^{-1}\bm{T}_{m+k}^\star(\bm{t}_k),\Zeta_{m+k})$. Let $C_i$ be node $i$'s component in $\bm{D}$ and $\tilde{C}_i = \medcup_{j \in C_i} \N_{\bm{M}}(j,1)$.

  For any $i\in\N_{m+k}$ and $H \subseteq \N_{m+k}$, define $Q_i(H) \equiv \tilde{J}_i^\star(m+k; \bm{t}_k) \backslash H$. We first prove 
  \begin{equation}
    \tilde{J}_1^\star(m+k; \bm{t}_k) \subseteq Q_1(\{2\}) \cup Q_2(\{1\}).
    \label{f2039jbn}
  \end{equation}
  
  \noindent Let $\bm{M}^-$ and $\bm{D}^-$ be the networks obtained from $\bm{M}$ and $\bm{D}$ respectively by deleting node 1 and its links. Fix $i \in \tilde{J}_1^\star(m+k; \bm{t}_k)$. If $i = 2$, then clearly $i \in Q_2(\{1\})$, so suppose $i\neq 2$. Then either (a) $i \in \N_{\bm{M}}(1,K)$ or (b) $i \in \tilde{C}_j$ for some $j \in \N_{\bm{M}}(1,K)$. For case (a), if $i \not\in Q_1(\{2\})$, then node 2 must lie on a path of length at most $K$ in $\bm{M}$ connecting $i$ and node 1. Therefore $i \in \N_{\bm{M}^-}(2,K)$, so $i \in Q_2(\{1\})$. For case (b), if $i \not\in Q_1(\{2\})$, then there are two possibilities. The first is that nodes 2 and $i$ are linked in $\bm{M}$, in which case $i \in \N_{\bm{M}^-}(2,K)$.  The second is that node 2 lies on some path in $\bm{D}$ connecting $i$ and some node $j$, in which case nodes 2 and $i$ are connected through a path in $\bm{D}^-$. In both cases, $i \in Q_2(\{1\})$. 
  
  Applying the logic of \eqref{f2039jbn} several times yields
  \begin{align}
    \tilde{J}_1^\star(m+k; \bm{t}_k) &\subseteq Q_1(\{3\}) \cup Q_3(\{1\}) \nonumber\\ &\subseteq \big(Q_1(\{2,3\}) \cup Q_2(\{1,3\})\big) \cup \big(Q_3(\{1,2\}) \cup Q_2(\{1,3\})\big). \label{ogic}
  \end{align}

  \noindent As shown in \autoref{introbp}, for sufficiently large $m,n$, the size of $Q_1(\{2,3\})$ (say) is stochastically dominated by $\abs{\mathfrak{X}^H_{r_n}(x,z;K)}$ defined in \eqref{hatfrak} where $(x,z)$ is node 1's type under $\bm{t}_k$. By \autoref{JCN}, the distribution of $\abs{\mathfrak{X}^H_{r_n}(x,z;K)}$ satisfies the desired exponential tail bound.
\end{proof}

The next lemma provides a tail bound for the distribution of the maximal distance between the (scaled) position of a node $i\in\N_m$ and the positions of nodes in 
\begin{equation}
  \N_{\bm{J}(m)}(i,K,K') = \bigcup_{j \in \N_{\bm{M}(m)}(i,K)} \bigcup_{k \in \N_{\bm{D}(m)}(j,K')} \N_{\bm{M}(m)}(k,1).
  \label{NJ}
\end{equation}

\noindent The relevance of the latter object is its relation to $\tilde{J}_i(m)$ defined in \eqref{Ji}. If $K'$ equals the diameter of $\bm{D}(m)$, which is the maximum path distance between any node pair, then $\N_{\bm{J}(m)}(i,K,K') = \tilde{J}_i(m)$. 

Let $m,k\in\mathbb{N}$ and $\bm{t}_k \in \mathcal{T}^k$. Let $\bm{D}(m+k)$, $\bm{M}(m+k)$, \eqref{NJ}, and $\tilde{C}_i(m+k)$ defined prior to \eqref{Ji} be defined for nodes with types $r_n^{-1}\bm{T}_{m+k}$. Let $\bm{D}^\star(m+k; \bm{t}_k)$, $\bm{M}^\star(m+k; \bm{t}_k)$, $\N_{\bm{J}^\star(m+k; \bm{t}_k)}(i,K,K')$, and $\tilde{C}_i^\star(m+k; \bm{t}_k)$ be the respective analogs of these objects obtained by replacing the node types with $r_n^{-1}\bm{T}_{m+k}^\star(\bm{t}_k)$ defined prior to \autoref{main-bp-exp-stab}. Then, for example, $\prob(\abs{\N_{\bm{J}^\star(m+k; \bm{t}_k)}(i,K,K')} > w)$ is a version of the conditional probability $\prob(\abs{\N_{\bm{J}(m+k)}(i,K,K')} > w \mid \bm{T}_k=\bm{t}_k)$. The lemma below uses this notation for the case $k=1$ with $\bm{t}_1=(x,z) \in \mathcal{T}$.

\begin{seclemma}\label{MKN}
  Under Assumptions \ref{Vc}, \ref{dfrag}, and \ref{dreg}, for any $\epsilon>0$ and $K \in \mathbb{N}$, there exist constants $n_0,c>0$ such that for $w,w'\rightarrow\infty$ with $w/w' \rightarrow \alpha \in (0,\infty)$,
  \begin{multline*}
    \max \left\{ \sup_{n>n_0} \sup_{m\in((1-\epsilon)n,(1+\epsilon)n)} \sup_{(x,z) \in \mathcal{T}} \prob\big( \max_{i \in \N_{\bm{J}^\star(m+1; (x,z))}(1,K,w')} r_n^{-1}\norm{x-X_i} > w \big), \right. \\
    \left. \sup_{n>n_0} \sup_{(x,z) \in \mathcal{T}} \prob\big( \max_{i \in \N_{\bm{J}^\star(N_n+1; (x,z))}(1,K,w')} r_n^{-1}\norm{x-X_i} > w \big) \right\} \\ = O(w^{d+2}\text{exp}\{-c\,w^{1/2}\}). 
  \end{multline*}
\end{seclemma}
\begin{proof}
  We prove the bound for the term involving $\bm{J}^\star(N_n+1; (x,z))$. The argument for the other term is almost identical except it does not require use of Corollary 3.2.3 to simplify the expectation in \eqref{g903nje} below. 

  Abbreviate $\bm{M} = \bm{M}(N_n+1)$ with $ij$th entry $M_{ij}$ and $\bm{D} = \bm{D}(N_n+1)$ with $ij$th entry $D_{ij}$. Let $\N_{\bm{J}^\star(N_n+1; (x,z))}^\partial(1,\ell,\ell',\ell'')$ be the set of nodes $k$ for which there exist nodes $i,j$ such that the path distance between 1 and $i$ in $\bm{M}$ is $\ell$, the path distance between $i$ and $j$ in $\bm{D}$ is $\ell'$, and the path distance between $j$ and $k$ in $\bm{M}$ is $\ell''$. Then
  \begin{align}
    \prob\bigg( \max_{i\in \N_{\bm{J}^\star(N_n+1; (x,z))}(1,K,w')} 
    &r_n^{-1}\norm{x-X_i} > w \bigg) \nonumber\\ 
    &\leq \sum_{\ell=1}^K \sum_{\ell''=0}^1 \prob\bigg( \max_{i\in \N_{\bm{J}^\star(N_n+1; (x,z))}^\partial(1,\ell,0,\ell'')} r_n^{-1}\norm{x-X_i} > w \bigg) \nonumber\\
    &+ \sum_{\ell'=0}^{w'} \sum_{\ell''=0}^1 \prob\bigg( \max_{i\in \N_{\bm{J}^\star(N_n+1; (x,z))}^\partial(1,0,\ell',\ell'')} r_n^{-1}\norm{x-X_i} > w \bigg) \nonumber\\
    &+ \sum_{\ell=1}^K \sum_{\ell'=1}^{w'} \prob\bigg( \max_{i\in \N_{\bm{J}^\star(N_n+1; (x,z))}^\partial(1,\ell,\ell',0)} r_n^{-1}\norm{x-X_i} > w \bigg) \nonumber\\
    &+ \sum_{\ell=1}^K \sum_{\ell'=1}^{w'} \prob\bigg( \max_{i\in \N_{\bm{J}^\star(N_n+1; (x,z))}^\partial(1,\ell,\ell',1)} r_n^{-1}\norm{x-X_i} > w \bigg). \label{t23b0e9}
  \end{align}

  We begin by bounding the fourth term on the right-hand side. Under the event that $\max_{i\in \N_{\bm{J}^\star(N_n+1; (x,z))}^\partial(1,\ell,\ell',1)} r_n^{-1}\norm{x-X_i} > w$, there must exist distinct nodes $j_0 \neq j_1 \neq \dots \neq j_{\ell+\ell'+1}$ with $j_0 = 1$ such that (a) for some $m=1,\dots,\ell+\ell'+1$, we have $r_n^{-1}\norm{X_{j_{m-1}} - X_{j_m}} > w/(\ell+\ell')$ where $X_1\equiv x$; (b) for all $k=1,\dots,\ell$, we have $M_{j_{k-1}j_k}=1$; (c) for all $k'=\ell+1,\dots,\ell+\ell'$, we have $D_{j_{k'-1}j_{k'}}=1$; and (d) $M_{j_{\ell+\ell'}j_{\ell+\ell'+1}}=1$. 

  By \autoref{Vc}, these facts imply $\zeta_{j_{m-1}j_m} > \bar{V}^{-1}(w/(\ell+\ell'),0)$ for some $m=1,\dots,\ell+\ell'+1$, so 
  \begin{multline}
    \prob\bigg( \max_{i\in \N_{\bm{J}^\star(N_n+1; (x,z))}^\partial(1,\ell,\ell',1)} r_n^{-1}\norm{x-X_i} > w \bigg) \\ 
    \leq \prob\bigg( \bigcup_{j_0 \neq \dots \neq j_{\ell+\ell'+1}  \in \N_{N_n+1}, j_0=1} \medcup_{m=1}^{\ell+\ell'+1} \big\{ \zeta_{j_{m-1}j_m} > \bar{V}^{-1}(w/(\ell+\ell'),0) \big\} \\ \medcap \bigg\{ \prod_{k=1}^\ell M_{j_{k-1}j_{k}} \cdot \prod_{k'=\ell+1}^{\ell+\ell'} D_{j_{k'-1}j_{k'}} \cdot M_{j_{\ell+\ell'}j_{\ell+\ell'+1}} = 1 \bigg\} \bigg). \label{gj90qw}
  \end{multline}

  \noindent Let $X_1\equiv x$, $\bar{p}_{r_n}(\cdot)$ be defined as in \eqref{barp}, $p_{r_n}(\cdot)$ be defined as in \eqref{L_r}, and for $x,x'\in\R^d$,
  \begin{equation*}
    \rho_n(x,x',w/(\ell+\ell')) \equiv \prob\big( \zeta_{j_{m-1}j_m} > \max\{\bar{V}^{-1}(w/(\ell+\ell'),0), \bar{V}^{-1}(r_n^{-1}\norm{x-x'},0)\}\big).
  \end{equation*}

  \noindent Using the law of iterated expectations conditioning on all node types, 
  \begin{align}
    \eqref{gj90qw} \leq \E\bigg[ &\sum_{j_0 \neq \dots \neq j_{\ell+\ell'+1} \in \N_{N_n+1}, j_0=1} \sum_{m=1}^{\ell+\ell'+1} \rho_n(X_{j_{m-1}},X_{j_m},w/(\ell+\ell')) \nonumber\\ 
	     &\times  \prod_{\substack{k=1 \\ k\neq m}}^\ell \bar{p}_{r_n}(X_{j_{k-1}},X_{j_k}) \cdot \prod_{\substack{k'=\ell+1 \\ k'\neq m}}^{\ell+\ell'} p_{r_n}(X_{j_{k'-1}},Z_{j_{k'-1}},X_{j_{k'}},Z_{j_{k'}}) \nonumber\\
	     &\times \big( \ind\{m=\ell+\ell'+1\} + \ind\{m\neq\ell+\ell'+1\}\bar{p}_{r_n}(X_{j_{\ell+\ell'}},X_{\ell+\ell'+1}) \big) \bigg]. \label{g903nje}
  \end{align}

  \noindent By Corollary 3.2.3 of \cite{schneider2008stochastic}, 
  \begin{align*}
    \eqref{g903nje} 
    &\leq (\ell+\ell'+1) \sup_{x_{j_0}} \underbrace{\int_{\R^d} \cdots \int_{\R^d}}_{\ell \text{ times}} \prod_{k=1}^\ell \bar{p}_{r_n}(x_{j_{k-1}},x_{j_k})\, nf(x_{j_k}) \,\text{d}x_{j_k} \\ 
    &\times \int_{\R^d} \rho_n\big(x_{j_{\ell}},x_{j_\ell+1},w/(\ell+\ell')\big) \, nf(x_{\ell+1}) \,\text{d}x_{\ell+1} \\
    &\times \underbrace{\int_{\R^d} \cdots \int_{\R^d}}_{\ell'-2 \text{ times}} \tilde\E\bigg[ \prod_{k'=\ell+2}^{\ell+\ell'} p_{r_n}(x_{j_{k'-1}},Z_{j_{k'-1}},x_{j_{k'}},Z_{j_{k'}}) \bigg] \prod_{k'=\ell+2}^{\ell+\ell'} nf(x_{j_{k'}}) \,\text{d}x_{j_{k'}} \\ 
    &\times \int_{\R^d} \bar{p}_{r_n}(x_{j_{\ell+\ell'}},x_{j_{\ell+\ell'+1}}) \, nf(x_{\ell+\ell'+1}) \,\text{d}x_{\ell+\ell'+1},
  \end{align*}
  
  \noindent where we abbreviate $\tilde\E[\cdot] = \E[\cdot \mid X_{j_\ell+2}=x_{j_\ell+2}, \dots, X_{j_{\ell+\ell'}}=x_{j_{\ell+\ell'}}]$. By repeatedly applying the Cauchy-Schwarz inequality and \autoref{dfrag}(a), we can bound this by
  \begin{equation*}
    (\ell+\ell'+1) \left( \sup_x n \int_{\R^d} \bar{p}_{r_n}(x,y) f(y) \,\text{d}y \right)^{\ell+1} \norm{h^*}_{\mathbf{m}}^{\ell'-2} \sup_x n\,\E\big[\rho_n\big(x,X_2,w/(\ell+\ell')\big)\big].
  \end{equation*}

  \noindent Note that $\norm{h^*}_{\mathbf{m}}<1$ by the assumption. Using \autoref{Vc}, for sufficiently large $w$ there exist $c_1>0$ and $c_2>1$ such that this is bounded above by
  \begin{multline}
    (\ell+\ell'+1) \cdot c_1^{\ell+1} \cdot c_2^{-\ell'+2} \cdot \kappa\bar{f} \bigg( \int_{w/(\ell+\ell')}^\infty u^{d-1} \bar{\Phi}_\zeta(\bar{V}^{-1}(u,0)) \text{d}u \\ + \frac{(w/(\ell+\ell'))^d}{d} \bar{\Phi}_\zeta(\bar{V}^{-1}(w/(\ell+\ell'),0)) \bigg) \label{gh9024jhjergs}
  \end{multline}

  \noindent by a change of variables to polar coordinates. Since $\bar{\Phi}_\zeta(\bar{V}^{-1}(u,0)) \leq c_3 e^{-c_4 u}$ for some $c_3,c_4>0$ by \autoref{Vc}, the second term involving $\bar{\Phi}_\zeta$ is at least as large as the first term for sufficiently large $w$, as can be seen from integration by parts. Then 
  \begin{multline*}
    \eqref{gh9024jhjergs} \leq (\ell+\ell'+1) \cdot c_1^{\ell+1} \cdot c_2^{-\ell'+2} \cdot \kappa\bar{f} (w/(\ell+\ell'))^d c_3 \,\text{exp}\{-c_4 w/(\ell+\ell')\} \\
    < C (\ell+\ell') (w/(\ell+\ell'))^d \,\text{exp}\{-c (w/(\ell+\ell') - \ell + \ell')\}
  \end{multline*}

  \noindent for $w$ sufficiently large and some $C,c>0$ that do not depend on $w,\ell,\ell'$. 

  Therefore, for $w$ large enough, the fourth term on the right-hand side of \eqref{t23b0e9} is bounded above by
  \begin{multline*}
    C \sum_{\ell=1}^K \sum_{\ell'=1}^{w'} (\ell+\ell') (w/(\ell+\ell'))^d \,\text{exp}\{-c (w/(\ell+\ell') - \ell + \ell')\} \\
    < CK(w'+K) w^d \sum_{\ell'=1}^{w'} \text{exp}\{-c (w/(K+\ell') + \ell' - K)\} \\
    \leq C'w'w^d \sum_{\ell'=1}^{w'} \text{exp}\{-c (w/(K+\ell') + \ell')\}. 
  \end{multline*}

  \noindent for some $C'>0$ that does not depend on $w,w'$. Since $w/(K+\ell') + \ell'$ is optimized at $\ell' = \sqrt{w} - K$, for $w$ large, the right-hand side is bounded above by
  \begin{equation*}
    C'w'w^d \cdot w' \text{exp}\left\{ -c \left( 2\sqrt{w} - K \right) \right\} = O(w^{d+2} \text{exp}\{-c\,w^{1/2}\}).
  \end{equation*}

  \noindent By a similar argument, the other terms of \eqref{t23b0e9} are also of this order.
\end{proof}

The main result of this subsection is the following tail bound for a quantity that upper bounds the radius of stabilization. Recall the definitions of $r_n^{-1}\bm{T}_{m+k}^\star(\bm{t}_k)$ prior to \autoref{main-bp-exp-stab}, $\tilde{J}_1^\star(m+k; \bm{t}_k)$ prior to \autoref{Jtails}, and the objects prior to \autoref{MKN}. For each $i\in\mathcal{N}_{m+k}$, let $X_i(\bm{t}_k)$ denote node $i$'s position according to the type array $r_n^{-1}\bm{T}_{m+k}^\star(\bm{t}_k)$. Then $\prob( \max_{i\in \tilde{J}_1^\star(m+k; \bm{t}_k)} r_n^{-1}\norm{X_1(\bm{t}_k)-X_i(\bm{t}_k)} > w )$ is a version of the conditional probability $\prob( \max_{i\in \tilde{J}_1(m+k)} r_n^{-1}\norm{X_1-X_i} > w \mid \bm{T}_k=\bm{t}_k)$ for $\tilde{J}_i$ defined in \eqref{Ji}.

\begin{seclemma}\label{dyDC}
  Under Assumptions \ref{Vc}, \ref{dfrag}, and \ref{dreg}, for $n,m\in\mathbb{N}$ sufficiently large, $k \in \{1,\dots,4\}$, and $k'\in\{1,2\}$, there exist $n_0,\epsilon>0$ such that 
  \begin{align*}
    &\limsup_{w\rightarrow\infty} w^{-1/2} \max\{ \log \tau_{b,\epsilon}^*(w), \log \tau_p^*(w)\} < 0, \quad\text{where} \\
    \tau_{b,\epsilon}^*(w) &= \sup_{n>n_0} \sup_{m \in ((1-\epsilon)n, (1+\epsilon)n)} \sup_{\bm{t}_k \in \mathcal{T}^k} \prob\big( \max_{i\in \tilde{J}_1^\star(m+k; \bm{t}_k)} r_n^{-1}\norm{X_1(\bm{t}_k)-X_i(\bm{t}_k)} > w \big), \\
    \tau_p^*(w) &= \sup_{n>n_0} \sup_{\bm{t}_{k'} \in \mathcal{T}^{k'}} \prob\big( \max_{i\in \tilde{J}_1^\star(N_n+k'; \bm{t}_{k'})} r_n^{-1}\norm{X_1(\bm{t}_{k'})-X_i(\bm{t}_{k'})} > w \big).
  \end{align*}
\end{seclemma}
\begin{proof}
  We prove the result for $\tilde{J}_1^\star(m+k; \bm{t}_k)$ with $k=3$. The argument is very similar for the other cases. Let $w'=\alpha w$ for some $\alpha>0$. By the law of total probability,
  \begin{multline*}
    \prob\big( \max_{i\in \tilde{J}_1^\star(m+k; \bm{t}_k)} r_n^{-1}\norm{X_1(\bm{t}_k)-X_i(\bm{t}_k)} > w \big) \leq \prob\big( \abs{\tilde{J}_1^\star(m+k; \bm{t}_k)} > w' \big) \\ +  \prob\big( \max_{i\in \tilde{J}_1^\star(m+k; \bm{t}_k)} r_n^{-1}\norm{X_1(\bm{t}_k)-X_i(\bm{t}_k)} > w \medcap \abs{\tilde{J}_1^\star(m+k; \bm{t}_k)} \leq w' \big). 
  \end{multline*}

  \noindent The first term on the right-hand side decays exponentially in $w$ uniformly over $n,m$ by \autoref{Jtails}.
  
  For the second term, by \eqref{ogic},
  \begin{align*}
    \big\{ \max_{i\in \tilde{J}_1^\star(m+k; \bm{t}_k)} r_n^{-1}\norm{X_1(\bm{t}_k)-X_i(\bm{t}_k)} > w \big\} &\subseteq \big\{ \max_{i\in Q_1(\{2,3\})} r_n^{-1}\norm{X_1(\bm{t}_k)-X_i(\bm{t}_k)} > w \big\} \\ \cup \big\{ \max_{i\in Q_2(\{1,3\})} r_n^{-1}\norm{X_1(\bm{t}_k)-X_i(\bm{t}_k)} > w \big\} &\cup \big\{ \max_{i\in Q_3(\{1,2\})} r_n^{-1}\norm{X_1(\bm{t}_k)-X_i(\bm{t}_k)} > w \big\}.
  \end{align*}

  \noindent We proceed to bound
  \begin{equation*}
    \prob\big( \max_{i\in Q_1(\{2,3\})} r_n^{-1}\norm{X_1(\bm{t}_k)-X_i(\bm{t}_k)} > w \medcap \abs{Q_1(\{2,3\})} \leq w' \big),
  \end{equation*}

  \noindent noting that the argument is identical for analogous probabilities involving the other $Q$ sets. Let $\bm{t}_1=(x,z)$. Given $k=3$, the previous expression equals
  \begin{equation}
    \prob\big( \max_{i\in \tilde{J}_1^\star(m+1; (x,z))} r_n^{-1}\norm{x-X_i} > w \medcap \abs{\tilde{J}_1^\star(m+1; (x,z))} \leq w' \big). \label{p3p9820yu}
  \end{equation}

  If $\abs{\tilde{J}_1^\star(m+1; (x,z))} \leq w'$, then for all $i \in \N_{\bm{M}^\star(m+1; (x,z))}(1,K)$, we have $\tilde{C}_i^\star(m+1; (x,z)) \subseteq \medcup\{\N_{\bm{M}^\star(m+1; (x,z))}(j,1)\colon j \in \N_{\bm{D}^\star(m+1; (x,z))}(i,w')\}$, where these objects are defined prior to \autoref{MKN}. Recalling \eqref{NJ}, this implies $\tilde{J}_1^\star(m+1; (x,z)) \subseteq \N_{\bm{J}^\star(m+1; (x,z))}(1,K,w')$. Therefore, 
  \begin{equation}
    \eqref{p3p9820yu} \leq \delta(w) \equiv \prob\big( \max_{i\in \N_{\bm{J}^\star(m+1; (x,z))}(1,K,w')} r_n^{-1}\norm{x-X_i} > w \big). \label{g390en3}
  \end{equation}
  
  \noindent By \autoref{MKN} and the choice of $w'$, $\delta(w) = O(w^{d+2}\text{exp}\{-c\,w^{1/2}\})$ for some $c>0$, so $\limsup_{w\rightarrow\infty} w^{-1/2} \log \delta(w) < 0$.
\end{proof}

\subsection{Proof of \autoref{primclt}}

\autoref{constructJT} constructs for each node $i$, a set $J_i$ satisfying \eqref{Jigoal}, which is contained in $\tilde{J}_i$ defined in \eqref{Ji}. Using the branching process results in \autoref{introbp} and \autoref{sexpstab}, \autoref{dyDC} derives a tail bound on the distribution of $\max_{i \in \tilde{J}_1} r_n^{-1}\norm{X_1-X_i}$. This plus one is an upper bound on the radius of stabilization by \eqref{ros}, which proves the result. \qed


\end{document}